\documentclass[11pt]{article}
\usepackage[margin=1in]{geometry}
\usepackage{fullpage}
\usepackage{tgtermes}
\usepackage[T1]{fontenc}
\usepackage[colorlinks,citecolor=blue,linkcolor=blue,urlcolor=black]{hyperref}
\usepackage{graphicx}
\usepackage{mathtools}
\usepackage{amsfonts,amsmath,amsthm,amssymb,dsfont}
\usepackage{xfrac,nicefrac}
\usepackage{mathdots}
\usepackage{bm,bbm}
\usepackage{url}
\usepackage{paralist}
\usepackage{enumerate}
\usepackage[normalem]{ulem}
\usepackage{xspace}
\xspaceaddexceptions{]\}}
\usepackage[capitalize]{cleveref}
\usepackage{comment}
\usepackage{tabu}
\usepackage{framed}
\usepackage{float,wrapfig}
\usepackage{tikz}
\usepackage[usenames,dvipsnames]{pstricks} 

\usepackage{enumitem}
\usepackage{adjustbox}
\usepackage{algorithm}
\usepackage[noend]{algpseudocode}

\usepackage{thmtools}
\usepackage{thm-restate}
\usepackage{tikzsymbols}

\usepackage{caption}
\usepackage{subcaption}

\usepackage{regexpatch}
\makeatletter
\xpatchcmd\thmt@restatable{%
\csname #2\@xa\endcsname\ifx\@nx#1\@nx\else[{#1}]\fi
}{%
\ifthmt@thisistheone
\csname #2\@xa\endcsname\ifx\@nx#1\@nx\else[{#1}]\fi
\else
\csname #2\@xa\endcsname[{Restated}]
\fi}{}{}
\makeatother

\algnewcommand{\algorithmicand}{\textbf{AND }}
\algnewcommand{\algorithmicor}{\textbf{OR }}
\algnewcommand{\algorithmicxor}{\textbf{XOR }}
\algnewcommand{\algorithmicnot}{\textbf{NOT }}
\algnewcommand{\OR}{\algorithmicor}
\algnewcommand{\AND}{\algorithmicand}
\algnewcommand{\XOR}{\algorithmicxor}
\algnewcommand{\NOT}{\algorithmicnot}
\algnewcommand{\var}{\texttt}

\algnewcommand{\algorithmicbreak}{\textbf{break}}
\algnewcommand{\Break}{\algorithmicbreak}

\usetikzlibrary{hobby}
\usetikzlibrary{decorations.markings}
\usetikzlibrary{arrows.meta, positioning}
\usetikzlibrary{quotes,angles}
\usepackage{xcolor}

\DeclarePairedDelimiter\abs{\lvert}{\rvert}%
\DeclarePairedDelimiter\norm{\lVert}{\rVert}%

\makeatletter
\let\oldabs\abs
\def\abs{\@ifstar{\oldabs}{\oldabs*}}
\let\oldnorm\norm
\def\norm{\@ifstar{\oldnorm}{\oldnorm*}}
\makeatother

\theoremstyle{plain}

\newtheorem{theorem}{Theorem}[section]
\newtheorem{lemma}[theorem]{Lemma}
\newtheorem*{lemma*}{Lemma}

\newtheorem{claim}[theorem]{Claim}
\newtheorem*{claim*}{Claim}
\newtheoremstyle{named}{}{}{\itshape}{}{\bfseries}{.}{.5em}{#3}
\theoremstyle{named}

\theoremstyle{definition}
\newtheorem{definition}[theorem]{Definition}
\newtheorem{observation}[theorem]{Observation}

\def\ShowAuthNotes{1}
\ifnum\ShowAuthNotes=1
\newcommand{\authnote}[2]{\ \\ \textcolor{red}{\parbox{0.9\linewidth}{[{\footnotesize {\bf #1:} { {#2}}}]}}\newline}
\else
\newcommand{\authnote}[2]{}
\fi
\newcommand{\xiao}[1]{}
\newcommand{\Aviad}[1]{}

\renewcommand{\epsilon}{\varepsilon}
\newcommand{\eps}{\varepsilon}
\renewcommand{\Pr}{\operatorname*{\mathbf{Pr}}}
\newcommand{\Ex}{\operatorname*{\mathbf{E}}}

\newcommand{\poly}{\operatorname{\mathrm{poly}}}
\newcommand{\polylog}{\poly\log}
\newcommand{\mA}{\mathcal{A}}

\newcommand{\ver}{{V}\xspace}
\newcommand{\e}{{E}\xspace}
\renewcommand{\tilde}{\widetilde}
\renewcommand{\hat}{\widehat}

\newcommand{\floor}[1]{\left\lfloor #1 \right\rfloor}
\newcommand{\ceil}[1]{\left\lceil #1 \right\rceil}
\newcommand{\round}[1]{\left\lfloor #1 \right\rceil}

\DeclareMathOperator{\md}{TD}
\DeclareMathOperator{\rd}{SD}
\DeclareMathOperator{\treemd}{\Wintertree TD}

\DeclareMathOperator{\treerd}{\Wintertree SD}

\DeclareMathOperator{\curveAviad}{Curv}
\newcommand{\vertex}[2]{\left\langle #1, #2 \right\rangle}
\makeatletter
\newsavebox{\@brx}
\newcommand{\llangle}[1][]{\savebox{\@brx}{\(\m@th{#1\langle}\)}%
  \mathopen{\copy\@brx\kern-0.5\wd\@brx\usebox{\@brx}}}
\newcommand{\rrangle}[1][]{\savebox{\@brx}{\(\m@th{#1\rangle}\)}%
  \mathclose{\copy\@brx\kern-0.5\wd\@brx\usebox{\@brx}}}
\makeatother
\newcommand{\vertexoriginal}[2]{\llangle #1, #2 \rrangle}

\DeclareMathOperator*{\sign}{sign}
\DeclareMathOperator*{\appr}{appr}
\DeclareMathOperator*{\apprew}{apprew}
\DeclareMathOperator*{\ed}{ED}
\DeclareMathOperator*{\lcs}{LCS}

\DeclareMathOperator{\ans}{ans}

\DeclareMathOperator{\est}{\hat{\ans}}
\DeclareMathOperator{\whest}{est}
\DeclareMathOperator{\estnaive}{bf}
\DeclareMathOperator{\ew}{ew}

\newcommand{\NC}{{\sf NC}\xspace}

\newcommand{\NTIME}{{\sf NTIME}\xspace}

\newcommand{\ESTSUM}{\textsc{EstSum}}

\newcommand{\COMESTED}{\textsc{ComputeEstimateED}}
\newcommand{\COMESTLCS}{\textsc{ComputeEstimateLCS}}

\newcommand{\mval}{0.109}

\definecolor{brightblue}{rgb}{0.1, 0.9, 1.0}   
\definecolor{deepred}{rgb}{0.6, 0.0, 0.0}      

\definecolor{tempred}{RGB}{192,0,0}
\definecolor{tempgreen}{RGB}{56,87,35}

\usepackage{fontawesome5}
\newcommand{\anchor}{\text{\faAnchor}}
\newcommand{\anc}{m}
	
\title{Approximations schemes for Edit Distance and LCS in quasi-strongly subquadratic time}
\author{  Xiao Mao\footnote{Supported by NSF CAREER Award CCF-2337901.} \\  Stanford University \\  \texttt{xiaomao@stanford.edu} \and Aviad Rubinstein\footnote{Supported by NSF CAREER Award CCF-2337901.}  \\  Stanford University \\  \texttt{aviad@stanford.edu}}

\begin{document}

	\setcounter{page}{0} \clearpage
	\maketitle
	\thispagestyle{empty}
    \begin{abstract}
We present novel randomized approximation schemes for the Edit Distance (ED) problem and the Longest Common Subsequence (LCS) problem that, for any constant $\epsilon>0$, compute a $(1+\epsilon)$-approximation for ED and a $(1-\epsilon)$-approximation for LCS in time $n^2 / 2^{\log^{\Omega(1)}(n)}$ for two strings of total length at most $n$. This running time improves upon the classical quadratic-time dynamic programming algorithms by a quasi-polynomial factor. 

Our results yield significant insights into fine-grained complexity: Firstly, for ED, prior work indicates that any exact algorithm cannot be improved beyond a few logarithmic factors without refuting established complexity assumptions [Abboud, Hansen, Vassilevska Williams, Williams, 2016]; our quasi-polynomial speed-up shows a separation the complexity of approximate ED from that of exact ED, even for approximation factor arbitrarily close to $1$. Secondly, for LCS, obtaining similar approximation-time tradeoffs via deterministic algorithms would imply breakthrough circuit lower bounds [Chen, Goldwasser, Lyu, Rothblum, Rubinstein, 2019]; our randomized algorithm demonstrates derandomization hardness for LCS approximation.\\\\

\end{abstract}

\newpage
\section{Introduction}
The {\em edit distance (ED)} between two strings $X$ and $Y$, henceforth denoted $\ed(X, Y)$, is the minimum number of character insertions, deletions, and substitutions required to transform $X$ to $Y$. The complexity of edit distance and its maximization counterpart\footnote{Technically, LCS is the complement of a variant of ED where substitutions are not allowed (or rather counted as 1 deletion + 1 insertion).}, {\em longest common subsequence (LCS)}, 
is a long-standing, central problem in algorithm design, with important applications, especially in bioinformatics. The simple quadratic-time dynamic programming solution is core material for undergraduate algorithms; this running time has been improved by one or two log factors (depending on the exact access model)~\cite{MP80,Grabowski16}. Despite extensive research for over half a century%
\footnote{Beating the quadratic time was proposed as an open problem as early as~\cite{knuth72}.}, no better algorithms for exact ED or LCS are known.

Indeed, in one of the flagship successes of fine-grained complexity,~\cite{BI18} showed that there are no exact algorithms for Edit Distance running in ``truly subquadratic'' time, namely it is impossible to improve by $n^{\delta}$ for any constant $\delta>0$, assuming fine-grained complexity hypothesis SETH (see also e.g., \cite{AWW14-Local_Alignment,ABW15-LCS,BK15-LCS}). Furthermore~\cite{AHWW16-polylog_shaved,AbboudB18} showed that improving even by a few log factors would refute other fine-grained complexity assumptions and imply new circuit complexity lower bounds.

Given the limitations of exact algorithms, and especially in light of fine-grained hardness results, approximation algorithms for Edit Distance have been extensively studied: A sequence of earlier works gave super-constant-factor approximation factors in $n^{1+o(1)}$ time~\cite{ADGIR03-edit, BEKMRRS03-edit-testing, BJKK04-edit, BES06-edit, OR07-edit,  AK10-edit-CC, AN10-edit-query, AKO10-edit, AO12-edit}. More recently, constant-factor approximation algorithms times were obtained by~\cite{CDGKS20,BEGHS21, Andoni20, BR20, GRS20, KS20,AN20}. The current state of the art has several algorithms with different trade-offs between running time and approximation (see \cref{table:tradeoffs}).  
While we cannot cover all the techniques used in previous work on approximate edit distance, at a high level we seem to have hit a fundamental obstacle: all known constant-factor approximation algorithms for ED rely on a triangle inequality technique that inherently incurs at least a 3-factor in the approximation~\cite{Rub18-blog}. No better-than-3-approximation algorithms are known with any non-trivial running time.

Remarkably, good multiplicative approximations for LCS are even more scarce than ED: despite significant attention (e.g., \cite{HSSS19,RS20,AkmalW21,RSSS23,BCD23,HeL23}), for worst-case instances the best previous known approximation factor is just barely sub-polynomial~\cite{Nosatzki21}. Furthermore, \cite{CGLRR19} shows that obtaining such approximation-time tradeoffs for approximate LCS with {\em deterministic} algorithms would imply new breakthrough circuit lower bounds%
\footnote{Specifically, Item 3 of Theorem 1.9 of~\cite{CGLRR19} says that if a deterministic algorithm for LCS can obtain even a much weaker $O(\polylog(N))$-approximation guarantee  in $N^2 / 2^{\omega(\log \log N)^3}$ time, then $\NTIME[2^{O(n)}]$ is not contained in non-uniform $\NC^1$.}. 

\paragraph{Main result}
We give approximation schemes for ED and LCS, i.e., for any constant $\eps>0$, we get a $(1+\eps)$-approximation for ED and a $(1-\eps)$-approximation for LCS. Our approximation schemes run in time 
\[
\frac{n^2}{2^{\log^{\Omega(1)}(n)}},
\]
for two strings with total length $n$. This running time represents a quasi-polynomial improvement over the classical quadratic-time algorithms.

\paragraph{Why this is interesting}
\begin{itemize}
    \item {\bf Algorithmic improvement:} The most obvious achievement of our main result is a quasi-polynomial improvement in running time. While the quantitative improvement is modest, in the field of fine-grained complexity quasi-polynomial~\cite{Williams18, LW17, AFK+24} and even smaller improvements (e.g., \cite{gawrychowski2019minimum, JX24, NewSSSP}) are often celebrated as significant steps towards pinpointing the complexity of fundamental problems (and approximate ED and LCS are certainly fundamental!).  More importantly, perhaps, this improvement shows the algorithmic potential of our new structural insights (see discussion below on ``tree total deviation''). 
    \item \textbf{Separation between $(1 + \eps)$-Approximate and Exact Computation (both ED and LCS):} For exact ED and LCS, prior work~\cite{AHWW16-polylog_shaved,AbboudB18} shows that even shaving off a few logarithmic factors is highly unlikely without refuting established fine-grained complexity assumptions. Although a sequence of approximation algorithms faster than this lower bound have recently been known~\cite{CDGKS20,BEGHS21, Andoni20, BR20, GRS20, KS20,AN20}, none of them achieve an approximation factor better than $3$. Therefore, subject to certain circuit complexity hardness assumptions (see~\cite{AHWW16-polylog_shaved,AbboudB18} for details), our result is the first to separate the complexity of $(1+\eps)$-approximation from exact computation of ED and LCS.
    \item \textbf{Separation Between Randomized and Deterministic Algorithms (LCS):} For LCS, obtaining similar approximation-time tradeoffs using \emph{deterministic} algorithms would imply breakthrough circuit lower bounds~\cite{CGLRR19} (see also~\cite{AB17-deterministic,AR18}). Indeed, many randomized approximation algorithms for both LCS and ED have been proposed, but it was not clear whether this randomization was necessary. Our algorithm is the first to break the frontier of deterministic lower bounds by using randomization, thereby demonstrating a fundamental separation: subject to certain circuit complexity hardness assumptions (see~\cite{CGLRR19}), randomized algorithms for approximate LCS are strictly faster than their deterministic counterparts.
\end{itemize}

\subsection*{Additional related Work}
In addition to aforementioned works on approximate ED and LCS on worst case inputs, many works also give improved guarantees under assumptions on the inputs such as small edit distance%
\footnote{The careful reader may notice that some of these citations do not explicitly mention ``small edit distance'' in the title, but they give fast algorithms for solving a ``gap edit distance'' problem of distinguishing e.g.~edit distance $k$ vs $k^{1.5}$; this kind of gap problem is only non-trivial when $k \ll n$.}~\cite{CGK16, GKS19,  KociumakaS20, BCR20, BCFN22b, BCFN22a, GKKS22, KS23, RSSS23}, semi-random strings~\cite{AK12,Kuszmaul19,BSS20}, or edit distance to error correcting codewords~\cite{HRS19}. See also e.g., \cite{KR09-edit, CGKK18, HSS19, BGS21, KPS21, CFHJLRSZ21, JNW21, BCFN22, KS24, BCFK24} for additional related works in other models. 

Ideas related to the breakthrough constant factor approximation algorithm for edit distance have also been applied to related problems, such as tree edit distance \cite{SS22} and, more recently, Frechet distance \cite{cheng2025constant}. For tree edit distance, we leave it for future work whether our techniques can also apply to improve \cite{BoroujeniGHS19}'s state of the art (nearly) quadratic running time for $(1+\eps)$-approximation. For the Frechet distance, existing hardness of approximation results rule out polynomial savings assuming SETH \cite{Bringmann14}, and it remains open whether quasi-polynomial saving is possible.

Our work gives a quasi-polynomial improvement in running times for central problems of interest in fine-grained complexity. Other notable examples of fine-grained complexity with quasi-polynomial improvements include All-Pairs Shortest Paths \cite{Williams18} (equivalent to the Min-plus matrix multiplication), online matrix-vector multiplication \cite{LW17}, and combinatorial boolean matrix multiplication \cite{AFK+24}.

\begin{table}
\centering
    \begin{tabular}{c|c|c}
    Approximation factor & Running time & Reference \\ \hline
    $\sqrt{n}$ & $O(n)$ & Folklore (building on~\cite{LMS98})\\
    $n^{1/3+o(1)}$ & $\tilde{O}(n)$
    & \cite{BES06-edit}\\
    $2^{O(\sqrt{\log(n)\log\log(n)})}$ & $n \cdot 2^{O(\sqrt{\log(n)\log\log(n)})}$  & \cite{AO12-edit}\\
    $1/\eps$ & $O(n^{1+\delta(\eps)})$ & \cite{AN20}\\
    $3+o(1)$ & $n^{1.6+o(1)}$ & \cite{GRS20}\\ \hline
    $1+o(1)$ & $n^2 / 2^{\log^{\Omega(1)}(n)}$ & {\bf this work}
    \end{tabular}
    \caption{State of the art trade-offs of approximation ratios vs running time for ED} \label{table:tradeoffs}
\end{table}




\section{High-level Technical Overview}\label{sec:techniques} 
First, note that if the actual ED or LCS is very small, say $\eps n$ for some $\eps = 2 ^ {-\log ^ {\Omega(1)}(n)}$, there are exact algorithms with running time $O(\eps n^2)$ for both problems. 
Since our goal is to obtain algorithms with a running time of $n ^ 2 / 2 ^ {\log ^ {\Omega(1)}(n)}$, we henceforth focus on the case where the ED or LCS is at least $\eps n$ for $\eps = 2 ^ {-\log ^ {O(1)}(n)}$.

Before we describe our main techniques for ED and LCS, we begin with a simple observation about sums of sequences and their sub-sequences (\cref{sub:sketchobs}). This observation plays a  crucial role in our algorithm, and analogous results have been applied to other areas of research, (see, for example, \cite{APred13}\footnote{The authors are grateful to Moses Charikar for bringing this connection to our attention.}, \cite{peng2025}). 
We then setup the notation with a standard reduction to optimal path problems on a 2D grid (\cref{sub:overview-grid}).
In \cref{sub:overview-additive} we explain the core algorithmic idea, which suffices to get an additive approximation scheme (i.e., approximate ED or LCS to within $\pm o(n)$). Finally in \cref{sub:overview-multiplicative} we extend those ideas to obtain multiplicative approximation schemes. 

\subsection{An Observation on Tree Total Deviation} \label{sub:sketchobs}


One of our core ideas is a simple but powerful observation about sums of sequences and their sub-sequences. Suppose that we have a length-$n$ binary sequence $A = (a_1, a_2, \ldots, a_{n})$ for some $n$ that is the power of some number $M$. We define the ($M$-way) {\em tree total deviation} as a measure of the total ``unevenness'' of $A$ across scales corresponding to blocks of size $M, M ^ 2, \ldots, M ^ {\log_M{n}}$. 

In more detail, to define the tree total deviation, let $S := \log_M{n}$ be the number of scales. For each scale $s \in [S]$, we go through all $n / M ^ {S - s + 1}$ intervals $(l, r]$ of form $\left(M ^ {S - s + 1} \cdot k, M ^ {S - s + 1} \cdot (k + 1)\right]$. Fixing some interval $(l,r]$, for each $0 \le i \le M - 1$, let $\sigma_i = \sum\limits_{l + iM ^ {S - s} < x \le l + (i + 1)M ^ {S - s}}{a_x}$. Namely, if we partition $a_{(l, r]}$ into $M$ equal parts, then $\sigma_i$ is the total sum of the $i$-th part. 
Let $\sigma = \sum_{l < x \le r}{a_x}$. The contribution from interval $(l, r]$ to the tree total deviation is equal to
$$\sum_i{\left|\sigma_i - \sigma / M\right|},$$
which is a measurement of how ``uneven'' the sequence $\{\sigma_0, \sigma_1, \ldots, \sigma_{M - 1}\}$ is.

Since $A$ is binary, it is easy to see that the local unevenness at an interval $(l, r]$ is $O(r - l)$. Thus the total unevenness on a particular scale is $O(n)$. Thus clearly the overall total unevenness (i.e., tree total deviation) is upper bounded by $O(n\log_M{n})$. 

We observe, however, that for a sequence to be maximally uneven on a particular scale $s$, the density of 1's has to be maximally different (i.e., close to either $0$ or $1$) at each block of width $M ^ {s - 1}$, and, intuitively, such density requirement makes it more difficult for the sequence to be maximally uneven on the remaining smaller scales. For the simplest example, suppose $n = 4$ and $M = 2$. We have two scales corresponding to blocks of size $2$ and $4$. The largest total deviation on the larger scale is achieved by $A = \{0, 0, 1, 1\}$ or $A = \{1, 1, 0, 0\}$, but for both sequences the total deviations on the smaller scale become zero. We prove that this is not a coincidence: 
\begin{quote}
(Ideal and informal version of \cref{lemma:totalmeandeviation})  For all sequences with each entry being $O(1)$ in absolute value, the overall total unevenness is roughly bounded by $n\sqrt{\log_M{n}}$, i.e., roughly $\sqrt{\log_M{n}}$ smaller than the na\"ive $O(n\log_M{n})$ upper bound. Moreover, if the total sum of the sequence is $\eps n$ and $\eps$ is at least $2 ^ {-\polylog(n)}$, then the total unevenness is roughly bounded by $\eps n \sqrt{\log_M{n}}$.
\end{quote}


\paragraph{Proof Sketch} We now give a sketch of the proof for $M = 2$, which contains all the techniques we need for proving the more general case but is much less tedious. The key idea is to decompose the total tree deviation into the contribution of each element, and then bound the contribution from an average element using a probabilistic argument. 

Towards the goal of decomposing the total tree deviation to individual elements' contributions, we first transform the sequence into another sequence with the same total tree deviation. Suppose we have a sequence $A$ with each entry being $O(1)$ with total sum $\eps n$. 
Let $S := \log_2{n}$ be the number of scales. For each scale $s \in [S]$, we go through all $n / 2 ^ {S - s + 1}$ intervals $(l, r]$ of form $\left(2 ^ {S - s + 1} \cdot k, 2 ^ {S - s + 1} \cdot (k + 1)\right]$. Fix a scale $s$. For every interval $(l,r]$ on scale $s$, we divide it into two halves $(l, (l + r) / 2]$ and $((l + r) / 2, r]$. We let $\sigma_L := \sum_{l < x \le (l + r) / 2}{a_x}$ be the sum on the left half, and $\sigma_R := \sum_{(l + r) / 2 < x \le r}{a_x}$ be the sum on the right half. In the case where $M = 2$, we can see that the contribution from the interval to the tree total deviation is
$$\left|\sigma_L - (\sigma_L + \sigma_R) / 2\right|+\left|\sigma_R - (\sigma_L + \sigma_R) / 2\right|=\left|\sigma_L - \sigma_R\right|.$$
In other words, the local unevenness is simply the difference between the sums of the two halves. Notice that if we swap the two halves to obtain a new sequence, the total tree deviation does not change: \begin{quote}
    $(l, r]$ is the only interval where sums of the two halves are swapped (which does not change their difference) --- for other intervals, the sums of the two halves remain identical.
\end{quote} We can apply this swap to every applicable interval. Therefore, without loss of generality, we can assume that $\sigma_L < \sigma_R$, and the contribution from the interval to the tree total deviation is
$$\sigma_R - \sigma_L = \left(\sum_{l < x \le (l + r) / 2}{-a_x}\right) + \left(\sum_{(l + r) / 2 < x \le r}{+a_x}\right).$$
Thus, for every $x \in (l, r]$, it contributes $-a_x$ if $x$ is in the first half, and $+a_x$ if $x$ is in the second half. Since $(l, r]$ is a scale-$s$ interval, the sign of $x$'s contribution depends only on the $s$-bit binary representation of $(x - 1)$: on scale-$s$, $x$ contributes $-a_x$ if $(x - 1)$'s $s$-th bit is $0$, and $+a_x$ if $(x - 1)$'s $s$-th bit is $1$. Summing up the total contribution from each index $x$ for each scale $s$, we have
$$\textrm{tree total deviation} = \sum_{1 \le x \le n}{a_x\left(\text{(number of 1's in $(x - 1)$ in binary) - (number of 0's in $(x - 1)$ in binary)}\right)}.$$
From well-known concentration bounds, for all but a small ``tail'' fraction of $x \in \{1,\dots, n\}$, the numbers of 0's and 1's in $(x - 1)$ in binary do not differ by too much more than $\sqrt{\log_2{n}}$. Therefore, if the total sum of the sequence is $\eps n$ and $\eps$ is not too small for the tail to become significant, 
the total unevenness is roughly bounded by $\eps n \sqrt{\log_2{n}}$.

\subsection{Setup and Notation (a Fairly Standard Reduction to Shortest/Longest Path on a 2D Grid)} \label{sub:overview-grid}
We reduce ED to a somewhat more general problem of finding the shortest path on a 2D grid\footnote{Do not confuse our 2D grid graphs with the very-closely-related but not identical {\em alignment graph}~\cite{CGMW21, AKW23, gorbachev2024bounded}, which is not rotated.} with diagonal short-cuts
(\cref{def:edgridoriginal} or \cref{fig:grid-basic}), where the weight of going from $\vertexoriginal{u}{v}$ to $\vertexoriginal{u}{v+1}$ or $\vertexoriginal{u+1}{v}$ is always $1$ (corresponding to insertion or deletion), and the weight of a {\em short-cut edge} from $\vertexoriginal{u}{v}$ to $\vertexoriginal{u+1}{v+1}$ is either $0$ (corresponding to a match) or $1$ (corresponding to a substitution). The goal is to find the length (sum of weights) of the shortest path from $\vertexoriginal{0}{0}$ to $\vertexoriginal{n}{n}$. 
This shortest path problem corresponds to the standard DP algorithm for ED. 

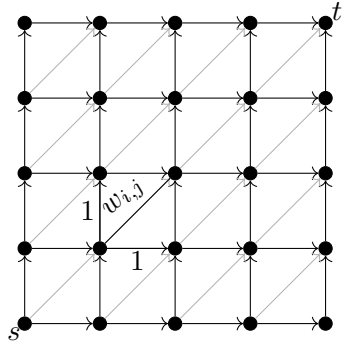
\begin{figure}[t]
\centering
\caption{Edit distance as a shortest path on 2D grid (standard)}

\hfill
\begin{minipage}[c]{0.48\textwidth}
\centering

\begin{tikzpicture}[scale=1, every node/.style={circle,draw,fill=black,inner sep=0pt,minimum size=5pt}]
\foreach \x in {1,...,5}{
\foreach \y in {1,...,5}{
\node (v\x\y) at (\x,\y) {};
}
}
\foreach \x in {1,...,4}{
\pgfmathtruncatemacro{\xx}{\x+1}
\foreach \y in {1,...,5}{
\draw[->] (v\x\y) -- (v\xx\y);
}
}
\foreach \x in {1,...,5}{
\foreach \y in {1,...,4}{
\pgfmathtruncatemacro{\yy}{\y+1}
\draw[->] (v\x\y) -- (v\x\yy);
}
}
\foreach \x in {1,...,4}{
\pgfmathtruncatemacro{\xx}{\x+1}
\foreach \y in {1,...,4}{
\pgfmathtruncatemacro{\yy}{\y+1}
\draw[->, gray!60] (v\x\y) -- (v\xx\yy);
}
}
\node[draw=none,fill=none,below left=2pt] at (v11) {$s$};
\node[draw=none,fill=none,above right=2pt] at (v55) {$t$};
\draw (v22) -- (v32) node[midway, below, draw=none,fill=none] {$1$};   
\draw (v22) -- (v23) node[midway, left,  draw=none,fill=none] {$1$};   
\draw (v22) -- (v33) node[midway, sloped, above, yshift=-3pt, draw=none,fill=none] {$w_{i,j}$}; 
\end{tikzpicture}

\end{minipage}
\begin{minipage}[c]{0.18\textwidth}
\centering
\begin{align*}
w_{i,j} &= \begin{cases}
    0 & X[i] = Y[j]\\
    1 & X[i] \neq Y[j].
\end{cases}
\end{align*}
\end{minipage}
\hfill
\label{fig:grid-basic}
\end{figure}

Similarly, we reduce LCS to a somewhat more general problem of finding the longest path on a 2D grid with diagonal short-cuts (\cref{def:lcsgridoriginal}), where the weight of going from $\vertexoriginal{u}{v}$ to $\vertexoriginal{u}{v+1}$ or $\vertexoriginal{u+1}{v}$ is always $0$, and the weight of a {\em short-cut edge} from $\vertexoriginal{u}{v}$ to $\vertexoriginal{u+1}{v+1}$ is either $0$ (no match) or $1$ (match). The goal is to find the length of the longest path from $\vertexoriginal{0}{0}$ to $\vertexoriginal{n}{n}$. 
This longest path problem corresponds to the standard DP algorithm for LCS.

We now use a 45\textdegree-rotated  $2^{-\mathrm{poly} \log n}$"system (see \cref{fig:coordinates}) to describe 2D grids with the transformation $\vertexoriginal{u}{v} \rightarrow \vertex{u + v}{v - u + n}$ (\cref{def:edgrid}). We use $\vertex{x}{y}$ when referring to the rotated coordinated system and refer to the x- and y-coordinates defined with respect to the first and second entries of $\vertex{x}{y}$.
In this notation, we are looking for the shortest/longest path from $\vertex{0}{n}$ to $\vertex{2n}{n}$. We use {\em column} to refer to a subset $\vertex{x}{\cdot}$ of the vertices, and {\em row} for $\vertex{\cdot}{y}$. It is easy to see that any $\vertex{0}{n}$-to-$\vertex{2n}{n}$ path visits exactly one vertex of each column.
A path $p$ in this notation can be given as a function $p:\{0,\dots.2n\}\rightarrow \{0,\dots.2n\}$ mapping each x-coordinate to the corresponding y-coordinate. %
Importantly, such a rotation transform ensures that adjacent vertices the path visits differ in their rows by at most 1 (i.e., $\abs{p(i) - p(i - 1)} \le 1$), which will be crucial to bounding the curvature of a path.

\begin{figure}[h!]
\caption{Convenient coordinate system}\label{fig:coordinates}
\begin{center}
\includegraphics[width=0.9\textwidth]{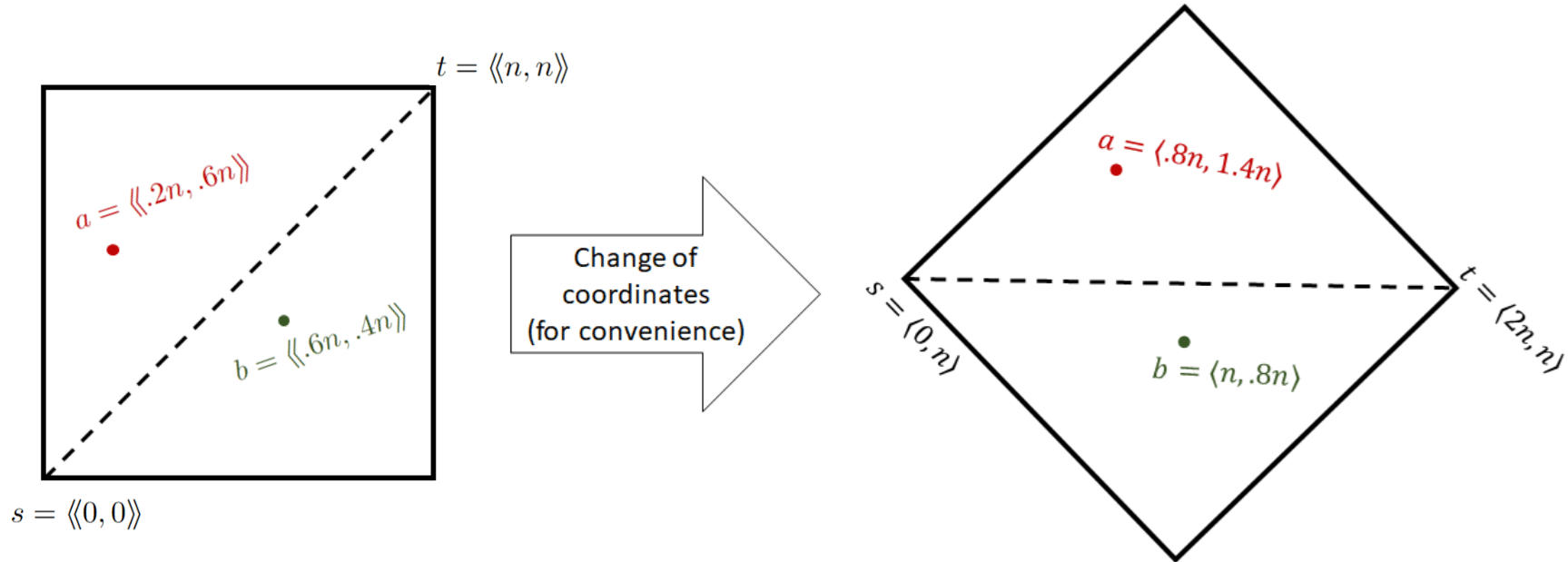}    
\end{center}
\end{figure}

\subsection{The Main Idea: Sample and Round} \label{sub:overview-additive}

To begin the discussion of our approaches, notice that it is possible to compute the shortest/longest path on the grid graphs exactly in a divide-and-conquer fashion: Suppose for the sub-problem on $(l, r]$, we want to know the boundary-to-boundary length of the shortest/longest path between $\vertex{l}{y_l}$ and $\vertex{r}{y_r}$ for every pair of rows $(y_l, y_r)$. For a branching factor $M$, we divide the sub-problem into $M$ equal parts $({\anc}_i, {\anc}_{i+1}]$ for ${\anc}_i = l+i\cdot\frac{r-l}{M}$ and solve them recursively%
\footnote{In the actual algorithm we only recurse until reaching a slightly sub-linear size problem (see ``rounding threshold for columns'' in Definitions~\ref{def:edgridoriginal} and~\ref{def:lcsgridoriginal})  and then use the standard DP in sub-quadratic time for smaller sub-problems.}. To find the optimal sub-path from $\vertex{l}{y_l}$ and $\vertex{r}{y_r}$, we look for the best \emph{anchors} $(p({\anc}_1), p({\anc}_2), \ldots, p({\anc}_{M - 1}))$ where $p({\anc}_i)$ is the row the path visits in column ${\anc}_i$ (i.e., the unique cell $\vertex{{\anc}_i}{p({\anc}_i)}$ that the sub-path visits). More formally, to compute the length of the shortest path%
\footnote{The computation for longest path follows analogously by replacing $\min$ with $\max$. We will highlight the difference between the setting when we discuss going from additive to multiplicative factor approximation.}
between $\vertex{l}{y_l}$ and $\vertex{r}{y_r}$, denoted by $\ans[l, y_l, r, y_r]$, we can use the following recursion:
\begin{equation} \label{eq:estsketch}
\ans[l, y_l, r, y_r] := \min\limits_{\hat{Y}}\sum\limits_{i = 1} ^ M{\ans[{\anc}_{i - 1}, \hat{y}_{i - 1}, {\anc}_i, \hat{y}_i]},
\end{equation}
where the minimization is over every anchor path defined by some sequence $\overline{Y} = (\overline{y}_i)$ that satisfies 
\begin{itemize}
    \item $\overline{y}_0 = y_l$ and $\overline{y}_M = y_r$;
    \item For $1 \le i \le M$, $\abs{\overline{y}_i - \overline{y}_{i - 1}} \le \frac{r - l}{M}$.
\end{itemize}

At a high level, the algorithm for constant additive error is very simple: we sub-sample the columns%
\footnote{For the sake of the overview, we slightly abuse notation and talk about discarding an interval of (consecutive) columns, by which we simply mean setting all the edge weights between columns in this interval to $0$.} of the (rotated) grid, discarding all but a $2^{-\log^{\Omega(1)}(n)}$-fraction (this is where we get our running time saving). Specifically, starting from $(0, 2n]$, we recursively partition the x-axis into equal-width intervals, with a branching factor $M := \polylog(n)$. Ideally, this gives rise to $S := \log_M(2n)$ {\em scales} in the divide-and-conquer recursion tree.
Each scale $s \in [S]$ contains all intervals of width $2n / M ^ {S - s}$ on the same level of the recursion tree. 
Each scale $s \in [S]$ is sampled i.i.d.~with probability $1/\log ^ c_M(n)$ for some constant $0 < c < 1$ into the set of {\em active scales} $\tilde{S}$. For each interval on an active scale, we randomly discard half of its sub-intervals, pretending that all edges weights in those sub-intervals are $0$. In total, we discard all but $2^{-|\tilde{S}|} = 2^{-\log^{\Omega(1)}(n)}$-fraction of the x-axis.

\begin{figure}[h!]\label{fig:subsampling}
\caption{Sub-sampling}
\begin{center}
\includegraphics[width=0.8\textwidth]{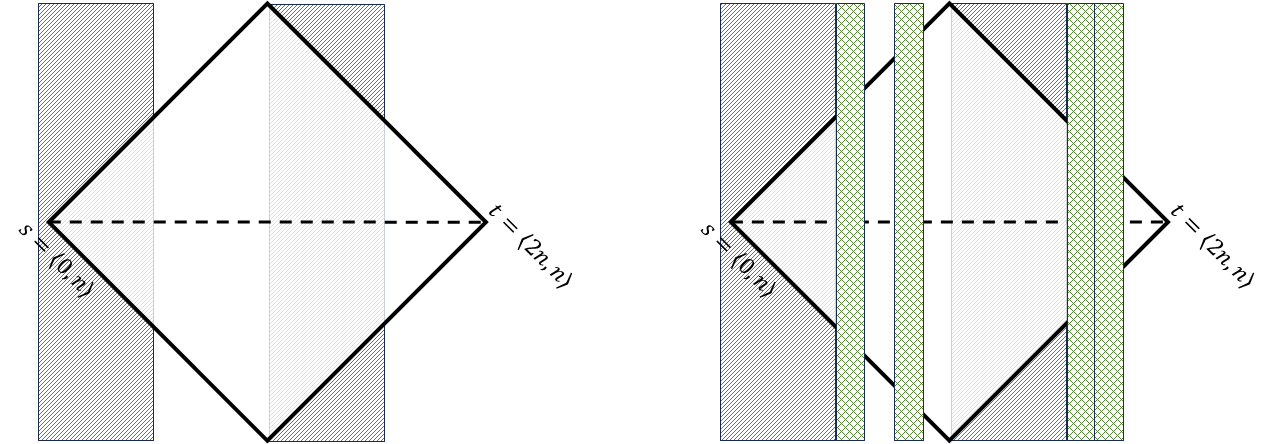}  \\    
\end{center}
{\small Illustration of sub-sampling columns with branching factor $M=4$. 
On the left, we have $\tilde{S} = \{S\}$, and {\color{blue} diagonal-shaded} intervals are discarded. On the right, we have $\tilde{S} = \{S - 1, S\}$, and {\color{olive} diamond-grid-shaded} intervals are also discarded.}
\end{figure}

\begin{figure}[htbp]
\caption{Overfitting Issue With Na\"ive Sub-sampling Approach for ED} \label{fig:diagrams}
\input{diagrams_.tex}
\end{figure}
Given a fixed good path, it is not hard to show that we can, with high probability, accurately estimate its true length by simply dividing its length in the surviving columns by the probability of survival, or more formally, we can change \eqref{eq:estsketch} to the following:
\begin{equation} \label{eq:estsketchnaive}
\ans[l, y_l, r, y_r] := 2\min\limits_{\hat{Y}}\sum\limits_{i = 1} ^ M{\eta_i\ans[{\anc}_{i - 1}, \hat{y}_{i - 1}, {\anc}_i, \hat{y}_i]},
\end{equation}
where each $\eta_i \in \{0, 1\}$ is uniformly sampled, and the $i$-th sub-interval is discarded if and only if $\eta_i = 0$.

However, the algorithm may hallucinate other paths that overfit the samples, as illustrated in \cref{fig:diagrams}. To handle this, we first observe that if there were only one candidate path, we would have obtained a good estimate. In general the algorithm should consider exponentially many candidate paths. But our divide-and-conquer algorithm considers the path one scale at a time, and we show that, at the granularity of a typical scale, it indeed suffices to consider just a single candidate path without losing too much precision.

The key observation driving our algorithm is that we can also restrict our attention to paths of bounded ``curvature.''
Specifically, consider an interval $(l,r]$ with sub-intervals $({\anc}_i, {\anc}_{i+1}]$ for ${\anc}_i = l+i\cdot\frac{r-l}{M}$. To reduce the possibility of hallucinating a path that overfits the samples, we reduce the number of candidate sequences of anchors $\left(p({\anc}_1), p({\anc}_2), \ldots, p({\anc}_{M - 1})\right)$ that the path $p$ can have to just one single ``straight-line'' sequence, ideally one where for every $i$ we have%
\footnote{Here and throughout, we use the standard notation of $\round{\cdot}$ for round-to-nearest-integer, with ties e.g.~rounded down.}
$p({\anc}_i) = \round{p(l) + i\frac{p(r) - p(l)}{M}}$, or more formally, we modify \eqref{eq:estsketchnaive} to the following:
\begin{equation*} 
\ans[l, y_l, r, y_r] := 2\sum\limits_{i = 1} ^ M{\eta_i\ans\left[{\anc}_{i - 1}, \round{p(l) + (i - 1)\frac{p(r) - p(l)}{M}}, {\anc}_i, \round{p(l) + i\frac{p(r) - p(l)}{M}}\right]}.
\end{equation*}
Namely, instead of going through all possible sequences of anchors, we force the anchors to resemble a straight line. 
An illustration of this candidate restriction is shown in \cref{fig:diagram_curves}.
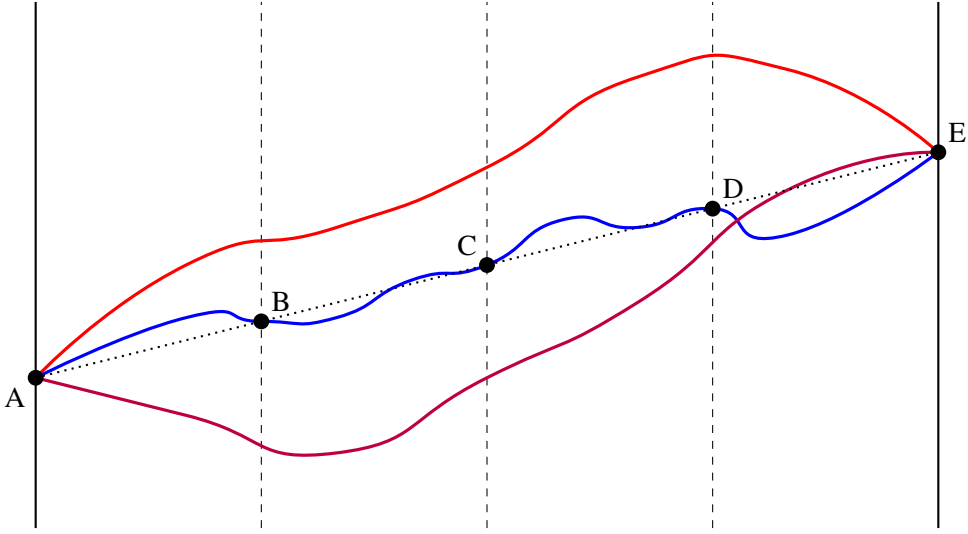
\begin{figure}[htbp]
\centering
\begin{tikzpicture}[scale=1]

  \draw[thick] (0,-2) -- (0,5);      
  \draw[dashed, thin] (3,-2) -- (3,5); 
  \draw[dashed, thin] (6,-2) -- (6,5); 
  \draw[dashed, thin] (9,-2) -- (9,5); 
  \draw[thick] (12,-2) -- (12,5);     
  
  \draw[red, very thick, domain=0:12, samples=100, smooth, tension=1] 
       plot coordinates { (0,0) (2, 1.5) (4,2) (6, 2.8) (8,4) (10, 4.1) (12,3) };
  
  \draw[blue, very thick, domain=0:12, samples=100, smooth, tension=1] 
       plot coordinates { (0,0) (2, 0.8) (3, 0.75) (4, 0.8) (5, 1.3) (6, 1.5) (7, 2.1) (8,2) (9, 2.25) (10, 1.9) (12,3) };
  
  \draw[purple, very thick, domain=0:12, samples=100, smooth, tension=1] 
       plot coordinates { (0,0) (2, -0.5) (4,-1) (6, 0) (8,1) (10, 2.5) (12,3) };
  
  \draw[dotted, thick] (0,0) -- (12,3);
  
  \fill[black] (0,0) circle (3pt);
  \node[anchor=north east] at (0,0) {A};
  
  \fill[black] (3,0.75) circle (3pt);
  \node[anchor=south west] at (3, 0.75) {B};
  
  \fill[black] (6,1.5) circle (3pt);
  \node[anchor=south east] at (6, 1.5) {C};
  
  \fill[black] (9,2.25) circle (3pt);
  \node[anchor=south west] at (9, 2.25) {D};
  
  \fill[black] (12,3) circle (3pt);
  \node[anchor=south west] at (12,3) {E};
  
\end{tikzpicture}
\caption{In this example, the branching factor is $M = 4$. Na\"ively, all three paths are candidate paths. However, in our approach we only consider as a valid candidate the \textcolor{blue}{middle path}, which is rounded to the straight line from A to E on ``anchors'' B, C, and D.} \label{fig:diagram_curves}
\end{figure}

We define the {\em curvature} of a path to be the total, across all scales and intervals, deviation from a linear path between the endpoints of each interval:
$$\curveAviad(p) := \sum_{s \in [\log_M{(2n)}]} \sum_{(l,r] \in s} \sum_{i \in [M]} \left|(p({\anc}_i)-p({\anc}_{i-1})) - \frac{p(r) -p(l)}{M}\right|,$$
where we slightly abuse notation and use $s$ to denote both the scale and the set of scale-$s$ intervals. Note that curvature is closely related to tree total deviation --- curvature of $p$ is the tree total deviation on the forward difference of $(p(0), p(1), \dots, p(2n))$. 

For ED, after we rotate the grid by 45\textdegree, all diagonal edges have weight 1. Thus for any interval $(l, r]$, $\abs{p(l) - p(r)}$ must be at most the length of the sub-path of $p$ inside $(l, r]$.
Suppose that the actual ED is $\eps n$. Our result on the tree total deviation (\cref{lemma:totalmeandeviation}) implies that the total curvature across all scales is roughly only $\eps n \sqrt{\log_M{n}}$. This saving of $\sqrt{\log_M{n}}$ from the na\"ive upper bound of $O(\eps n \log_M{n})$ is crucial to our ultimate quasi-polynomial saving in running time: it allows us to force parts of the paths to be a straight line to avoid overfitting to the sampled columns. 

If the curvature of the optimal path inside an interval $(l, r]$ is small, then the optimal path is very close to a ``straight line'' from $\vertex{l}{p(l)}$ to $\vertex{r}{p(r)}$. This means that, at least on the ${\anc}_i$'s, we can ``round'' the optimal path to agree with the straight line without incurring too much error. Thus when computing the answers inside $(l, r]$, even if we only consider the points $\vertex{{\anc}_i}{p({\anc}_i)}$ that agree with the straight line from $\vertex{l}{p(l)}$ to $\vertex{r}{p(r)}$, we do not lose too much precision. Fixing the endpoints of the path on each sampled column enables us to sub-sample without the risk of overfitting.

\begin{figure}[ht!]
\centering
\caption{Rounding to a ``Straight-Line'' Path}\label{fig:diagram_rounding}
\begin{tikzpicture}[scale=1.5]

\draw[-, black, very thick] (0,0) -- (4, 4) -- (8, 0) -- (4, -4) -- (0, 0);

\foreach \i in {0,...,8} {
\draw[very thin] (\i,-4) -- (\i,4);

  }
  
  \draw[brightblue, line width=1mm, smooth, tension=0.3] 
     plot coordinates { 
       (0,0) (0.8, 0.3) (1, 0.4)(1.2, 0.5) (2, 1) (2.8, 0.7) (2.9, 0.65) (3, 0.6) (3.1, 0.55) (3.2, 0.5) 
       (3.5, 0.4) (3.6, 0.37) (3.75, 0.35) (4, 0.3) 
       (4.25, 0.105) (4.3, 0.07) (4.4, 0.035) (4.5, 0) (4.8, -0.105) (5, -0.175) 
       (5.3, -0.28) (6, -0.7) (6.8, -0.35) (6.9, -0.315) (7, -0.28) (7.1, -0.245) (7.2, -0.21) (8, 0) 
     };

\draw[deepred, line width=0.6mm, smooth, tension=0.3] 
     plot coordinates { 
       (0,0) (0.8, 0.3) (1, 0.5) (1.2, 0.5) (2, 1) (2.8, 0.7) (2.9, 0.6) (3, 0.5) (3.1, 0.55) (3.2, 0.5) 
       (3.5, 0.4) (3.6, 0.37) (3.75, 0.2) (4, 0) 
       (4.25, 0.084) (4.3, 0.07) (4.4, 0.035) (4.5, 0) (4.8, -0.245) (5, -0.35) 
       (5.3, -0.28) (6, -0.7) (6.8, -0.35) (7, -0.35) (7.2, -0.21) (8, 0) 
     };

  \draw[dotted, thick] (0,0) -- (2, 1) -- (4, 0) -- (6, -0.7) -- (8, 0);
  
  \draw[dashed, thick] (0,0) -- (8, 0);
  
  \fill[black] (0,0) circle (1pt);
  \node[anchor=east] at (0,0) {$\vertex{0}{n}$};
  
  \fill[black] (8,0) circle (1pt);
  \node[anchor=west] at (8,0) {$\vertex{2n}{n}$};

  \fill[black] (4,0) circle (1pt);
  \node[anchor=north] at (4,0) {$\vertex{n}{n}$};

  \fill[black] (2,1) circle (1pt);
  \node[anchor=south] at (2,1) {$\vertex{0.5n}{1.25n}$};

  \fill[black] (3,0.5) circle (1pt);
  \node[anchor=south west] at (3,0.5) {$\vertex{0.75n}{1.125n}$};

  \fill[black] (2, 0) circle (1pt);
  \node[anchor=north] at (2,0) {$\vertex{0.5n}{n}$};

\end{tikzpicture}
\caption*{In this example, the branching factor $M = 2$, and $\tilde{S} = \{S - 2, S\}$ (i.e., scales with problem sizes $2n$ and $0.5n$). The \textcolor{brightblue}{\bf optimal path} can be rounded to the \textcolor{deepred}{\bf regular path} without losing too much precision. The regular path must resemble ``straight-lines'' inside sub-problems on active scales. For example, since scale $S$ is active, and the {\color{deepred}regular path} visits $\vertex{0}{n}$ and $\vertex{2n}{n}$, it must also visit $\vertex{n}{n}$; since we sub-sample on scale $(S - 2)$, and the {\color{deepred}regular path}  visits $\vertex{0.5}{1.25n}$ and $\vertex{n}{n}$, it must also visit $\vertex{0.75n}{1.125n}$. However, note that we do not sub-sample on scale $(S - 1)$, so despite the fact that the {\color{deepred}regular path} visits $\vertex{0}{n}$ and $\vertex{n}{n}$, it does not have to visit $\vertex{0.5n}{n}$.} 
\end{figure}

We show (\cref{claim:regularsinglescale}) that the error incurred by rounding the optimal path to straight line anchors inside $(l, r]$ is bounded by roughly $M$ times its curvature inside $(l, r]$. We only do rounding and sub-sampling on active scales, which constitute a $1/\log ^ c_M(n)$-fraction of the scales. 
Now, as long as 
\begin{align*}
|\tilde{S}| \ll \sqrt{|S|}/ M,
\end{align*} 
We expect the total error incurred by rounding the optimal path to straight lines on all active scales (which results in a ``regular path'' (\cref{def:regular}), see \cref{fig:diagram_rounding} for an example of such a rounding) to be:
\begin{align*}
    M \cdot \left(\text{$\curveAviad(p)$ on active scales}\right) & \lesssim M \cdot \curveAviad(p) \cdot (\text{fraction of active scales}) \\
    & \lesssim  M \cdot \eps n \sqrt{|S|} \cdot \frac{|\tilde{S}|}{|S|} \\
    & = o(\eps n).
\end{align*}

These ideas suffice for obtaining an additive $o(n)$-approximation for ED and LCS\footnote{We introduced part of our ideas for ED exclusively, but for additive $o(n)$-approximation, these two problems are equivalent since LCS is (mostly) a complementary version of ED.}. 

\subsection{Multiplicative Approximations} \label{sub:overview-multiplicative}

To obtain multiplicative $(1\pm o(1))$-factor approximations for both ED and LCS, we need to introduce a few additional ideas. Here also the differences between ED and LCS become more important. 

\paragraph{Challenge 1 (LCS): vertical difference could be much larger than weight of the path}
For ED, we use the fact that the vertical difference $\abs{p(l) - p(r)}$ is at most the length of the $\vertex{l}{p(l)} - \vertex{r}{p(r)}$ path to get a bound on the curvature of $p$ inside the interval $(l,r]$. For LCS, the vertical difference is no longer bounded by the weight of the path so locally the curvature could generally be much higher than the contribution to LCS. 
%
%

Intuitively, we can simply discard intervals with very high curvature compared to low contribution to the LCS. 
To bound the total error incurred by discarding such intervals, we introduce a scaled variant of the tree total deviation (\cref{sec:scaledmd}).
We use it to show that indeed most of the contribution to LCS comes from segments of the path where the ratio between curvature and weight is bounded, so we can safely discard the remaining segments.

\paragraph{Challenge 2 (ED+LCS): uneven contribution from sub-paths}
For a fixed candidate path, if we simply split it into $M$ parts with equal widths, sample a random subset of $M / 2$ parts, and let our estimate be twice the answer in the sample, one would hope that the estimate would concentrate around the true sum. However, this concentration may not hold when the sub-paths on different parts have very different weights; e.g., in the extreme case, the sub-path in one of the $M$ parts has a much higher weight than the rest combined, and the estimate would fluctuate greatly based on whether this part is in the sample or not. 

Fortunately, if we apply the total tree deviation result on the ``weight distribution'' of the path, we can also see that on average, the sub-paths should not differ too much in lengths. It turns out that it suffices to obtain an estimate that satisfies the following desiderata:
\begin{itemize}
\item The estimate is ``(nearly) sound'', i.e., it does not underestimate the answer for ED (by less than $(1 - o(1))$), and does not overstimate the answer for LCS (by more than $(1 + o(1))$). 
\item As long as the sub-paths do not differ too much in lengths, the estimate is also ``(nearly) complete'' , i.e., it does not overestimate the answer for ED (by more than $(1 + o(1))$), and does not underestimate the answer for LCS (by less than $(1 - o(1))$).
\item The estimate also has ``back-up'' worst-case error guarantee, i.e., even if the sub-paths lengths vary greatly, we incur an additive error of at most $O(1)$ times the length of the path. 
\end{itemize}
For LCS, a minimization problem, we can simply let our back-up estimate be $0$ when the sub-paths differ a lot in lengths, since this is always sound. For ED, 
we can use an off-the-shelf constant-factor approximation (e.g., \cite{CDGKS20,GRS20,AN20}) to obtain a sound back-up estimate with a multiplicative error of $O(1)$.

\paragraph{Putting it all together}
If we never misestimate by more than a multiplicative factor of $\left(1 + \frac{\eps'} {|\tilde{S}|}\right)$ for each interval $(l, r]$ on an active scale for some $\eps' = o(1)$, the total error over $\tilde{S}$ active layers is bounded by:
 $$\left(1 + \frac{\eps'} {|\tilde{S}|}\right) ^ {|\tilde{S}|} \approx (1 + \eps'),$$
which gives us a $(1 + o(1))$-approximation. In our final algorithms, for an interval $(l, r]$ on an active scale, we use a simple concentration inequality to show that, for any specific path $p$ and interval $(l,r]$, the probability of misestimating the answer by more than $\left(1 + \frac{\eps'} {|\tilde{S}|}\right)$ due to sub-sampling is roughly $2^{-M \frac{{(\eps')} ^ 2}{|\tilde{S}|^2}}$.
Suppose that
\begin{align*}
    M \gg \log ^ {\Omega(1)}(n) / {{(\eps')} ^ 2}.
\end{align*} 
Then the probability of a large misestimation is $2 ^ {-\log ^ {\Omega(1)}(n)}$. Even when we largely misestimate, the estimate inside an interval $(l, r]$ should never exceed its width $r - l$. Thus intuitively, the total additive error is expected to be bounded by $n / 2 ^ {-\log ^ {\Omega(1)}(n)}$ (we will formally show this using an observation on a special type of layered graphs (\cref{claim:directed})), which suffices since we assume the answer is $n / 2 ^ {\log ^ {O(1)}(n)}$.

\section{Preliminaries}
\label{sec:prelims}
\subsection{Notations}



For $i \in \mathbb{N}$, let $[i]$ denote the set $\{1, 2, \dots, i\}$.

For $r \in \mathbb{R}$, let 
$\sign(r) := 
\begin{cases}
	-1, & \text{if $r \le 0$} \\
	+1, & \text{if $r > 0$}
\end{cases}
$ be the sign of $r$.

For a sequence $A = (a_i)$ of length $n$, we write $\sum A$ to denote its sum $\sum\limits_{i = 1} ^ n{a_i}$.



\subsection{Concentration Bound}
As an ingredient for our analysis, we will use Hoeffding's Inequality \cite{hoeffding}:
\begin{lemma} \label{lemma:hoeffding}
Let \( X = (x_1, \ldots, x_N) \) be $N$ independent random variables in $[-c, c]$.
Then, for all \( \delta > 0 \),
\[
\Pr\left(\abs{\sum_{i = 1}^{N} x_i - \mu} \ge \delta \right) \leq 2\exp\left(-\frac{\delta^2}{2Nc ^ 2}\right),
\]
where \( \mu = \Ex\left[ \sum_{i = 1}^{N} x_i \right] \).
\end{lemma}



\subsection{Black-Box Calls to Existing ED and LCS Algorithms}

We will use an $O(1)$-approximate ED algorithm, e.g., the following tradeoff proved in~\cite{GRS20}%
\footnote{We note that the choice of using~\cite{GRS20} is arbitrary. Since neither the running time nor the exact approximation factor is the bottleneck for our analysis, any of the truly sub-quadratic $O(1)$-approximation algorithms, starting with~\cite{CDGKS20} and including~\cite{CDGKS20, Andoni20, BR20, KS20, AN20,RSSS23} would suffice for our purpose.}: 
\begin{lemma} \label{lemma:estimate}
    Given two strings $X$ and $Y$ with lengths that sum to $n$, there is a randomized algorithm that computes a 4-approximation of $\ed(X, Y)$ in $O(n ^ {1.61})$ time.
\end{lemma}

We also use Ukkonen's algorithm~\cite{Ukkonen85} for exact ED when the ED is small%
\footnote{There are faster algorithms for this problem, e.g., \cite{LMS98}. However the difference would not translate to a qualitative improvement for our purposes. It is interesting that Ukkonen's algorithm~\cite{Ukkonen85} also works for the more general setting of shortest paths on grids with diagonal shortcuts. Thus the only part in our algorithm for Edit Distance that requires actual strings is the $O(1)$-approximation algorithm from \cref{lemma:estimate}.}:
\begin{lemma} \label{lemma:bounded}
    Given two strings $X$ and $Y$ with lengths that sum to $n$ and an upper bound $k$, there is an algorithm running in $O(nk)$ time that can:
    \begin{itemize}
        \item Decide whether $\ed(X, Y) \le k$;
        \item If $\ed(X, Y) \le k$, compute $\ed(X, Y)$ exactly.
    \end{itemize}
\end{lemma}

When the LCS is small, the following result has been known since the 70s \cite{10.1145/322033.322044}:
\begin{lemma} \label{lemma:boundedlcs}
    Given two strings $X$ and $Y$ with lengths that sum to $n$ and an upper bound $k$, there is an algorithm running in $O(nk + n \log (n))$ time that can:
    \begin{itemize}
        \item Decide whether $\lcs(X, Y) \le k$;
        \item If $\lcs(X, Y) \le k$, compute $\lcs(X, Y)$ exactly\footnote{The original result states that there is an algorithm running in $O(n\lcs(X, Y) + n \log (n))$ time, this can easily be applied to the decision problem here by simply testing if the algorithm terminates within $O(nk + n \log (n))$ time.}.
    \end{itemize}
\end{lemma}

\subsection{Rounding}
To avoid fractional results after divisions, for $\lambda = 2 ^ {\floor{0.99\log(n)}} = \Theta(n ^ {0.99})$, we ensure that the lengths of two input strings are multiples of $\lambda$ by padding them with arbitrary characters. The error we introduce with this rounding is $O(\lambda)$, which is within $(1 + o(1))$ when the ED/LCS is $\omega(\lambda)$. When the ED/LCS is $O(\lambda)$, we can use \cref{lemma:bounded} or \cref{lemma:boundedlcs} to compute the answer exactly in $O(n\lambda) = O(n ^ {1.99})$ time. Thus we can assume that input strings both have a length that is a multiple of a power of $2$ of order $\Omega(n ^ {0.99})$.

\section{Tree Total Deviation} \label{sec:md}
In this section we will introduce the notion of tree total deviation, which will be an important tool in analyzing our main algorithm for Edit Distance and LCS. We first introduce it in its basic form (\cref{sec:basic}). We then study a scaled variant of it (\cref{sec:scaledmd}).
\subsection{Basic Form} \label{sec:basic}
Given a sequence of real numbers $A = (a_1, a_2, \ldots, a_{n})$ with $n = M ^ S$ for $M \in 2\mathbb{Z} ^ +$ and $S \in \mathbb{Z} ^ +$, we consider its $S$-scale recursive partition with branching factor $M$.  In other words, for each scale $s \in [S]$, we consider the partition of $A$ to disjoint length-$M ^ s$ subsequences $A_{(l,r]} = (a_{l+1} ,\dots ,a_{r})$. We say that $(l,r]$ is a \emph{scale-$s$ interval}. Every interval $(l, r]$ that is involved in the recursive partitioning (i.e., is a scale-$s$ interval for some scale $s \in [S]$) is called a \emph{recursion interval}.

For a scale-$s$ interval $(l,r]$, let the $i$-th ($0 \le i \le M$) \emph{anchor} be ${\anc}_i := l + i \cdot \frac{r - l}{M}$. Let $({\anc}_{i - 1}, {\anc}_{i}]$ denote its $i$-th \emph{sub-interval}. We consider the sum of each of $A_{(l,r]}$'s sub-intervals, and compute the total deviation of the sums, namely the sum of how far each sub-interval is from the mean of the sub-intervals:
$$\md_M(A_{(l,r]}) := \sum_{i= 1}^M {\abs{{A_{\Sigma({\anc}_{i - 1},{\anc}_{i}]}} - \frac{1}{M}{A_{\Sigma(l,r]}}  }},$$
where ${{A'}_{\Sigma(l', r']}}$ denotes the sum $\sum\limits_{l' < i \le r'}{{a'}_i}$, a shorthand we will use throughout the paper.

We are particularly interested in the ($M$-way) \emph{tree total deviation}, which sums the total deviation across the entire recursive partitioning of $A$: 
$$\treemd_{M, S}(A) := \sum_{s \in [S]}\sum_{(l,r] \in s}{\md_M(A_{(l,r]})},$$
where we slightly abuse notation and use $s$ to denote both the scale and the set of scale-$s$ intervals. We also define a refined version of tree total deviation on a subset $\mathcal{S} \subset [S]$ of scales.
$$\treemd_{M, \mathcal{S}}(A) := \sum_{s \in \mathcal{S}}\sum_{(l,r] \in s}{\md_M(A_{(l,r]})}.$$

Intuitively, $\md_M(A)$ is a measure of how  ``uneven'' $A$ is after dividing it into $M$ parts --- namely if we divide $A$ into $M$ parts, how far they are from all having the same sum.
$\treemd_{M, \mathcal{S}}(A)$ is a measure of the sum of the total deviation of $A$ over all scales in $\mathcal{S}$.

Suppose $|a_i| \le 1$ for all $i$. A na\"ive upper bound on $\treemd_{M, S}(A)$ would count a contribution of $O(1)$ from each $a_i$
to the total deviation on each scale, or $\treemd_{M, S}(A) = O(nS)$ in total. 
A key observation for our analysis is that in fact we can show that $\treemd_{M, S}(A)$ is close to $n\sqrt{S}$:

\begin{restatable}{lemma}{lemmatotalmeandeviation} \label{lemma:totalmeandeviation}
    Let $A = (a_1, a_2, \ldots, a_n)$ be a sequence of real numbers where $n = M ^ S$. 
    Then, 
    $$\treemd_{M, S}(A) = O\left(\left( \sum_x{\abs{a_x}}\right)S ^ {0.51}  + \max_x{\abs{a_x}} \cdot \left(n / 2 ^ {\log_M ^ {0.015}(n)}\right)\right).$$
\end{restatable}

In \cref{app:infotheory}, we provide an information-theoretic proof for a slightly stronger version of the lemma. In this section, we provide a more accessible and self-contained proof using basic probabilistic arguments.

Towards proving \cref{lemma:totalmeandeviation}, we will first prove the following claim:
\begin{claim} \label{claim:totalmeandeviationtool}
    Let $B = (b_1, b_2, \ldots, b_{M})$ be a sequence of numbers. There exists $\gamma \in \{-1, +1\} ^ M$ with exactly $M / 2$ (+1)'s such that
    $$2\sum_{i}{b_i\gamma_i} \ge \md_M(B).$$
\end{claim}
\begin{proof}
    Without loss of generality, suppose $B$ is \emph{sorted in increasing order}, and let $\overline{B}$ denote the average over its elements. We will show that we can satisfy the inequality by setting $\gamma_{1 \ldots M / 2}$ to $-1$ and $\gamma_{(M/2 + 1) \ldots M}$ to $+1$. Namely,
    $$2\sum_{1 \le i \le M / 2}{(b_{M - i + 1} - b_i)} \ge \sum_{i}{\abs{b_i - \overline{B}}}.$$

    Let $t$ be the smallest index such that $b_t \ge \overline{B}$. Due to symmetry, we can assume that $t - 1 \le M / 2$. Let $s := \sum\limits_{1 \le i \le M}{\abs{b_i - \overline{B}}}$, $s ^ {-} := \sum\limits_{1 \le i \le t - 1}{(\overline{B} - b_i)}$, and $s ^ {+} := \sum\limits_{t \le i \le M}{(b_i - \overline{B})}$. Then we have $s = s ^ {-} + s ^ {+}$. Moreover, since $\overline{B}$ is the mean of $B$, we must have $s ^ {+} - s ^ {-} = 0$. Therefore, $s ^ {-} = s ^ {+} = s / 2$. Thus 
    \begin{align*}
        2\sum_{1 \le i \le M / 2}{(b_{M - i + 1} - b_i)}    &\ge 2\sum_{1 \le i \le t - 1}{(b_{M - i + 1} - b_i)} \\
                                                        &\ge 2\sum_{1 \le i \le t - 1}{(\overline{B} - b_i)} \\
                                                        &= 2s ^ {-} \\
                                                        &= s.
    \end{align*}
\end{proof}

Now we are ready to prove \cref{lemma:totalmeandeviation}:
\begin{proof} [Proof of \cref{lemma:totalmeandeviation}]
For every scale-$s$ interval $(l,r]$, by \cref{claim:totalmeandeviationtool} we can find a vector $$\left(\gamma_{({\anc}_0, {\anc}_1]}, \gamma_{({\anc}_1, {\anc}_2]}, \ldots, \gamma_{({\anc}_{M - 1}, {\anc}_M]}\right)$$ such that 
$$\md_M(A_{(l,r]}) \le 2\sum_{i=1}^M{\gamma_{({\anc}_{i - 1}, {\anc}_{i}]} {A_{\Sigma({\anc}_{i - 1},{\anc}_{i}]}}}.$$
For each index $x \in (l, r]$, let $g_{x,s}$ denote the sign of  $\gamma_{({\anc}_{i - 1}, {\anc}_{i}]}$ corresponding to the sub-interval $({\anc}_{i - 1}, {\anc}_{i}] \ni x$. We have
\begin{equation} \label{eq:totalmeandeviation0}
    \md_M(A_{(l,r]}) \le 2\sum_{x \in (l,r]} g_{x,s}a_x.
\end{equation}
    
Let $g_x := \sum_{s \in [S]} g_{x,s}$ denote the sum of the signs for index $x$ across all scales. By taking the sum of \eqref{eq:totalmeandeviation0} over every recursion interval, we get
\begin{gather} \label{eq:tcurv-bound}
    \treemd_{M, S}(A) \le 2\sum_x g_x a_x \le  2\sum_x \abs{g_x}\abs{a_x}.  
\end{gather}
    
We bound $\sum_x \abs{g_x}\abs{a_x}$ with a probabilistic argument bounding $\Ex_x \left[\abs{g_x}\abs{a_x}\right]$. Namely, we randomly sample an index $x \in [n]$.
As a warm-up, let us ignore the $\abs{a_x}$ term for now. Notice that for any scale-$s$ interval $A_{(l,r]}$, the values of $g_{x,s'}$ for $s' > s$ are identical for all $x \in (l,r]$ since those values depend on $\gamma_{(l',r']}$ of larger intervals  $(l',r'] \supseteq (l,r]$. Furthermore, $g_{x,s}$ is $+1$ (respectively $-1$) for exactly half of the $x \in (l,r]$. Hence $g_x$ is distributed identical to the sum of $S$ i.i.d.~$\pm 1$ coin flips. Therefore we have $\Ex_x \left[\abs{g_x}\right] \le \sqrt{S}$.

Similarly, bringing back $\abs{a_x}$, we essentially need to bound the $\sum_x \abs{a_x}$-tail of a distribution of $S$ i.i.d.~$\pm 1$ coin flips. Formally, we say that an index $x$ is in the ``tail'' if $|g_x| > S ^ {0.51}$, and we want to bound the sum of $\abs{a_x}$ for every $x$ in the tail, which is no more than $\max_x{\abs{a_x}}$ times the number of $x$'s in the tail. By applying Hoeffding's inequality (\cref{lemma:hoeffding}), the probability that $\abs{g_x} > S ^ {0.51}$ is  $1 / 2 ^ {\Omega\left(\log ^ {0.02}_M(n)\right)}$ (recall that $S = \log_M(n)$). Therefore, there are at most an $1 / 2 ^ {\Omega\left(\log ^ {0.02}_M(n)\right)}$-fraction of tail indices $x$ such that $\abs{g_x} > S ^ {0.51}$, and each of these $\abs{g_x}$ values is no more than $S$. 
By summing the contribution up of the ``core'' indices satisfying  $\abs{g_x} \le S ^ {0.51}$ and the tail indices, we have
\begin{align*} \sum_x \abs{g_x}\abs{a_x}   & \le  \left(\sum_x \abs{a_x}\right) S ^ {0.51} + \max_x{\abs{a_x}} \cdot \left|\{x \text{\;s.t.\;} \abs{g_x} > S ^ {0.51} \}\right| \cdot \max \{\abs{g_x}\} \\
& \le \left(\sum_x \abs{a_x} \right) S ^ {0.51}  + \max_x{\abs{a_x}} \cdot \left(n / 2 ^ {\Omega\left(\log ^ {0.02}_M(n)\right)}\right) \cdot S \\
& = \left(\sum_x \abs{a_x} \right) S ^ {0.51}  + \max_x{\abs{a_x}} \cdot o\left(n / 2 ^ {\log_M ^ {0.015}(n)}\right).
\end{align*}
Plugging this into \eqref{eq:tcurv-bound}, we have that indeed,
$$\treemd_{M, S}(A) \le O\left(\left( \sum_x{\abs{a_x}}\right)S ^ {0.51}  + \max_x{\abs{a_x}} \cdot \left(n / 2 ^ {\log_M ^ {0.015}(n)}\right)\right).$$
\end{proof}

\subsection{Tree Total Deviation Scaled by Another Sequence} \label{sec:scaledmd}
Suppose we have another sequence $B = (b_1, b_2, \ldots, b_n)$ of length $n$. We consider the following question: how much does $B$ concentrate on the high deviation regions of $A$? For an interval $(l, r]$ on some scale $s$, we define the \emph{scaled total deviation} of $A(l, r]$ by $B(l, r]$ to be $\md_M(A(l, r])$ scaled by the $B$-to-$A$ ratio:
\[
\rd_M{(A(l, r], B(l, r])} := \frac{B_{\Sigma(l,r]}}{A_{\Sigma(l, r]}}\md_M(A(l, r]),
\]
And specifically, if $A_{\Sigma(l, r]} = 0$, we let $\rd_M{(A(l, r], B(l, r])} := 0$.
We define the scaled tree total deviation of $A$ by $B$ to be the total sum of scaled total deviation over all recursion intervals:
\[\treerd_{M, S}(A, B) := \sum_{s \in [S]}\sum_{(l,r] \in s}{\rd_M(A(l, r], B(l, r])}\]

We have the following lemma which states that $B$ does not concentrate on the high deviation regions of $A$ as long as the $B$-to-$A$ ratio is within some reasonable range in every recursion interval:
\begin{lemma} \label{lemma:treerd}
    Let $A = (a_1, a_2, \ldots, a_n)$ and $B = (b_1, b_2, \ldots, b_n)$ be sequences of non-negative  real numbers where $n = M ^ S$. Suppose that there exists $0 < \alpha \le 1$ such that for every recursion interval $(l, r]$, we have
    \begin{align*}
        \alpha A_{\Sigma(l, r]} \le B_{\Sigma(l, r]} \le A_{\Sigma(l, r]}.
    \end{align*}
    Then we have
    \begin{equation} \label{eq:treerd}
    \treerd_{M, S}(A, B) = O\left(\left(1 + \log{\left(\frac{1}{\alpha}\right)}\right)\left(\left( \sum_x{b_x}\right)S ^ {0.51}  + \max_{x'}{(a_{x'})} \cdot \left(n / 2 ^ {\log_M ^ {0.015}(n)}\right)\right)\right)
    \end{equation}
\end{lemma}
\begin{proof}
Let us consider the constant-ratio case where $\alpha = \Omega(1)$.
    In which case, we have
\[
\rd_M(A(l, r], B(l, r]) = \frac{B_{\Sigma(l,r]}}{A_{\Sigma(l, r]}}\md_M(A(l, r]) = O(\md_M(A(l, r]))
\]
(when $A_{\Sigma(l, r]} = 0$, this term becomes $0$ by definition of scaled total deviation, which is also $O(\md_M(A(l, r]))$ since total deviation is non-negative). 
Therefore, we have
\begin{align*}
\treerd_{M, S}(A, B) &= 
\sum_{s \in [S]}\sum_{(l,r] \in s}{\rd_M(A(l, r], B(l, r])} \\ 
&= \sum_{s \in [S]}\sum_{(l,r] \in s}{O(\md_M(A(l, r]))} \\ &= O(\treemd_{M, S}(A))\\
&= O\left(\left( \sum_x{a_x}\right)S ^ {0.51}  + \max_{x'}{(a_{x'})} \cdot \left(n / 2 ^ {\log_M ^ {0.015}(n)}\right)\right) && \text{(from \cref{lemma:totalmeandeviation})}.
\end{align*}
Since $\alpha = \Omega(1)$, we have $B_{\Sigma(0, n]} = \Omega(A_{\Sigma(0, n]})$. Thus we also have
\begin{equation} \label{eq:rdlemma}
\treerd_{M, S}(A, B) = O\left(\left( \sum_x{b_x}\right)S ^ {0.51}  + \max_{x'}{(a_{x'})} \cdot \left(n / 2 ^ {\log_M ^ {0.015}(n)}\right)\right).
\end{equation}

\xiao{added transition and simplified argument via linearity} Before going into the fully general case, we partially generalize \eqref{eq:rdlemma} by stating that it holds when for some $\alpha \le \beta \le 1$ we have 
    \begin{align*}
        B_{\Sigma(l, r]} = \Theta(\beta A_{\Sigma(l, r]})
    \end{align*}
    for every recursion interval $(l, r]$ (i.e., the $B$-to-$A$ ratio falls within an $O(1)$-narrow range for every recursion interval). To see this, we can easily see that $\treerd$ is a linear function of $B$: for any $\gamma \in \mathbb{R}$,
    $$\treerd_{M, S}(A, \gamma \cdot B) = \gamma \cdot \treerd_{M, S}(A, B),$$
    where $\gamma \cdot b$ is the element-wise product of $\gamma$ and $B$. Note that the case with $A$ and $(1 / \beta) \cdot B$ falls into the constant-ratio case above. We have
    \begin{align*}
    \treerd_{M, S}(A, B) &= \beta \cdot \treerd_{M, S}(A, (1 / \beta) \cdot B) && \text{from linearity} \\
    &= \beta \cdot O\left(\left(\sum_x{\left(b_x / \beta\right)}\right)S ^ {0.51}  + \max_{x'}{(a_{x'})} \cdot \left(n / 2 ^ {\log_M ^ {0.015}(n)}\right)\right) && \text{from \eqref{eq:rdlemma}} \\
    &= O\left(\left(\beta \cdot \sum_x{\left(b_x / \beta\right)}\right)S ^ {0.51}  + \beta \cdot \max_{x'}{(a_{x'})} \cdot \left(n / 2 ^ {\log_M ^ {0.015}(n)}\right)\right)  \\
    &= O\left(\left(\beta \cdot \sum_x{\left(b_x / \beta\right)}\right)S ^ {0.51}  + \max_{x'}{(a_{x'})} \cdot \left(n / 2 ^ {\log_M ^ {0.015}(n)}\right)\right) && \text{since $\beta \le 1$} \\
    &= O\left(\left(\sum_x{b_x}\right)S ^ {0.51}  + \max_{x'}{(a_{x'})} \cdot \left(n / 2 ^ {\log_M ^ {0.015}(n)}\right)\right).
\end{align*} \xiao{changed part ends}

Motivated by this, we now consider partitioning the range of possible $B$-to-$A$ ratio $[\alpha, 1]$ into $\ceil{{\log_2{\left(\frac{1}{\alpha}\right)}}}$ parts:
\[[\alpha, 1] \subset \left[\frac{1}{2}, 1\right] \cup \left[\frac{1}{4}, \frac{1}{2}\right] \cup \left[\frac{1}{8}, \frac{1}{4}\right] \cup \cdots \cup\left[\frac{1}{2 ^ {\ceil{{\log_2{\left(\frac{1}{\alpha}\right)}}}}}, \frac{1}{2 ^ {\ceil{{\log_2{\left(\frac{1}{\alpha}\right)}}} - 1}}\right].\]
For each part $\left[\frac{1}{2 ^ {k}}, \frac{1}{2 ^ {k - 1}}\right]$, we define the sum
\[\sigma_k := \sum_{s \in [S]}\sum_{(l,r] \in s}{\mathbbm{1}\left[B_{\Sigma(l, r]} \in \left[\frac{1}{2 ^ {k}}A_{\Sigma(l, r]}, \frac{1}{2 ^ {k - 1}}A_{\Sigma(l, r]}\right]\right]\rd_M(A(l, r], B(l, r])},\]
Namely the contribution to \eqref{eq:treerd} from intervals whose $B$-to-$A$ ratios are within $\left[\frac{1}{2 ^ {k}}, \frac{1}{2 ^ {k - 1}}\right]$. By definition we have
\[\sum_{s \in [S]}\sum_{(l,r] \in s}{\rd_M(A(l, r], B(l, r])} \le \sum_{k=1}^{\ceil{{\log_2{\left(\frac{1}{\alpha}\right)}}}}\sigma_k.\]
Since $\ceil{{\log_2{\left(\frac{1}{\alpha}\right)}}} < 1 + \log_2{\left(\frac{1}{\alpha}\right)}$, to show \eqref{eq:treerd} it suffices to show that for each $k \in \left[1, \ceil{{\log_2{\left(\frac{1}{\alpha}\right)}}}\right]$, we have
\[\sigma_k = O\left(\left( \sum_x{b_x}\right)S ^ {0.51}  + \max_{x'}{(a_{x'})} \cdot \left(n / 2 ^ {\log_M ^ {0.015}(n)}\right)\right).\]

Fix a $k \in \left[1, \ceil{{\log_2{\left(\frac{1}{\alpha}\right)}}}\right]$. Consider the set of intervals $I$ containing all recursion intervals $(l, r]$ such that $B_{\Sigma(l, r]} \in \left[\frac{1}{2 ^ {k}}A_{\Sigma(l, r]}, \frac{1}{2 ^ {k - 1}}A_{\Sigma(l, r]}\right]$, namely intervals that contribute to $\sigma_k$. 
Consider length-$n$ sequences $C = \{c_x\}$ and $D = \{d_x\}$ defined as copies of $A, B$ with the regions that do not contribute to $\sigma_k$ zeroed out:
\begin{align*}
c_x &:= \begin{cases}
    a_x & \text{if $\exists (l, r] \in I$ s.t.~$x \in (l, r]$} \\
    0 & \text{otherwise}
\end{cases}, \\
d_x &:= \begin{cases}
    b_x & \text{if $\exists (l, r] \in I$ s.t.~$x \in (l, r]$} \\
    0 & \text{otherwise}
\end{cases}.
\end{align*}
For every $(l, r] \in I$, we then have
\begin{equation} \label{eq:abcdrel}
\rd_M(A(l, r], B(l, r]) = \rd_M(C(l, r], D(l, r])
\end{equation}
and
\begin{equation} \label{eq:rdtdrelcd}
\rd_M(C(l, r], D(l, r]) = O\left(\frac{1}{2 ^ k} \cdot \md_M(C(l, r])\right).
\end{equation}

For any distinct $(l_0, r_0], (l_1, r_1] \in I$, either $(l_0, r_0]$ and $(l_1, r_1]$ are disjoint, or one contains another. Thus, it is possible to find a ``dominating'' subset $I' \subset I$ of disjoint intervals such that for every $(l, r] \in I$, there exists $(l', r'] \in I'$ such that $(l, r] \subseteq (l', r']$. By taking the sums over each of these disjoint intervals in $I'$, we must have 
\begin{equation} \label{eq:rdtotalsum}
    \sum_x{d_x} = \Theta\left(\frac{1}{2 ^ {k}}\sum_{x'}{c_{x'}}\right).
\end{equation}

Now we have
\begin{align*}
\sigma_k &= \sum_{s \in [S]}\sum_{(l,r] \in s}{\mathbbm{1}\left[(l, r] \in I]\right] \cdot \rd_M(A(l, r], B(l, r])} \\
&= \sum_{s \in [S]}\sum_{(l,r] \in s}{\mathbbm{1}\left[(l, r] \in I]\right] \cdot \rd_M(C(l, r], D(l, r])} && \text{due to \eqref{eq:abcdrel}} \\
            &= O\left(\sum_{s \in [S]}\sum_{(l,r] \in s}{\mathbbm{1}\left[(l, r] \in I]\right] \cdot \frac{1}{2 ^ k} \cdot \md_M(C(l, r])}\right) && \text{due to \eqref{eq:rdtdrelcd}} \\
            &\le O\left(\sum_{s \in [S]}\sum_{(l,r] \in s}{\frac{1}{2 ^ k} \cdot \md_M(C(l, r])}\right) \\
            &= O\left(\frac{1}{2 ^ k} \cdot \treemd_{M, S}(C)\right) && \text{by definition} \\ 
            &= O\left(\frac{1}{2 ^ k} \cdot \left(\left( \sum_x{c_x}\right)S ^ {0.51}  + \max_{x'}{(c_{x'})} \cdot \left(n / 2 ^ {\log_M ^ {0.015}(n)}\right)\right)\right) && \text{from \cref{lemma:totalmeandeviation}} \\
            &= O\left(\frac{1}{2 ^ k} \cdot \left( \sum_x{c_x}\right)S ^ {0.51}  + \frac{1}{2 ^ k} \cdot \max_{x'}{(c_{x'})} \cdot \left(n / 2 ^ {\log_M ^ {0.015}(n)}\right)\right) \\
            &= O\left(\frac{1}{2 ^ k} \cdot \left( \sum_x{c_x}\right)S ^ {0.51}  + \max_{x'}{(c_{x'})} \cdot \left(n / 2 ^ {\log_M ^ {0.015}(n)}\right)\right) && \text{since $\frac{1}{2 ^ k} \le 1$}\\
            &= O\left(\left(\sum_x{d_x}\right)S ^ {0.51}  + \max_{x'}{(c_{x'})} \cdot \left(n / 2 ^ {\log_M ^ {0.015}(n)}\right)\right) && \text{due to \eqref{eq:rdtotalsum}} \\
            &= O\left(\left(\sum_x{b_x}\right)S ^ {0.51}  + \max_{x'}{(c_{x'})} \cdot \left(n / 2 ^ {\log_M ^ {0.015}(n)}\right)\right) && \text{since $b_{x} \le d_{x}$} \\
            &= O\left(\left(\sum_x{b_x}\right)S ^ {0.51}  + \max_{x'}{(a_{x'})} \cdot \left(n / 2 ^ {\log_M ^ {0.015}(n)}\right)\right). && \text{since $c_{x'} \le a_{x'}$}
\end{align*}
\end{proof}

In an earlier manuscript, we mistakenly thought the lemma would hold without the requirement on the $W$-to-$A$ ratio (without the $\left(1 + \log{\left(\frac{1}{\alpha}\right)}\right)$ term in the bound). As an anonymous reviewer kindly pointed out, the requirements on the $W$-to-$A$ ratio are necessary, see \cref{appendix:treerdold} for further discussion. 

\section{Shortest Path Problems on Grids}
\subsection{Introduction}
\subsubsection{Grid Graph for LCS and Edit Distance} \label{sec:gridgraphbasic}
Given two strings $X,Y$ of length $n_X$ and ${n_Y}$, we define an edge-weighted graph, corresponding to the cost of moving between states in the standard quadratic-time DP algorithm for ED (closely related to the \emph{alignment graph} defined in e.g., \cite{AKW23}):
\begin{definition}[Edit Distance Grid Graph] \label{def:edgridoriginal}
	The Edit Distance Grid Graph is a grid graph with diagonal shortcuts that can be used to compute $\ed(X, Y)$, and is defined as follows:
\begin{description}
    \item[Vertices] The grid points of $[0,n_X] \times [0,{n_Y}]$. The vertex at the intersection of the $x$-th column and the $y$-th row is $\vertexoriginal{x}{y}$;
    \item[Edges] The set of edges includes standard grid edges and diagonal shortcut edges:
    \begin{itemize}
        \item Standard grid edges from $\vertexoriginal{u}{v}$ to $\vertexoriginal{u}{v + 1}$ or $\vertexoriginal{u + 1}{v}$, always with weight $1$ (corresponding to insertion or deletion);
        \item Diagonal shortcut edges from $\vertexoriginal{u}{v}$ to $\vertexoriginal{u + 1}{v + 1}$, whose weight is $0$ if $X[u] = Y[v]$ (match) and $1$ otherwise (substitution).
    \end{itemize}
    \item[Objective] $\ed(X, Y)$ is the length (sum of weights) of the shortest $(0,0)$-$(n_X,{n_Y})$ path. 
\end{description}
\end{definition}

Given two strings $X,Y$ of length $n_X$ and ${n_Y}$, we define an edge-weighted graph, corresponding to the profit of moving between states in the standard quadratic-time DP algorithm for LCS:
\begin{definition}[LCS Grid Graph] \label{def:lcsgridoriginal}
	The LCS Grid Graph is a grid graph with diagonal shortcuts that can be used to compute $\lcs(X, Y)$, and is defined as follows:
\begin{description}
    \item[Vertices] The grid points of $[0,n_X] \times [0,{n_Y}]$. The vertex at the intersection of the $x$-th column and the $y$-th row is $\vertexoriginal{x}{y}$;
    \item[Edges] The set of edges includes standard grid edges and diagonal shortcut edges:
    \begin{itemize}
        \item Standard grid edges from $\vertexoriginal{u}{v}$ to $\vertexoriginal{u}{v + 1}$ or $\vertexoriginal{u + 1}{v}$, always with weight $0$ (corresponding to no match between the two characters);       \item Diagonal shortcut edges from $\vertexoriginal{u}{v}$ to $\vertexoriginal{u + 1}{v + 1}$, whose weight is $1$ if $X[u] = Y[v]$ (match) and $0$ otherwise (no match).
    \end{itemize}
    \item[Objective] $\lcs(X, Y)$ is the length (sum of weights) of the longest $(0,0)$-$(n_X,{n_Y})$ path. 
\end{description}
\end{definition}

\subsubsection{Rotated Grid Graph: $\vertexoriginal{u}{v} \rightarrow \vertex{u + v}{v - u + {n_X}}$ } \label{sec:grid}
	To describe our algorithm, we will rotate the grids by 45 degrees and obtain graphs where for any fixed path, each column (corresponding to a diagonal in the original grid) is visited exactly once. 
\begin{restatable}[Rotated Edit Distance Grid Graph]{definition}{defedgrid} \label{def:edgrid}
We use $\vertexoriginal{x}{y}$ to denote a vertex in the (original) Edit Distance Grid Graph at the intersection of column $x$ and row $y$ and $\vertex{x}{y}$ to denote a vertex in the Rotated Edit Distance Grid Graph at the intersection of column $x$ and row $y$. The description of the new graph on the rotated grid is as follows:
\begin{description}
    \item[Vertices] The vertex for $\vertexoriginal{u}{v}$ in the original grid corresponds to $\vertex{u + v}{v - u + {n_X}}$ on the new grid;
    \item[Edges] We start with the same set of edges as in the original graph. Then, we divide each shortcut edge with weight $w$ ($\in \{0, 1\}$) from $\vertex{u}{v}$ to $\vertex{u + 2}{v}$ into two edges, one from $\vertex{u}{v}$ to $\vertex{u + 1}{v}$ with weight $0$ and the other from $\vertex{u + 1}{v}$ to $\vertex{u + 2}{v}$ with weight $w$. 
    \item[Objective] The length of the shortest $\vertex{0}{{n_X}}$-$\vertex{n_X + {n_Y}}{n_Y}$ path equals $\ed(X, Y)$.
\end{description}
\end{restatable}
\begin{definition}[Rotated LCS Grid Graph] \label{def:lcsgrid}
We use $\vertexoriginal{x}{y}$ to denote a vertex in the (original) LCS Grid Graph at the intersection of column $x$ and row $y$ and $\vertex{x}{y}$ to denote a vertex in the Rotated LCS Grid Graph at the intersection of column $x$ and row $y$. The description of the new graph on the rotated grid is as follows:
\begin{description}
    \item[Vertices] The vertex for $\vertexoriginal{u}{v}$ in the original grid corresponds to $\vertex{u + v}{v - u + {n_X}}$ on the new grid;
    \item[Edges] We start with the same set of edges as in the original graph. Then, we divide each shortcut edge with weight $w$ ($\in \{0, 1\}$) from $\vertex{u}{v}$ to $\vertex{u + 2}{v}$ into two edges, one from $\vertex{u}{v}$ to $\vertex{u + 1}{v}$ with weight $0$ and the other from $\vertex{u + 1}{v}$ to $\vertex{u + 2}{v}$ with weight $w$. 
    \item[Objective] The length of the longest path from $\vertex{0}{{n_X}}$ to $\vertex{n_X + {n_Y}}{n_Y}$ equals $\lcs(X, Y)$.
\end{description}
\end{definition}

	For the remainder of the paper, a \emph{path} refers to a route from any $\vertex{l}{y_l}$ to any $\vertex{r}{y_r}$. Due to the split edges, any path $p$ in the grid graph from any column $l$ to any column $r$ visits each column $x \in [l, r]$ exactly once. We denote by $p(x)$ the y-coordinate (row) of the vertex $\vertex{x}{p(x)}$ where the path crosses column $x$. By the construction of the rotated grid graphs, $\abs{p(x) - p(x - 1)} \le 1$ for every $l < x \le r$. For an interval $(l ^ {\prime}, r ^ {\prime}] \subset(l, r]$, we write $p(l ^ {\prime}, r ^ {\prime}]$ to denote the sub-path of $p$ from $l ^ {\prime}$ to $r ^ {\prime}$. 
 We define $w(p)$ as the length (total weight of edges) of $p$. 

We have the following observation:
\begin{observation} \label{obs:pathdiff}
For a path $p$, and for any two columns $i \le j$ that the path visits, $$\abs{p(i) - p(j)} \le j - i.$$    
\end{observation}
\begin{proof}
    We prove this by induction by increasing value of $j - i$. When $i = j$, the claim is trivially true since the path visits each column exactly once. Now suppose that the claim holds for $i, j-1$. The only edges from $\vertex{j-1}{p(j-1)}$ go to $\vertex{j}{p(j-1)}$, and (possibly) $\vertex{j}{p(j-1)\pm 1}$. 
    Therefore, 
    $$ \abs{p(i) - p(j)} \le \abs{p(i) - p(j-1)} +1 \le (j-1 -i) + 1 = j - i.$$
\end{proof}
Due to \cref{obs:pathdiff}, given two vertices $\vertex{l}{y_l}$ and $\vertex{r}{y_r}$, if $\vertex{r}{y_r}$ is reachable from $\vertex{l}{y_l}$, we must have $\abs{y_l - y_r} \le r - l$. 

\subsection{Reduction to Query-Based Approximations on Sparse Grids} \label{sec:sparsified}
We now show that the problems on the Rotated Grid Graphs can be further reduced to query-based approximation problems on sparsified versions of these grid graphs defined just now in \cref{sec:grid}. 

Let $M := 2 ^ {\floor{\log (\log ^ {\mval} (n))}} = \Theta(\log ^ {\mval} (n))$ and $S := 3 \cdot \floor{\log_M{(n ^ {0.01})}}$\footnote{Making $S$ divisible by $3$ is convenient for our algorithm for LCS.}. Let $n := n_X + n_Y$\footnote{Notice that this is different from the high-level overview where the last column had index $2n$.}.
Similar to \cref{sec:md}, for each scale $s \in [S]$, we consider the partition of $[n]$ into disjoint intervals $(l, r]$ of length $I_s := n / M ^ {S - s}$, and we say that $(l,r]$ is a \emph{scale-$s$ interval}. Every interval $(l, r]$ that is involved in the recursive partitioning (i.e., is a scale-$s$ interval for some scale $s \in [S]$) is called a \emph{recursion interval}.

Na\"ively, we can design an $M$-way $S$-scale divide-and-conquer scheme to compute the length of the shortest $\vertex{l}{y_l}$-$\vertex{r}{y_r}$ path where $(l, r]$ is a scale-$s$ interval for $s \in [S] \cup \{0\}$\footnote{As an edge case, sometimes we also use consider scale-$0$ intervals which are the sub-intervals of scale-$1$ intervals. We usually ignore scale $0$ unless specified otherwise. }, henceforth notated by $\estnaive[l, y_l, r, y_r]$. If $s = 0$, we calculate $\estnaive[l, y_l, r, y_r]$ using brute force. Suppose that $s > 0$. For $0 \le i \le M$ we say that the $i$-th {\em anchor} is ${\anc}_i := l + i \cdot \frac{r - l}{M}$. Every sequence $\hat{Y} = (y_l  = \hat{y}_0, \hat{y}_1, \ldots, \hat{y}_M = y_r)$ where for $1 \le i \le M$, $\abs{\hat{y}_i - \hat{y}_{i - 1}} \le \frac{r - l}{M}$ defines an \emph{anchor path} visiting vertices $\left(\vertex{{\anc}_i}{\hat{y}_i}\right)_i$. $\estnaive[l, y_l, r, y_r]$ can be found via minimizing over the anchor paths:
\begin{equation*}
\estnaive[l, y_l, r, y_r] = \min\limits_{\hat{Y} \in \anchor}\sum\limits_{i = 1} ^ M{\estnaive[{\anc}_{i - 1}, \hat{y}_{i - 1}, {\anc}_i, \hat{y}_i]}.
\end{equation*}

The idea towards improving the efficiency of the brute force idea above is to only consider a small subset of the possible anchor paths, at the cost of some error.
If we incur an additive error of $I_s / \phi$ for every recursion interval where length-relative factor $\phi = \omega(\log (n))$, we should only incur an additive error of $n / \omega(\log n)$ per scale and the total additive error across $S = O(\log(n))$ scales is $n / \omega(1)$. Consider the following ``hierarchically-sparsified'' grid graph:
\begin{restatable}[Sparsified Edit Distance Grid Graph]{definition}{defgridgraphmin} \label{def:gridgraphmin}
The Sparsified Edit Distance Grid Graph is defined given the following parameters:
\begin{itemize}
    \item Rounding threshold for columns: ${\mathcal{B}}_x := n / M ^ S = \tilde{\Theta} (n ^ {0.99})$;
    \item Length-relative rounding factor for rows: $\phi_y := 2 ^ {\floor{\log ^ {0.009} (n)}}$.
\end{itemize}
\begin{description}
    \item[Vertices] We only consider vertices $\vertex{x}{y}$ such that
    \begin{itemize}
        \item ${\mathcal{B}}_x \mid x$;
        \item Let $s \in [S] \cup \{0\}$ be the largest scale such that $I_s \mid x$. Then $\left(I_s / \phi_y\right) \mid y$. 
    \end{itemize} 
    \item[Edges] There is an edge from $\vertex{x}{y}$ to $\vertex{x + {\mathcal{B}}_x}{y'}$ such that $\abs{y' - y} \le {\mathcal{B}}_x$ with weight $\ew[x, y, x + {\mathcal{B}}_x, y']$ equal to the shortest $\vertex{x}{y}$-$\vertex{x + {\mathcal{B}}_x}{y'}$ path\footnote{Note that both coordinates (and in fact all coordinates involved in our algorithm) are always even, meaning that these vertices are never created from splitting diagonal edges in the original grid.} in the Rotated Edit Distance Grid Graph.
    \item[Objective] Find the shortest path from $\vertex{0}{{n_X}}$ to $\vertex{n_X + {n_Y}}{n_Y}$. 
\end{description}
\end{restatable}

We can identically define the Sparsified Grid Graph for LCS: 
\begin{restatable}[Sparsified LCS Grid Graph]{definition}{defgridgraphmax}\label{def:gridgraphmax}
The Sparsified LCS Grid Graph is defined given the following parameters:
\begin{itemize}
    \item Rounding threshold for columns: ${\mathcal{B}}_x := n / M ^ S = \tilde{\Theta} (n ^ {0.99})$;
    \item Length-relative rounding factor for rows: $\phi_y := 2 ^ {\floor{\log ^ {0.009} (n)}}$.
\end{itemize}
\begin{description}
    \item[Vertices] We only consider vertices $\vertex{x}{y}$ such that
    \begin{itemize}
        \item ${\mathcal{B}}_x \mid x$;
        \item Let $s \in [S] \cup \{0\}$ be the largest scale such that $I_s \mid x$. Then $\left(I_s / \phi_y\right) \mid y$. 
    \end{itemize} 
    \item[Edges] There is an edge from $\vertex{x}{y}$ to $\vertex{x + {\mathcal{B}}_x}{y'}$ such that $\abs{y' - y} \le {\mathcal{B}}_x$ with weight $\ew[x, y, x + {\mathcal{B}}_x, y']$ equal to the longest $\vertex{x}{y}$-$\vertex{x + {\mathcal{B}}_x}{y'}$ path in the Rotated LCS Grid Graph.
    \item[Objective] Find the longest path from $\vertex{0}{{n_X}}$ to $\vertex{n_X + {n_Y}}{n_Y}$. 
\end{description}
\end{restatable}

We show that such sparsifications do not incur too much error: 
\begin{restatable}[Sparsified problems approximates original problems]{lemma}{lemmasparsify} \label{lemma:sparsify}\hfill \\
Let $p$ be the shortest $\vertex{0}{{n_X}}$-$\vertex{n_X + {n_Y}}{n_Y}$ path in the Rotated Edit Distance Grid Graph. There exists a path $\hat{p}$ with $w(\hat{p}) \le w(p) + n / 2 ^ {\Omega\left(\log ^ {0.009} (n)\right)}$ in the Sparsified Edit Distance Grid Graph.

Similarly, let $p$ be the longest $\vertex{0}{{n_X}}$-$\vertex{n_X + {n_Y}}{n_Y}$ path in the Rotated LCS Grid Graph. There exists a path $\hat{p}$ with $w(\hat{p}) \ge w(p) - n / 2 ^ {\Omega\left(\log ^ {0.009} (n)\right)}$ in the Sparsified LCS Grid Graph.
\end{restatable}
Towards proving \cref{lemma:sparsify}. For both the Rotated LCS Grid Graph and the Rotated Edit Distance Grid Graph, we have the following lemma, which says that we can slightly change the endpoints of a path so that the resulting paths differ in weights only after a few columns.
\begin{lemma} [Rounding the endpoints of a path] \hfill \\ \label{lemma:pathdiff}
	Let $p_A$ be a path from any $\vertex{l}{{y_l}}$ to any $\vertex{r}{{y_r}}$ on the Rotated LCS Grid Graph or the Rotated Edit Distance Grid Graph, where $l, y_l, r, y_r$ are all even. For any even $y_l^{\prime}$ and $y_r^{\prime}$with $\abs{y_l^{\prime} - y_r^{\prime}} \le r - l$, we can change $p_A$ into a path $p_B$ from $\vertex{l}{y_l^{\prime}}$ to $\vertex{r}{y_r^{\prime}}$ such that $\abs{w(p_A) - w(p_B)} \le \abs{y_l - y_l^{\prime}} + \abs{y_r - y_r^{\prime}}$. Moreover, there are at most $\abs{y_l - y_l^{\prime}} + \abs{y_r - y_r^{\prime}}$ columns before which $p_A$ and $p_B$ differ in weights.
\end{lemma}
\begin{proof}
\Aviad{TBD: re-read the proof of this lemma after Xiao makes another pass.}
We will give the proof for the Rotated LCS Grid Graph. The proof for the Rotated Edit Distance Grid Graph is identical.

We first argue that we can locally adjust the endpoint of the path by $\pm 2$, losing at most a weight of $+2$. In other words, there is a path of weight $\ge w(p_A) - 2$ from $\vertex{l}{y_l}$ to $\vertex{r}{y_r - 2}$. (The analogous claim also follows by symmetry for paths from $\vertex{l}{y_l}$ to $\vertex{r}{y_r + 2}$, and similarly at the left end.)

Suppose $y_l - (y_r - 2) \le r - l$, we will adjust $p_A$ to get a path ${p_A}'$ from $\vertex{l}{y_l}$ to $\vertex{r}{y_r - 2}$ with a weight difference of at most $2$. Let $I$ be the maximum index such that $p_A(I) - p_A(I - 1) \ge 0$. Such an index exists since $y_l - y_r < r - l$ (otherwise $y_l - (y_r - 2) > r - l$). If $p_A(I) = p_A(I - 1) + 1$, we let ${p_A}'$ be such that
\[
{p_A}'(i) = p_A(i) - 2 \cdot \mathbbm{1}[i \ge I].
\]
Otherwise we have $p_A(I) = p_A(I - 1)$. This would imply that $p_A(I - 2) = p_A(I - 1)$ since $\vertex{I - 1}{p_A(I - 1)}$ is a vertex created by spliting a shortcut edge. We let ${p_A}'$ be such that
\[
{p_A}'(i) = p_A(i) - \mathbbm{1}[i \ge I - 1]- \mathbbm{1}[i \ge I].
\]
Since every edge connecting vertices on different rows has identical weights ($0$ for LCS), both $p_A$ and ${p_A}'$ have identical weights per column after column $I$. Therefore, in either case, the only possible weight differences between $p_A$ and ${p_A}'$ involve at most two edges, which means that the weight difference between $p_A$ and ${p_A}'$ is at most $2$.

To prove the lemma, we now apply this local “$\pm2$” adjustment $|y_l - y_l^{\prime}| / 2 + |y_r - y_r^{\prime}| / 2$ times. (Notice that there always exists an order of the local adjustments where a path exists between every intermediate pair of endpoints.) Each adjustment changes the weight by at most $2$, and affects at most one column. Therefore, the total weight difference is at most $|y_l - y_l^{\prime}| + |y_r - y_r^{\prime}|$, and the number of columns where the path differs is also at most $|y_l - y_l^{\prime}| + |y_r - y_r^{\prime}|$, as claimed.
\end{proof}
Now we prove \cref{lemma:sparsify}:
\lemmasparsify*
\begin{proof}
    From \cref{lemma:pathdiff}, if we round $p$ on $X$ columns so that the rows that $p$ visit on these columns become multiples of $Y$,  we incur a total additive error of at most $X \cdot Y$.

    We will perform rounding in $S$ phases, where for each $s = 0, 1, 2, \ldots, S$ in order, we generate a new path $p_{s}$ that satisfies the rounding condition for columns divisible by $I_s$. We start by setting $p_{-1} := p$.

    To go from $p_{s - 1}$ to $p_{s}$, for every column $x$ that is a multiple of $I_s$, we round $p(x)$ to a multiple of $I_s / \phi_y$. This incurs an error of at most $\left(n / I_s\right) \cdot \left(I_s / \phi_y\right) = n / \phi_y$.

    Thus after $S$ rounds, we incur an error of $S \cdot n / \phi_y = n / 2 ^ {\Omega\left(\log ^ {0.009} (n)\right)}$.
\end{proof}

From now on, unless specified otherwise, vertices, edges and paths refer to those in the Sparsified Grid Graphs. 
For LCS, our algorithm solves a query-based problem on the Sparsified LCS Grid Graph:
\begin{restatable}[Query-Based Approximation for LCS]{lemma}{lemmagridgraphmax} \label{lemma:gridgraphmax}
    Suppose that to know the exact edge weight of any edge from any $\vertex{x}{y}$ to any $\vertex{x + {\mathcal{B}}_x}{y'}$ requires an edge-weight query. We can approximate the problem in \cref{def:gridgraphmax} in $o(n ^ {0.1})$ time with high probability\footnote{i.e., a probability of $1 - 1 / n ^ c$ for some constant $c > 0$.} using $\frac{n ^ 2{\phi_y} ^ 2}{{{\mathcal{B}}_x} ^ 2} / 2 ^ {\Omega\left(\log ^ {0.01} (n)\right)}$ edge-weight queries, up to a multiplicative factor of $\left(1 - o(1)\right)$, assuming that the answer is at least $n / 2 ^ {o\left(\log ^ {0.009} (n)\right)}$.
\end{restatable}

We argue that this is sufficient for approximating LCS:
\begin{lemma} [Reducing Approximation of LCS to \cref{lemma:gridgraphmax}]
    Suppose that we can approximate the problem in \cref{def:gridgraphmax} in $o(n ^ {0.1})$ time with high probability using $\frac{n ^ 2{\phi_y} ^ 2}{{{\mathcal{B}}_x} ^ 2} / 2 ^ {\Omega\left(\log ^ {0.01} (n)\right)}$ edge-weight queries, up to a multiplicative factor of $\left(1 - o(1)\right)$, assuming that the answer is at least $n / 2 ^ {o\left(\log ^ {0.009} (n)\right)}$.
    For any constant $\eps > 0$, there is a randomized algorithm that runs in time $n^2/2^{\Omega\left(\log^{0.009}(n)\right)}$ and, with high probability, $(1 - \eps)$-approximates LCS for two strings with lengths that sum to $n$.
\end{lemma}
\begin{proof}
    Firstly, if the LCS is $n / 2 ^ {\Omega\left(\log ^ {0.009} (n)\right)}$, we can use \cref{lemma:boundedlcs} to compute the answer exactly in $n ^ 2 / 2 ^ {\Omega\left(\log ^ {0.01} (n)\right)}$ time. Thus we assume the answer is at least $n / 2 ^ {o\left(\log ^ {0.009} (n)\right)}$. Due to \cref{lemma:sparsify}, $(1 - o(1))$-approximating the problem on the sparsified graph suffices for $(1 - o(1))$-approximating LCS.

    We run the algorithm in \cref{lemma:gridgraphmax}. To answer an edge weight query, since the Rotated LCS Grid Graph is acyclic and every vertex $\vertex{x'}{y'}$ reachable from $\vertex{x}{y}$ satisfies $\max(x' - x, y' - y) \le {\mathcal{B}}_x$, the longest $\vertex{x}{y}$-$\vertex{x + {\mathcal{B}}_x}{y'}$ path can be found in $O({{\mathcal{B}}_x} ^ 2)$ time using standard techniques on acyclic graphs\footnote{The shortest/longest path on an acyclic graph can be found in linear time by combining topological sorting with dynamic programming. The optimal path to each vertex only depends its predecessors in the topological order. Thus we compute the optimal path to each vertex using dynamic programming in that order.}.  Since the algorithm only uses $\frac{n ^ 2{\phi_y} ^ 2}{{{\mathcal{B}}_x} ^ 2} / 2 ^ {\Omega\left(\log ^ {0.01} (n)\right)}$ edge weight queries. The total running time for this part is 
    \begin{align*}
    \frac{n ^ 2{\phi_y} ^ 2}{{{\mathcal{B}}_x} ^ 2} / 2 ^ {\Omega\left(\log ^ {0.01} (n)\right)} \cdot O({{\mathcal{B}}_x} ^ 2) &= n ^ 2{\phi_y} ^ 2 / 2 ^ {\Omega\left(\log ^ {0.01} (n)\right)} \\
    &= n ^ 2 / 2 ^ {\Omega\left(\log ^ {0.01} (n)\right)}.
    \end{align*}
\end{proof}

For Edit Distance, our algorithm solves a query-based problem on the Sparsified Edit Distance Grid Graph, assuming that we have a constant-factor approximation of every edge weight:
\begin{restatable}[Query-Based Approximation for ED]{lemma}{lemmagridgraphmin} \label{lemma:gridgraphmin}
    Suppose that for every edge from $\vertex{x}{y}$ to $\vertex{x + {\mathcal{B}}_x}{y'}$, we are given (``for free'') a constant-factor approximation $\apprew[x, y, x + {\mathcal{B}}_x, y']$ of the edge weight $\ew[x, y, x + {\mathcal{B}}_x, y']$ (i.e., $\ew[x, y, x + {\mathcal{B}}_x, y'] \le \apprew[x, y, x + {\mathcal{B}}_x, y'] \le 4\ew[x, y, x + {\mathcal{B}}_x, y']$). To know the exact edge weight requires an edge-weight query. Then we can approximate the problem in \cref{def:gridgraphmin} in $o(n ^ {0.1})$ time with high probability using $\frac{n ^ 2{\phi_y} ^ 2}{{{\mathcal{B}}_x} ^ 2} / 2 ^ {\Omega\left(\log ^ {0.01} (n)\right)}$ edge-weight queries, up to a multiplicative factor of $\left(1 + o(1)\right)$, assuming that the answer is at least $n / 2 ^ {o\left(\log ^ {0.009} (n)\right)}$.
\end{restatable}
We argue that this is sufficient for approximating Edit Distance:
\begin{lemma} [Reducing Approximation of Edit Distance to \cref{lemma:gridgraphmin}]
    Assume that given a constant-factor approximation of pairwise distances, with high probability, by querying only $\frac{n ^ 2{\phi_y} ^ 2}{{{\mathcal{B}}_x} ^ 2} / 2 ^ {\Omega\left(\log ^ {0.01} (n)\right)}$ edge weights, we can approximate the problem in \cref{def:gridgraphmin} in $o(n ^ {0.1})$ time up to a multiplicative factor of $\left(1 + o(1)\right)$, assuming that the answer is at least $n / 2 ^ {o\left(\log ^ {0.009} (n)\right)}$.
    For any constant $\eps > 0$, there is a randomized algorithm that runs in time $n^2/2^{\Omega\left(\log^{0.009}(n)\right)}$ and, with high probability, $(1 + \eps)$-approximates Edit Distance for two strings with lengths that sum to $n$.
\end{lemma}
\begin{proof}
    Firstly, if the Edit Distance is $n / 2 ^ {\Omega\left(\log ^ {0.009} (n)\right)}$, we can use Ukkonen’s algorithm (\cref{lemma:bounded}) to compute the answer exactly in $n ^ 2 / 2 ^ {\Omega\left(\log ^ {0.009} (n)\right)}$ time. Thus we assume the answer is at least $n / 2 ^ {o\left(\log ^ {0.009} (n)\right)}$. Due to \cref{lemma:sparsify}, $(1 + o(1))$-approximating the problem on the sparsified graph suffices for $(1 + o(1))$-approximating ED.

    Notice that the weight of an edge from $\vertex{x}{y}$ to $\vertex{x + {\mathcal{B}}_x}{y'}$ is equal to the Edit Distance between a sub-string of either input string. Using \cref{lemma:estimate}, we can obtain a constant-factor approximation of each edge weight of the Sparsified Edit Distance Grid Graph in $O(n ^ {1.61})$ time. Since the number of vertices is $O(n ^ 2 \phi_y / {{\mathcal{B}}_x} ^ 2) = o(n ^ {0.03})$, the number of edges is no more than $o(n ^ {0.06})$.
    Thus the total running time is $o(n ^ {1.61} \cdot n ^ {0.06}) = o(n ^ {1.67})$. 
    
    Now we have the ingredients for the algorithm in \cref{lemma:gridgraphmin}. To answer an edge weight query, since the Rotated Edit Distance Grid Graph is acyclic and every vertex $\vertex{x'}{y'}$ reachable from $\vertex{x}{y}$ satisfies $\max(x' - x, y' - y) \le {\mathcal{B}}_x$, the shortest $\vertex{x}{y}$-$\vertex{x + {\mathcal{B}}_x}{y'}$ path can be found in $O({{\mathcal{B}}_x} ^ 2)$ time using standard techniques on acyclic graphs. Since the algorithm only uses $\frac{n ^ 2{\phi_y} ^ 2}{{{\mathcal{B}}_x} ^ 2} / 2 ^ {\Omega\left(\log ^ {0.01} (n)\right)}$ edge weight queries, the total running time for this part is 
    \begin{align*}
    \frac{n ^ 2{\phi_y} ^ 2}{{{\mathcal{B}}_x} ^ 2} / 2 ^ {\Omega\left(\log ^ {0.01} (n)\right)} \cdot O({{\mathcal{B}}_x} ^ 2) &= n ^ 2{\phi_y} ^ 2 / 2 ^ {\Omega\left(\log ^ {0.01} (n)\right)} \\
    &= n ^ 2 / 2 ^ {\Omega\left(\log ^ {0.01} (n)\right)}.
    \end{align*}
\end{proof}

\begin{table}[h!]
    \caption{Table of notation}
    \renewcommand{\arraystretch}{1.25}{
    \begin{tabular}{c|c|c}
        \hline
        Symbol & What it stands for & Note \\\hline
        $M$ & branching factor &  
        $M := 2 ^ {\floor{\log (\log ^ {\mval} (n))}}$ \\\hline
        $S$ & number of scales & $S := 3 \cdot \floor{\log_M{(n ^ {0.01})}}$\\\hline
        $I_s$ & width of scale $s$ & $I_s := n / M ^ {S - s}$ \\\hline
        ${\mathcal{B}}_x$ & rounding threshold for columns & ${\mathcal{B}}_x := n / M ^ S$\\\hline
        $\phi_y$ & length-relative rounding factor for rows & $\phi_y := 2 ^ {\floor{\log ^ {0.009} (n)}}$\\\hline
        $\tilde{S}$ & set of active scales & 
        \begin{tabular}{c}
        \textbf{ED:} each scale chosen w.p.~$S ^ {-0.98}$ u.a.r. \\\hline
        \textbf{LCS:} each scale in $[S] \cap 3\mathbb{Z}$ chosen w.p.~$S ^ {-0.98}$ u.a.r.
        \end{tabular}\\\hline
        $\eta_{l, r, i}$ & whether to sample $i$-th sub-interval of $(l, r]$ & $\eta_{l, r, 1 \ldots M} \sim \textrm{Uni}(\eta': \eta' \in \{0,1\}^M, \sum \eta'_i = M/2)$\\\hline
        $\ew$ & edge weight of sparsified grid & equal to the length of best path in rotated grid \\\hline
        $\apprew$ & constant approximation of $\ew$ & for ED only \\\hline
        $\est$ & length of the best $\tilde{S}$-regular path & computed via \eqref{eq:est} \\\hline
        $\appr$ & constant approximation of $\est$ & for ED only, computed via \eqref{eq:appr} \\\hline
        $\whest$ & our estimate for $\est$ & computed via \eqref{eq:whest} \\\hline
        ${\anc}_i$ & the $i$-th anchor, given $(l, r]$ in context & ${\anc}_i := l + i \cdot \frac{r - l}{M}$ \\\hline
        $\anchor$ & the set of possible anchor paths $\hat{Y}$ & $\hat{Y} = \left(\hat{y}_i\right)$ defines a path that visits each $\vertex{m_i}{\hat{y}_i}$ \\\hline
        $\overline{y}_i$ & y-coordinate for $i$-th anchor on active scale & $\overline{y}_i := (I_{s - 1} / \phi_y) \cdot \round{\frac{y_l + i \cdot \frac{y_r - y_l}{M}}{I_{s - 1} / \phi_y}}$ \\\hline
    \end{tabular}
    }
    \label{tab:notation}
\end{table}

\section{Algorithm for Edit Distance}
In this section, we present our algorithm for Edit Distance:
\begin{restatable}{theorem}{theomain} \label{theo:mained}
For any constant $\eps > 0$, there exists a randomized algorithm that runs in time 
\[
\frac{n^2}{2^{\log^{\Omega(1)}(n)}}
\]
and, with high probability\footnote{A probability of at least $1 - 1/n$.}, computes a $(1+\eps)$-approximation to the edit distance for two strings with total length $n$.
\end{restatable}

We now focus on the query-based algorithm on the sparsified grid introduced in \cref{sec:sparsified}.
\defgridgraphmin*
\lemmagridgraphmin*

We first revisit the na\"ive recursion scheme we briefly talked about in \cref{sec:sparsified}, but in the context of the sparsified grid graph. 
Let $M := 2 ^ {\floor{\log (\log ^ {\mval} (n))}} = \Theta(\log ^ {\mval} (n))$ and $S := 3 \cdot \floor{\log_M{(n ^ {0.01})}}$.
For each {\em scale} $s \in [S] \cup \{0\}$, we consider the partition of $[n]$ into disjoint intervals of length $I_s := n / M ^ {S - s}$, and we say that each interval in this partition is a scale-$s$ interval. Suppose we want to compute the length of the shortest $\vertex{l}{y_l}$-$\vertex{r}{y_r}$ path where $(l, r]$ is a scale-$s$ interval, denoted by $\ans[l, y_l, r, y_r]$. If $s = 0$, we obtain $\ans[l, y_l, r, y_r]$ by querying the edge weight $\ew[l, y_l, r, y_r]$. 

If $s > 0$, for $0 \le i \le M$, let ${\anc}_i := l + i \cdot \frac{r - l}{M}$ be the $i$-th anchor. Then $\ans[l, y_l, r, y_r]$ can be found via minimizing over every anchor path defined by some sequence $\overline{Y} = (\overline{y}_i)$ that satisfies 
\begin{itemize}
    \item $\overline{y}_0 = y_l$ and $\overline{y}_M = y_r$;
    \item For $1 \le i \le M$, $\abs{\overline{y}_i - \overline{y}_{i - 1}} \le \frac{r - l}{M}$;
    \item For $0 \le i \le M$, $\vertex{{\anc}_i}{\overline{y}_i}$ exists in the Sparsified Grid Graph.
\end{itemize}
The recursion is as follows:
\begin{equation*}
\ans[l, y_l, r, y_r] := 
\begin{cases}
	\ew[l, y_l, r, y_r], & \text{if $s = 0$} \\
	\min\limits_{\hat{Y} \in \anchor}\sum\limits_{i = 1} ^ M{\ans[{\anc}_{i - 1}, \hat{y}_{i - 1}, {\anc}_i, \hat{y}_i]}, & \text{otherwise}
\end{cases}.
\end{equation*}
Unfortunately, the na\"ive recursion scheme will eventually visit every possible $(l, y_l, r, y_r)$, and it will query every edge weight. Thus the total number of edge-weight queries is $O(n ^ 2{\phi_y} ^ 2 / {{\mathcal{B}}_x} ^ 2)$.
Our final algorithm will shave a $2 ^ {\log ^ {\Omega(1)}(n)}$ factor off this query count. 

\subsection{Regular Paths}
We first review the high-level idea for estimating the shortest path on a sparsified grid graph. The bound on tree total deviation roughly tells us that a path should resemble a ``straight line'' on an average interval on an average scale. Thus if we sample a small fraction of scales to be \emph{active}, and force the path to be a straight line on these active scales, we should not lose too much precision. 

Let $\tilde{S} \subset [S]$ be a subset of active scales (other scales are \emph{passive}), we define the notion of $\tilde{S}$-regularity: 
\begin{definition} [$\tilde{S}$-regular paths] \label{def:regular}
    Given $\tilde{S} \subset [S]$, a path $\hat{p}$ is {\em $\tilde{S}$-regular} 
    if for every $s \in \tilde{S}$ and every scale-$s$ interval $(l, r]$, for every $0 \le i \le M$, we have
    $$\hat{p}({\anc}_i) = (I_{s - 1} / \phi_y) \cdot \round{\frac{\hat{p}(l) + i \cdot \frac{\hat{p}(r) - \hat{p}(l)}{M}}{I_{s - 1} / \phi_y}} =: \overline{y}_i.$$
    In other words, inside every active scale interval, an $\tilde{S}$-regular path is the straight line rounded to nearest multiples of $I_{s-1} / \phi_y$ (rounding is necessary due to the way the sparsified grid is defined). 
\end{definition}

For a column $x$, we use the shorthand ``the weight of $p$ before column $x$'' to denote the weight of the edge from $\vertex{x - 1}{p(x - 1)}$ to $\vertex{x}{p(x)}$.
We show that if $\tilde{S} \subset [S]$ is such that each scale is chosen u.a.r.~with probability $1 / S ^ {0.98}$ (thus of size $\Theta(S ^ {0.02})$ with probability $1 - o(1)$), then enforcing $\tilde{S}$-regularity does not increase the answer too much:
\begin{restatable}[Best regular path approximates best path]{lemma}{lemmaregular} \label{lemma:regular}
    Let $\tilde{S} \subset [S]$ be such that each scale is chosen u.a.r.~with probability $1 / S ^ {0.98}$. Let $p$ be the shortest $\vertex{0}{{n_X}}$-$\vertex{n_X + {n_Y}}{n_Y}$ path with $w(p) = n / 2 ^ {o\left(\log ^ {0.01} (n)\right)}$. With probability at least $19 / 20$, we can round $p$ to an $\tilde{S}$-regular $\vertex{0}{{n_X}}$-$\vertex{n_X + {n_Y}}{n_Y}$ path $\hat{p}$ such that $$w(\hat{p}) = \left(1 + o\left(\frac{1}{M\log ^ {0.02} (n)}\right)\right)w(p).$$
    Moreover, there are at most $o\left(\frac{w(p)}{M\log ^ {0.02} (n)}\right)$ columns before which the weights of $p$ and $\hat{p}$ differ.
\end{restatable}

Towards proving \cref{lemma:regular}, we first show that if the sum of total deviation is indeed small on some scale, then straight-line regularization on that scale does not incur too much error. For a path $p$ from $\vertex{l}{y_l}$ to $\vertex{r}{y_r}$, we let the forward difference operator $\Delta p$ be the sequence $$\Delta p := \left(p(l + {\mathcal{B}}_x) - p(l), p(l + 2{\mathcal{B}}_x) - p(l + {\mathcal{B}}_x), \ldots, p(r) - p(r - {\mathcal{B}}_x)\right).$$ 
We show the following claim:
\begin{claim}[Regularizing on a single scale] \label{claim:regularsinglescale}
    Let $s \in [S]$ be some scale. Let $p_s$ be a $\vertex{0}{{n_X}}$-$\vertex{n_X + {n_Y}}{n_Y}$ path. We can round $p_s$ to an $\{s\}$-regular path $p_{s - 1}$. Specifically, for every scale-$s$ interval $(l, r]$, we use \cref{lemma:pathdiff} to adjust $p_s({\anc}_i)$ to $p_{s - 1}({\anc}_i) := (I_{s - 1} / \phi_y) \cdot \round{\frac{p_s(l) + i \cdot \frac{p_s(r) - p_s(l)}{M}}{I_{s - 1} / \phi_y}}$.

    We have
    $$w(p_{s - 1}) \le w(p_s) + O(M\treemd_{M, \{s\}}(\Delta {p_s})).$$
    Moreover, there are at most $O(M\treemd_{M, \{s\}}(\Delta {p_s}))$ columns where the weights of $p_s$ and $p_{s - 1}$ differ.
\end{claim}
\begin{proof}
    It suffices to show that for every scale-$s$ interval $(l, r]$, rounding inside $(l, r]$ affects $O(M\md_M(\Delta p_s(l, r]))$\footnote{As a reminder, this is the notation for the total deviation (not tree total deviation).} weights. Due to \cref{lemma:pathdiff}, it suffices to show that
    \begin{equation} \label{eq:regularsinglescale0}
        \sum_{1 \le i < M}{\abs{p_s({\anc}_i) - (I_{s - 1} / \phi_y) \cdot \round{\frac{p_s(l) + i \cdot \frac{p_s(r) - p_s(l)}{M}}{I_{s - 1} / \phi_y}}}} = O(M\md_M(\Delta p_s(l, r])),
    \end{equation}

    Note that for $a, b \in \mathbb{R}$ and $c \in \mathbb{R} ^ +$, if $c \mid a$ and $\hat{b}$ is $b$ rounded to the closest multiple of $c$, then $\abs{a - \hat{b}} \le 2\abs{a - b}$. For every $i$, since $(I_{s - 1} / \phi_y) \mid p_s({\anc}_i)$,  we must have
    $$\abs{p_s({\anc}_i) - (I_{s - 1} / \phi_y) \cdot \round{\frac{p_s(l) + i \cdot \frac{p_s(r) - p_s(l)}{M}}{I_{s - 1} / \phi_y}}} \le 2\abs{p_s({\anc}_i) - \left(p_s(l) + i \cdot \frac{p_s(r) - p_s(l)}{M}\right)}.$$
    Thus to show \eqref{eq:regularsinglescale0}, it suffices to show that
    \begin{equation} \label{eq:regularsinglescale1}
        \sum_{1 \le i < M}{\abs{p_s({\anc}_i) - \left(p_s(l) + i \cdot \frac{p_s(r) - p_s(l)}{M}\right)}} = O(M\md_M(\Delta p_s(l, r])).
    \end{equation}

    By telescoping, we have
    \begin{equation*}
        \abs{\textcolor{red}{p_s({\anc}_i)} - \left(p_s(l) + i \cdot \frac{p_s(r) - p_s(l)}{M}\right)} = \abs{\textcolor{red}{\left(p_s(l) + \sum_{j \le i}{\left(p_s({\anc}_{j}) - p_s({\anc}_{j - 1})\right)}\right)} - \left(p_s(l) + i \cdot \frac{p_s(r) - p_s(l)}{M}\right)}.
    \end{equation*}
    The two $p_s(l)$ terms cancel out, thus
    \begin{align} \label{eq:regularsinglescale3}
        \abs{p_s({\anc}_i) - \left(p_s(l) + i \cdot \frac{p_s(r) - p_s(l)}{M}\right)} &= \abs{\left(\sum_{1 \le j \le i}{\left(p_s({\anc}_{j}) - p_s({\anc}_{j - 1})\right)}\right) - i \cdot \frac{p_s(r) - p_s(l)}{M}} \nonumber \\
        &= \abs{\sum_{1 \le j \le i}{\left(p_s({\anc}_{j}) - p_s({\anc}_{j - 1}) - \frac{p_s(r) - p_s(l)}{M}\right)}} \nonumber \\
        &\le \sum_{1 \le j \le i}{\abs{p_s({\anc}_{j}) - p_s({\anc}_{j - 1}) - \frac{p_s(r) - p_s(l)}{M}}} \nonumber \\
        &\le \sum_{1 \le j \le M}{\abs{p_s({\anc}_{j}) - p_s({\anc}_{j - 1}) - \frac{p_s(r) - p_s(l)}{M}}} \nonumber \\
        &= \md_M(\Delta p_s(l, r]).
    \end{align}
    Now that we have the bound \eqref{eq:regularsinglescale3}  on every summand we can easily deduce \eqref{eq:regularsinglescale1}: 
    \begin{align*} \label{eq:regularsinglescale1}
        \sum_{1 \le i < M}{\abs{p_s({\anc}_i) - \left(p_s(l) + i \cdot \frac{p_s(r) - p_s(l)}{M}\right)}}   &\le \sum_{1 \le i < M}{\md_M(\Delta p_s(l, r])}  \\
                                                                                                &= O(M\md_M(\Delta p_s(l, r]))
    \end{align*}
\end{proof}

The following ingredient is important towards proving our main regularity lemma:
\begin{lemma}[Deviation and weight bounds for arbitrary paths] \label{lemma:treemd-bounds}
Let $p'$ be any path from $\vertex{0}{n_X}$ to $\vertex{n_X + n_Y}{n_Y}$ on the Sparsified Rotated Edit Distance Grid Graph. Then:
\begin{equation} \label{eq:regulareq2}
    \sum_i \left| (\Delta p')_i \right| \le w(p'),
\end{equation}
and
\begin{equation} \label{eq:regulareq1}
    \treemd_{M, S}(\Delta p') = O(w(p') \cdot S^{0.51}),
\end{equation}
where
\[
\Delta p' := \left(p'({\mathcal{B}}_x) - p'(0),\; p'(2{\mathcal{B}}_x) - p'({\mathcal{B}}_x),\; \ldots,\; p'(n_X + n_Y) - p'(n_X + n_Y - {\mathcal{B}}_x)\right).
\]
\end{lemma}

\begin{proof}
The underlying (unsparsified) Rotated Edit Distance Grid Graph ensures that every vertical move between adjacent rows has weight at least 1. Thus, the total vertical deviation over any subinterval $(l, r]$ satisfies:
\[
\sum_i \left| (\Delta p')_i \right| \le w(p'),
\]
establishing \eqref{eq:regulareq2}.

Recall that $n = n_X + n_Y$. From \cref{lemma:totalmeandeviation}, we also have: 
\[
\treemd_{M, S}(\Delta p') = O\left( \left( \sum_i \left| (\Delta p')_i \right| \cdot S^{0.51} \right) + \left( \underbrace{{\mathcal{B}}_x}_{\text{max absolute value of $\Delta p'$}} \cdot \underbrace{(n / {\mathcal{B}}_x)}_{\text{number of entries in $\Delta p'$}} \cdot \frac{1}{2^{\log^{0.015}(n)}} \right) \right).
\]
The second term is \( o(w(p')) \) since \( w(p') \ge \frac{n}{2^{o(\log^{0.01} n)}} \). Plugging in \eqref{eq:regulareq2}:
\[
\treemd_{M, S}(\Delta p') = O(w(p') \cdot S^{0.51} ) + o(w(p')) = O(w(p') \cdot S^{0.51} ),
\]
which gives \eqref{eq:regulareq1}.
\end{proof}

We now prove our main regularity lemma:
\lemmaregular*
\begin{proof}
    We perform rounding on a scale-by-scale basis. Let $p_S := p$. For each scale $s \in [S]$ in decreasing order, if $s \in \tilde{S}$, we round $p_{s}$ in the manner described in \cref{claim:regularsinglescale} to obtain $p_{s - 1}$, and if $s \notin \tilde{S}$ we simply let $p_{s - 1} := p_{s}$. Let the final path be $\hat{p} := p_0$. This affects $O\left(M\sum_{s \in \tilde{S}}\treemd_{M, \{s\}}{\left(\Delta p_s\right)}\right)$ weights. 

    As a warm-up, suppose that the error is $O\left(M\sum_{s \in \tilde{S}}\treemd_{M, \{s\}}{(\Delta p)}\right)$ (namely $p$ instead of $p_s$). This error is equal to $O(M\treemd_{M, \tilde{S}}{(\Delta p)})$. Since each scale is in $\tilde{S}$ u.a.r.~with probability $1 / S ^ {0.98}$, we must have $\Ex_{\tilde{S}}\left[\treemd_{M, \tilde{S}}{(\Delta p)}\right] = \treemd_{M, S}{(\Delta p)} / S ^ {0.98}$. From \eqref{eq:regulareq1}, this means that $\Ex_{\tilde{S}}\left[M\treemd_{M, \tilde{S}}{(\Delta p)}\right] = O(Mw(p) \cdot S ^ {0.51}) / S ^ {0.98}$. Since $M = \Theta \left(\log ^ {\mval} (n)\right)$ and $S := 3 \cdot \floor{\log_M{(n ^ {0.01})}} = \Theta(\log (n) / \log \log (n))$, $\Ex_{\tilde{S}}\left[M\treemd_{M, \tilde{S}}{(\Delta p)}\right] = o\left(\frac{w(p)}{M\log ^ {0.02} (n)}\right)$, which implies that $M\treemd_{M, \tilde{S}}{(\Delta p)} = o\left(\frac{w(p)}{M\log ^ {0.02} (n)}\right)$ with probability at least $19 / 20$ due to Markov's Inequality.

    A technical hurdle for going from $p$ to $p_s$ is that every $p_s$ changes in response to scales being added to $\tilde{S}$, so we need to be more careful in our analysis. 
    
    Let $\mathcal{X}$ be the total number of columns before which the weights of $p$ and $\hat{p}$ differ. Then to prove the lemma we just need to show $\mathcal{X} = o\left(\frac{w(p)}{M\log ^ {0.02} (n)}\right)$ with probability at least $19 / 20$. Let 
    $$X_k := \begin{cases}
	\treemd_{M, \{S - k + 1\}}(\Delta p_{S - k + 1}), & \text{if $S - k + 1 \in \tilde{S}$} \\
	0, & \text{if $S - k + 1 \notin \tilde{S}$}
 \end{cases},$$
    which captures the additive error introduced on scale $(S - k + 1)$ due to \cref{claim:regularsinglescale} (scaled by roughly a factor of $1 / M$). From \cref{claim:regularsinglescale}, We have $\mathcal{X} = O\left(M \cdot \sum_{k \in [S]}{X_k}\right)$.
    Moreover, since each scale $s$ is added to $\tilde{S}$ with probability $1 / S ^ {0.98}$, the expectation of $X_k$, given $X_{1 \ldots k - 1}$, over whether $(S - k + 1) \in \tilde{S}$, is 
    \begin{equation} \label{eq:regulardefy}
        \Ex_{(S - k + 1) \in \tilde{S}}[X_k \mid X_1, \ldots, X_{k - 1}] = \frac{\treemd_{M, \{S - k + 1\}}(\Delta p_{S - k + 1})}{S ^ {0.98}} =: Y_k.
    \end{equation} 
    Thus the random process $Z_K = \sum_{k \le K}{\left(X_k - Y_k\right)}$ is a martingale over the randomness of $\tilde{S}$ and we have 
    \begin{equation} \label{eq:regularxyequal}
        \Ex_{\tilde{S}}\left[\sum_{k \in [S]}{X_k}\right] = \Ex_{\tilde{S}}\left[\sum_{k \in [S]}{Y_k}\right].
    \end{equation} 
    Now we can write out the following equation for $\Ex_{\tilde{S}}\left[\mathcal{X}\right]$:
    \begin{align} \label{eq:regularwhatpbound}
    \Ex_{\tilde{S}}\left[\mathcal{X}\right]   &= O\left(M \cdot \Ex_{\tilde{S}}\left[\sum_{k \in [S]}{X_k}\right]\right) \nonumber \\
                                        &= O\left(M \cdot \Ex_{\tilde{S}}\left[\sum_{k \in [S]}{Y_k}\right]\right). && \text{from \eqref{eq:regularxyequal}}
    \end{align}
    Recall that our goal is to give an upper bound on $\mathcal{X}$. From \eqref{eq:regularwhatpbound}, an upper bound on $\Ex_{\tilde{S}}\left[\sum_{k \in [S]}{Y_k}\right]$ can give an upper bound on $\Ex_{\tilde{S}}\left[\mathcal{X}\right]$ and we can then use Markov's inequality to convert that to a probabilistic bound. 
    \paragraph{An upper bound on $\Ex_{\tilde{S}}\left[\sum_{k \in [S]}{Y_k}\right]$} 
    For any fixed $\tilde{S}$, note that \begin{equation} \label{eq:regularysum}
        \sum_{k \in [S]}{Y_k} := \frac{1}{S ^ {0.98}}\sum_{s \in [S]}\treemd_{M, \{s\}}{\left(\Delta p_s\right)}.
    \end{equation}
    Unfortunately, $\sum_{s \in [S]}\treemd_{M, \{s\}}{\left(\Delta p_s\right)}$ is not equal to the total tree deviation of any particular sequence. To get around this, note that since there are only $\abs{\tilde{S}}$ active scales, the number of distinct $p_s$ is at most $\abs{\tilde{S}}$. Therefore, $\sum_{s \in [S]}\treemd_{M, \{s\}}{\left(\Delta p_s\right)}$ is at most the sum of the tree total deviations of $\Delta p'$ for $\abs{\tilde{S}}$ distinct paths $p'$. Due to \eqref{eq:regulareq1}, this is the sum of $O(w(p') \cdot S^{0.51})$ for $\abs{\tilde{S}}$ distinct paths $p'$.
    Let $s' \in \tilde{S}$ be the maximizer of $w(p_{s'})$. We have
    $$\sum_{s \in [S]}\treemd_{M, \{s\}}{(\Delta p_s)} \le O(\abs{\tilde{S}}w(p_{s'})S^{0.51}).$$
    Plugging into \eqref{eq:regularysum}, we have
    \begin{equation*}
        \sum_{k \in [S]}{Y_k} = O(\abs{\tilde{S}}w(p_{s'})S^{0.51} / S ^ {0.98}).
    \end{equation*}
    With probability $1 - o(1)$, $\abs{\tilde{S}} = O(S ^ {0.02})$, and we have\begin{equation*}
        \sum_{k \in [S]}{Y_k} =  O(w(p_{s'}) / S ^ {0.45}).
    \end{equation*}
    Recall that $M = \Theta \left(\log ^ {\mval} (n)\right)$ and $S = \Theta(\log (n) / \log \log (n))$, and therefore with probability $1 - o(1)$,
    \begin{equation*}
    \sum_{k \in [S]}{Y_k} = O(w(p_{s'}) / S ^ {0.45}) = o\left(\frac{w(p_{s'})}{M^2\log ^ {0.02} (n)}\right).
    \end{equation*}
    Thus if we consider the expections over the randomness of $\tilde{S}$, with probability $1 - o(1)$,
    \begin{equation} \label{eq:regulareqbound}
    \Ex_{\tilde{S}}\left[\sum_{k \in [S]}{Y_k}\right] = o\left(\frac{\Ex_{\tilde{S}}\left[w(p_{s'})\right]}{M^2\log ^ {0.02} (n)}\right).
    \end{equation}
    
    \paragraph{Finishing it off} 
    By plugging \eqref{eq:regulareqbound} into \eqref{eq:regularwhatpbound}, with probability $1 - o(1)$, we have
    $$\Ex_{\tilde{S}}\left[\mathcal{X}\right] = o\left(\frac{\Ex_{\tilde{S}}\left[w(p_{s'})\right]}{M\log ^ {0.02} (n)}\right).$$
    Further note that $w(p_{s'}) \le w(p) + \mathcal{X}$ since modifying an edge increases the length of a path by at most $1$. Thus we have
    $$\Ex_{\tilde{S}}\left[\mathcal{X}\right] = o\left(\frac{w(p) +\Ex_{\tilde{S}}\left[\mathcal{X}\right]}{M\log ^ {0.02} (n)}\right).$$
    Solving for $\Ex_{\tilde{S}}\left[\mathcal{X})\right]$, we get $\Ex_{\tilde{S}}\left[\mathcal{X}\right] = o\left(\frac{w(p)}{M\log ^ {0.02} (n)}\right)$. The lemma follows from Markov's Inequality.
\end{proof}

In \cref{sec:subsampling}, we will present how we can approximate the length of $\hat{p}$ with a small multiplicative error with probability at least $3 / 4$:
\begin{restatable}{lemma}{lemmasubsampling} \label{lemma:subsampling}
Let $\tilde{S} \subset [S]$ be such that each scale is chosen u.a.r.~with probability $1 / S ^ {0.98}$. Let $\hat{p}$ be the shortest $\tilde{S}$-regular $\vertex{0}{n_X}$-$\vertex{n_X + n_Y}{n_Y}$ path. With probability at least $3 / 4$, we can $(1 \pm o(1))$-approximate $w(\hat{p})$ in $o(n ^ {0.09})$ time by querying only $\frac{n ^ 2{\phi_y} ^ 2}{{{\mathcal{B}}_x} ^ 2} / 2 ^ {\Omega\left(\log ^ {0.01} (n)\right)}$ edge weights.
\end{restatable}

Assuming~\cref{lemma:subsampling}, we now show how to finish the rest of the proof for \cref{lemma:gridgraphmin}. 
\lemmagridgraphmin*
\begin{proof}
    Let $\tilde{S} \subset [S]$ be such that each scale is chosen u.a.r.~with probability $1 / S ^ {0.98}$. Let $p$ be the shortest $\vertex{0}{{n_X}}$-$\vertex{n_X + {n_Y}}{n_Y}$ path with $w(p) = n / 2 ^ {o\left(\log ^ {0.01} (n)\right)}$. Let $\hat{p}$ be the shortest $\tilde{S}$-regular $\vertex{0}{{n_X}}$-$\vertex{n_X + {n_Y}}{n_Y}$ path. 
 
    From \cref{lemma:regular}, with probability at least $19 / 20$, we have 
    \begin{equation*}
        w(\hat{p}) > \left(1 + o(1)\right)w(p).
    \end{equation*} 
    
    From \cref{lemma:subsampling}, with probability at least $3 / 4$, we can obtain $(1 \pm o(1))$-approximate $w(\hat{p})$ by querying only $\frac{n ^ 2\phi_y}{{{\mathcal{B}}_x}^2} / 2 ^ {\Omega\left(\log ^ {0.01} (n)\right)}$ edge weights. By scaling up the answer by $(1 + o(1))$ this becomes a $(1 + o(1))$-approximation.

    By the union bound, with probability at least $7 / 10$, we can $\left(1 + o(1)\right)$-approximate $w(p)$ in $o(n ^ {0.09})$ time by querying only $\frac{n ^ 2{\phi_y} ^ 2}{{{\mathcal{B}}_x} ^ 2} / 2 ^ {\Omega\left(\log ^ {0.01} (n)\right)}$ edge weights. By repeating the procedure $\Theta(\log (n))$ times and taking the median, we can approximate $w(p)$ with a multiplicative error of $\left(1 + o(1)\right)$ with high probability in $o(n ^ {0.09} \log (n)) = o(n ^ {0.1})$ time by querying only $\frac{n ^ 2{\phi_y} ^ 2\log (n)}{{{\mathcal{B}}_x} ^ 2} / 2 ^ {\Omega\left(\log ^ {0.01} (n)\right)} = \frac{n ^ 2{\phi_y} ^ 2}{{{\mathcal{B}}_x} ^ 2} / 2 ^ {\Omega\left(\log ^ {0.01} (n)\right)}$ edge weights. 
\end{proof}

\subsection{Algorithm Description} \label{sec:subsampling}
Let $\tilde{S} \subset [S]$ be such that each scale is chosen u.a.r.~with probability $1 / S ^ {0.98}$. Our goal now is to approximate the length of the shortest $\tilde{S}$-regular $\vertex{0}{n_X}$-$\vertex{n_X + n_Y}{n_Y}$ path $\hat{p}$:
\lemmasubsampling*
We can design an $M$-way $S$-scale divide-and-conquer scheme to compute $\est[l, y_l, r, y_r]$ which is the length of the shortest $\tilde{S}$-regular $\vertex{l}{y_l}$-$\vertex{r}{y_r}$ path where $(l, r]$ is a scale-$s$ interval. 
\begin{itemize}
    \item If $s = 0$, we obtain $\est[l, y_l, r, y_r]$ using an edge weight query;
    \item Suppose that $s > 0$ and $s \notin \tilde{S}$. Recall that the $i$-th anchor is ${\anc}_i := l + i \cdot \frac{r - l}{M}$. We find $\est[l, y_l, r, y_r]$ by minimizing over every anchor path defined by the sequence $\hat{Y} = (y_l  = \hat{y}_0, \hat{y}_1, \ldots, \hat{y}_M = y_r)$, where for $1 \le i \le M$, $\abs{\hat{y}_i - \hat{y}_{i - 1}} \le \frac{r - l}{M}$, visiting valid vertices $\left(\vertex{{\anc}_i}{\hat{y}_i}\right)_i$ in the sparsified grid graph. We use $\anchor$ to denote the set of possible anchor paths $\hat{Y}$.
    \item Suppose $s \in \tilde{S}$. Recall that we regularize the row after column ${\anc}_i$ to $\overline{y}_i := (I_{s - 1} / \phi_y) \cdot \round{\frac{y_l + i \cdot \frac{y_r - y_l}{M}}{I_{s - 1} / \phi_y}}$.
    $\est[l, y_l, r, y_r]$ is simply the sum of $\est[{\anc}_{i - 1}, \overline{y}_i, {\anc}_i, \overline{y}_i]$ over $1 \le i \le M$.  
\end{itemize} 
The full recursion is as follows:
\begin{equation} \label{eq:est}
\est[l, y_l, r, y_r] := 
\begin{cases}
	\ew[l, y_l, r, y_r], & \text{if $s = 0$} \\
	\min\limits_{\hat{Y} \in \anchor}\sum\limits_{i = 1} ^ M{\est[{\anc}_{i - 1}, \hat{y}_{i - 1}, {\anc}_i, \hat{y}_i]}, & \text{if $s \notin \tilde{S}$} \\
	\sum\limits_{i = 1} ^ M{\est[{\anc}_{i - 1}, \overline{y}_i, {\anc}_i, \overline{y}_i]}, & \text{if $s \in \tilde{S}$}
\end{cases}.
\end{equation}
Unfortunately, this scheme still queries too many edge weights. 

The key to efficiency is to use a \emph{sub-sampling} technique for active scales. We first show a naive 
attempt that can guarantee a small additive error (we later improve it to also obtain a bound on the multiplicative error). Notice that if $s \in \tilde{S}$, since no overfitting can happen when the path is fixed, we can just sample half of these values and double our sum from the sample. Thus, for every scale-$s$ interval $(l, r]$, we first uniformly randomly sample $M$ 0-1 variables $\eta_{l, r, 1 \ldots M}$ containing exactly $M / 2$ 1's. We use the following scheme:
\begin{equation} \label{eq:estasterisk}
\whest ^ *[l, y_l, r, y_r] := 
\begin{cases}
	\ew[l, y_l, r, y_r], & \text{if $s = 0$} \\
	\min\limits_{\hat{Y} \in \anchor}\sum\limits_{i = 1} ^ M{\whest ^ *[{\anc}_{i - 1}, \hat{y}_{i - 1}, {\anc}_i, \hat{y}_i]}, & \text{if $s \notin \tilde{S}$} \\
	2\sum\limits_{i = 1} ^ M{\eta_{l, r, i}\whest ^ *[{\anc}_{i - 1}, \overline{y}_i, {\anc}_i, \overline{y}_i]}, & \text{if $s \in \tilde{S}$}
\end{cases},
\end{equation}
Note that for every active scale, half of the sub-intervals $({\anc}_{i - 1}, {\anc}_i]$ for which $\eta_{l, r, i} = 0$ get discarded. Over $\abs{\tilde S}$ active scales, only $2 ^ {-\abs{\tilde S}}$ fraction of the intervals remain, which is a significant speed-up.

Unfortunately, the scheme with \eqref{eq:estasterisk} is not enough for a small multiplicative error, as $\whest ^ *[{\anc}_{i - 1}, \overline{y}_{i - 1}, {\anc}_i, \overline{y}_i]$ may be imbalanced over $1 \le i \le M$. In the most extreme example, if $\whest ^ *[{\anc}_{i - 1}, \overline{y}_{i - 1}, {\anc}_i, \overline{y}_i] \ne 0$ for only a single \emph{outlier} $i$, then \eqref{eq:estasterisk} either gets us an estimate that is $0$ or an estimate that is twice the actual value, depending on the value of $\eta_{l, r, i}$. 

Due to the bound on tree total deviation, the weight distribution on the shortest path should intuitively not be imbalanced everywhere and there are not too many outliers in total. Thus even if we ``give up'' when we detect outliers, there should not be too much error. We first need to detect outliers. Recall that for each edge from $\vertex{x}{y}$ to $\vertex{x + {\mathcal{B}}_x}{y'}$, we are given a constant factor approximation $\apprew[x, y, x + {\mathcal{B}}_x, y']$ of the edge weight $\ew[x, y, x + {\mathcal{B}}_x, y']$. We can (detailed later) use this to obtain a constant-factor approximation $\appr[l', {y'}_{l'}, r', {y'}_{r'}]$ of the length of the shortest $\vertex{l'}{{y'}_{l'}}$-$\vertex{r'}{{y'}_{r'}}$ $\tilde{S}$-regular path for every valid $\left(l', {y'}_{l'}, r', {y'}_{r'}\right)$ where $(l', r']$ is a recursion interval. We now describe how we can obtain our estimate:
\begin{enumerate}
    \item We detect outliers by checking if $\max\limits_{1 \le i \le M}{\appr[{\anc}_{i - 1}, \overline{y}_{i - 1}, {\anc}_i, \overline{y}_i]} > \frac{\appr[l, y_l, r, y_r]}{M}\log \log(n)$.
    \item If $\max\limits_{1 \le i \le M}{\appr[{\anc}_{i - 1}, \overline{y}_{i - 1}, {\anc}_i, \overline{y}_i]} \le \frac{\appr[l, y_l, r, y_r]}{M}\log \log(n)$, we let $$\whest[l, y_l, r, y_r] := \min\left(2\sum\limits_{i = 1} ^ M{\eta_{l, r, i}\whest[{\anc}_{i - 1}, \overline{y}_i, {\anc}_i, \overline{y}_i]}, \appr[l, y_l, r, y_r]\right).$$ Namely we use the same formula as in \eqref{eq:estasterisk}, but clamp $\whest[l, y_l, r, y_r]$ to at most $\appr[l, y_l, r, y_r]$ to make sure that $\whest[l, y_l, r, y_r]$ does not over-estimate the answer worse than the constant-factor approximation does;
    \item If $\max\limits_{1 \le i \le M}{\appr[{\anc}_{i - 1}, \overline{y}_{i - 1}, {\anc}_i, \overline{y}_i]} > \frac{\appr[l, y_l, r, y_r]}{M}\log \log(n)$, we simply ``give up'' and let $\whest[l, y_l, r, y_r] := \appr[l, y_l, r, y_r]$, which still only gives us a multiplicative error of $\left(1 + O(1)\right)$.
\end{enumerate}
Given a path $p'$ from column $l$ to column $r$, we let $W(p')$ be the following sequence capturing the ``weight distribution'' of $p'$:
$$W(p') = \left(w(p'(l, l + {\mathcal{B}}_x]), w(p'(l + {\mathcal{B}}_x, l + 2{\mathcal{B}}_x]), \ldots, w(p'(r - {\mathcal{B}}_x, r])\right)$$
We later show that this procedure satisfies the following claim:
\begin{restatable}{claim}{claimonesideded} \label{claim:onesideded}
    For $\whest[l, y_l, r, y_r]$ where $(l, r]$ is a scale-$s$ interval such that $s \in \tilde{S}$. 
    Let $\overline{p}_{l, r}$ be the shortest $\tilde{S}$-regular $\vertex{l}{y_l}$-$\vertex{r}{y_r}$ path. We have the following:
    \begin{description}
        \item[We only give up when weight distribution is imbalanced] if $\max\limits_{1 \le i \le M}{\appr[{\anc}_{i - 1}, \overline{y}_{i - 1}, {\anc}_i, \overline{y}_i]} = \omega\left(\frac{\appr[l, y_l, r, y_r]}{M}\right)$, then $\md_M(W(\overline{p}_{l, r})) = \omega\left(\frac{w(\overline{p}_{l, r})}{M}\right)$. 
        \item[If we do not give up, we obtain a good estimate] if $\max\limits_{1 \le i \le M}{\appr[{\anc}_{i - 1}, \overline{y}_{i - 1}, {\anc}_i, \overline{y}_i]}  \le \frac{\appr[l, y_l, r, y_r]}{M}\log \log (n)$, then with probability $1 - \exp\left(-\omega\left(\log ^ {0.02}(n)\right)\right)$, $\left(2\sum\limits_{i = 1} ^ M{\eta_{l, r, i}\whest[{\anc}_{i - 1}, \overline{y}_i, {\anc}_i, \overline{y}_i]}\right)$ approximates $\left(\sum\limits_{i = 1} ^ M{\whest[{\anc}_{i - 1}, \overline{y}_i, {\anc}_i, \overline{y}_i]}\right)$ with additive error $o\left(\frac{w(\overline{p}_{l, r})}{\log ^ {0.02} (n)}\right)$.
    \end{description}
\end{restatable}
To sum up, we compute the estimate in the following way:
\begin{equation} \label{eq:whest} 
\whest[l, y_l, r, y_r] := 
\begin{cases}
	\ew[l, y_l, r, y_r], & \text{if $s = 0$} \\
	\min\limits_{\hat{Y} \in \anchor}\sum\limits_{i = 1} ^ M{\whest[{\anc}_{i - 1}, \hat{y}_{i - 1}, {\anc}_i, \hat{y}_i]}, & \text{if $0< s \notin \tilde{S}$} \\
    \min{\left(
    \begin{tabular}{c}
    $2\sum\limits_{i = 1} ^ M{\eta_{l, r, i}\whest[{\anc}_{i - 1}, \overline{y}_i, {\anc}_i, \overline{y}_i]},$ \\
     $\appr[l, y_l, r, y_r]$
    \end{tabular}
    \right)}
        , & \text{if $\left(
    \begin{tabular}{c}
    $\left(s \in \tilde{S}\right)~\wedge$ \\
    $\left(\begin{tabular}{c}
     $\max\limits_{1 \le i \le M}{\appr[{\anc}_{i - 1}, \overline{y}_{i - 1}, {\anc}_i, \overline{y}_i]}$ \\
     $ \le \frac{\appr[l, y_l, r, y_r]}{M}\log \log (n)$
    \end{tabular}\right)$
    \end{tabular}\right)$} \\
        \appr[l, y_l, r, y_r]
        , & \text{if $\left(
    \begin{tabular}{c}
    $\left(s \in \tilde{S}\right)~\wedge$ \\
    $\left(\begin{tabular}{c}
     $\max\limits_{1 \le i \le M}{\appr[{\anc}_{i - 1}, \overline{y}_{i - 1}, {\anc}_i, \overline{y}_i]}$ \\
     $ > \frac{\appr[l, y_l, r, y_r]}{M}\log \log (n)$
    \end{tabular}\right)$
    \end{tabular}\right)$}
\end{cases},
\end{equation}

Finally, to obtain $\appr[l, y_l, r, y_r]$ which is a constant-factor approximation of the length of the shortest $\tilde{S}$-regular $\vertex{l}{y_l}$-$\vertex{r}{y_r}$ path where $(l, r]$ is a scale-$s$ interval, we can use the following recursion:
\begin{equation} \label{eq:appr}
\appr[l, y_l, r, y_r] := 
\begin{cases}
	\apprew[l, y_l, r, y_r], & \text{if $s = 0$} \\
	\min\limits_{\hat{Y} \in \anchor}\sum\limits_{i = 1} ^ M{\appr[{\anc}_{i - 1}, \hat{y}_{i - 1}, {\anc}_i, \hat{y}_i]}, & \text{if $s \notin \tilde{S}$} \\
	\sum\limits_{i = 1} ^ M{\appr[{\anc}_{i - 1}, \overline{y}_i, {\anc}_i, \overline{y}_i]}, & \text{if $s \in \tilde{S}$}
\end{cases},
\end{equation}

To efficiently minimize over anchor path $\hat{Y}$ on passive scale $s$, consider the following ``Scale-$(s - 1)$ Sparsified Grid Graph'':
\begin{definition} [Scale-$(s - 1)$ Sparsified Grid Graph] \label{def:scalesparse}
    The Scale-$(s - 1)$-Sparsified Grid Graph defined on $\whest$ is as follows:
    \begin{description}
        \item[Vertices] We only consider vertices $\vertex{x}{y}$ in the Sparsified Edit Distance Grid Graph with the additional constraint that $I_{s - 1} \mid x$;
        \item[Edges] There is an edge from $\vertex{x}{y}$ to $\vertex{x + I_{s - 1}}{y'}$ (if both exist) such that $\abs{y' - y} \le I_{s - 1}$ with a weight of $\whest[x, y, x + I_{s - 1}, y']$ (or $\appr[x, y, x + I_{s - 1}, y']$ for computing $\appr$).
    \end{description}
\end{definition} 
To obtain $\min\limits_{\hat{Y} \in \anchor}\sum\limits_{i = 1} ^ M{\whest[{\anc}_{i - 1}, \hat{y}_{i - 1}, {\anc}_i, \hat{y}_i]}$ (or  $\min\limits_{\hat{Y} \in \anchor}\sum\limits_{i = 1} ^ M{\appr[{\anc}_{i - 1}, \hat{y}_{i - 1}, {\anc}_i, \hat{y}_i]}$), we can simply find the shortest $\vertex{l}{y_l}$-$\vertex{r}{y_r}$ path. 

We give the full pseudocode for computing $\whest$ in \cref{alg:whest}. We use the memoization technique to compute $\whest$ values on demand (i.e., Line \ref{line:memoizationstart}-\ref{line:memoizationend} and Line \ref{line:memoizationstart2}-\ref{line:memoizationend2}).
\begin{algorithm} 
    \caption{Computation of $\whest$ for Edit Distance} \label{alg:whest}
    \begin{algorithmic}[1]
        \Function{$\COMESTED$}{$l, y_l, r, y_r, s$}
            \Statex \textbullet~\textbf{requirement 1:} $\vertex{l}{y_l}$ and $\vertex{r}{y_r}$ are both vertices in the sparsified grid graph
            \Statex \textbullet~\textbf{requirement 2:} $\vertex{r}{y_r}$ is reachable from $\vertex{l}{y_l}$
            \Statex \textbullet~\textbf{requirement 3:} $s \in [S] \cup \{0\}$
            \Statex \textbullet~\textbf{requirement 4:} $(l, r]$ is a scale-$s$ interval
            \Statex \textbullet~\textbf{returns:} an estimate $\whest[l, y_l, r, y_r]$ for the shortest $\vertex{l}{y_l}$-$\vertex{r}{y_r}$ path
            \If {$s = 0$}
                \Return $\ew[l, y_l, r, y_r]$
            \ElsIf {$s \notin \tilde{S}$}
                \State $\ver \gets \emptyset $\Comment{Initialize vertices for Scale-$(s - 1)$ Sparsified Grid Graph}
                \State $\e \gets \emptyset $\Comment{Initialize edges for Scale-$(s - 1)$ Sparsified Grid Graph}
                \For {$x \in $ multiples of $I_{s - 1}$ in $(l, r]$}
                    \For {$y \in $ multiples of $I_{s - 1} / \phi_y$ in $[0, n]$}
                        \State add $\vertex{x}{y}$ to $\ver$
                        \If {$x + I_{s - 1} \le r$}
                            \For {$y' \in $ multiples of $I_{s - 1} / \phi_y$ in $[0, n]$}
                                \If {$\abs{y' - y} \le I_{s - 1}$}
                                    \If {$\whest[x, y, x + I_{s - 1}, y'] = \perp$} \Comment{Estimate not yet computed} \label{line:memoizationstart}
                                        \State $\whest[x, y, x + I_{s - 1}, y'] \gets \COMESTED\left(x, y, x + I_{s - 1}, y', s - 1\right)$ \Comment{Compute estimate for sub-interval}
                                    \EndIf \label{line:memoizationend}
                                    \State add an edge from $\vertex{x}{y}$ to $\vertex{x + I_{s - 1}}{y'}$ with weight $\whest[x, y, x + I_{s - 1}, y']$ to $\e$
                                \EndIf
                            \EndFor
                        \EndIf
                    \EndFor
                \EndFor
                \State $G \gets \langle \ver, \e \rangle$ \Comment{Create the Scale-$(s - 1)$ Sparsified Grid Graph}
                \State \Return length of shortest $\vertex{l}{y_l}$-$\vertex{r}{y_r}$ path in $G$ via standard techniques for acyclic graphs
            \Else
                \State $\overline{y}_i \gets (I_{s - 1} / \phi_y) \cdot \round{\frac{y_l + i \cdot \frac{y_r - y_l}{M}}{I_{s - 1} / \phi_y}}$ for $1 \le i \le M$
                \State ${\anc}_i \gets l + i \cdot \frac{r - l}{M}$ for $1 \le i \le M$
                \For {$i \in [1, M]$}  
                    \If {$\eta_{l, r, i} = 1$} \Comment{$i$-th sub-interval is in sample}
                        \If {$\whest[{\anc}_{i - 1}, \overline{y}_{i - 1}, {\anc}_i, \overline{y}_{i}] = \perp$} \Comment{Estimate not yet computed} \label{line:memoizationstart2}
                            \State $\whest[{\anc}_{i - 1}, \overline{y}_{i - 1}, {\anc}_i, \overline{y}_{i}] \gets \COMESTED\left({\anc}_{i - 1}, \overline{y}_{i - 1}, {\anc}_i, \overline{y}_i, s - 1\right)$ \Comment{Compute estimate for sub-interval}
                        \EndIf \label{line:memoizationend2}
                    \EndIf
                \EndFor
                \If {$\max\limits_{1 \le i \le M}{\appr[{\anc}_{i - 1}, \overline{y}_{i - 1}, {\anc}_i, \overline{y}_i]} \le \frac{\appr[l, y_l, r, y_r]}{M}\log \log (n)$}
                    \State \Return $\min\left(2\sum\limits_{i = 1} ^ M{\eta_{l, r, i}\whest[{\anc}_{i - 1}, \overline{y}_i, {\anc}_i, \overline{y}_i]}, \appr[l, y_l, r, y_r]\right)$
                \Else
                    \State \Return $\appr[l, y_l, r, y_r]$
                \EndIf
            \EndIf
        \EndFunction
    \end{algorithmic}
\end{algorithm}

\subsection{Algorithm Analysis} \label{sec:analysised}
\subsubsection{Overview}
Our goal is to use our algorithm to prove \cref{lemma:subsampling}:
\lemmasubsampling*
The analysis of our algorithm proof consists of three parts:
\begin{restatable}[Our estimate is complete]{lemma}{lemmacompleteed}
    Let $\hat{p}$ be the shortest $\tilde{S}$-regular $\vertex{0}{n_X}$-$\vertex{n_X + n_Y}{n_Y}$ path.
     With probability at least $19 / 25$, $\whest[0, n_X, n_X + n_Y, n_Y] < \left(1 + o(1)\right)w(\hat{p})$.
\end{restatable}
\begin{restatable}[Our estimate is sound]{lemma}{lemmasounded} 
    Let $\hat{p}$ be the shortest $\tilde{S}$-regular $\vertex{0}{n_X}$-$\vertex{n_X + n_Y}{n_Y}$ path. Suppose $w(\hat{p}) = n / 2 ^ {o\left(\log ^ {0.01} (n)\right)}$.
    With probability $1 - o(1)$, $\whest[0, n_X, n_X + n_Y, n_Y] > \left(1 - o(1)\right)w(\hat{p})$;
\end{restatable}
\begin{restatable}[Our estimates can be computed efficiently]{lemma}{lemmaefficiencyed} 
    Computation of all $\appr$ values take $o(n ^ {0.09})$ time.
    With probability $1 - o(1)$, to compute $\whest[0, n_X, n_X + n_Y, n_Y]$, the running time is at most $o(n ^ {0.09})$, and the total number of edge weight queries is at most $\frac{n ^ 2{\phi_y} ^ 2}{{{\mathcal{B}}_x} ^ 2} / 2 ^ {\Omega\left(\log ^ {0.01} (n)\right)}$.
\end{restatable}
By union bound, we can see that our algorithm is both correct and efficient with probability at least $3 / 4$, thereby proving \cref{lemma:subsampling}.
\subsubsection{Our Estimate Is Complete}
We first show that our algorithm is complete. Recall that our sub-routine for obtaining the estimate has an error that is contingent on there being no outliers (\cref{claim:onesideded}). To show that outliers are rare on average, we will use the following ingredient, which shows that the deviation of the weight distribution of $\hat{p}$ is small on active layers:
\begin{restatable}[The deviation of the weight distribution of $\hat{p}$ is small on active layers]{claim}{claimtechnicalityed} \label{claim:technicalityed}    
    Assume that the length of the shortest $\vertex{0}{n_X}$-$\vertex{n_X + n_Y}{n_Y}$ path $p$ is at least $n / 2 ^ {o\left(\log ^ {0.02} (n)\right)}$.
    Let $\tilde{S} \subset [S]$ be such that each scale is chosen u.a.r.~with probability $1 / S ^ {0.98}$. Let the shortest $\tilde{S}$-regular $\vertex{0}{n_X}$-$\vertex{n_X + n_Y}{n_Y}$ path be $\hat{p}$.
    With probability at least $4 / 5$, we have $\treemd_{M, \tilde{S}}(W(\hat{p})) = o\left(\frac{w(\hat{p})}{M}\right)$. 
\end{restatable}

Before proving the ingredient above, we first give our main analysis that our estimate is complete:
\lemmacompleteed*
\begin{proof}
    For each active scale $s$, we consider the following question: how much would we overestimate if there were no subsampling on active scales larger than $s$? We can see that the final estimate is at most:
    $$\mathcal{P}_s := \sum\limits_{i = 1} ^ {n / I_s}{\whest[(i - 1) \cdot I_s, \hat{p}((i - 1) \cdot I_s), i \cdot I_s, \hat{p}(i \cdot I_s)]}.$$
    We can see that $\mathcal{P}_0 = w(\hat{p})$, and that $\mathcal{P}_S = \whest[0, n_X, n_X + n_Y, n_Y]$, namely our final estimate for $w(\hat{p})$. To see how much larger $\mathcal{P}_S$ can be compared to $\mathcal{P}_0$, we go through each scale $s$ and examine the difference between $\mathcal{P}_{s - 1}$ and $\mathcal{P}_s$. 

    If scale-$s$ is passive, then since we compute every $\whest[l, \hat{p}(l), r, \hat{p}(r)]$ to be the length of the shortest $\tilde{S}$-regular $\vertex{l}{\hat{p}(l)}$-$\vertex{r}{\hat{p}(r)}$ path in the Scale-$(s - 1)$ Sparsified Grid Graph, $\mathcal{P}_s \le \mathcal{P}_{s - 1}$.  
    
    If scale-$s$ is active, fix a scale-$s$ interval $(l, r]$. From \cref{claim:onesideded}, the magnitude of overestimation is bounded by the sum of the following:
    \begin{itemize}
        \item Type I: If $\md_M(W(\hat{p}(l, r]]))   = \omega\left(\frac{w(\hat{p}(l, r])}{M}\right)$, an additive error of $O(w(\hat{p}(l, r]))$;
        \item Type II: An additive error of $o\left(\frac{w(\hat{p}(l, r])}{\log ^ {0.02} (n)}\right))$;
        \item Type III: With probability $\exp\left(-\omega\left(\log ^ {0.02}(n)\right)\right)$, an additive error of $O(w(\hat{p}(l, r]))$.
    \end{itemize}

    The total sum of type III error is at most $\exp\left(-\omega\left(\log ^ {0.02}(n)\right)\right)$ times the total sum of $\hat{p}(l, r]$ on every interval $(l, r]$ on an active scale. The total sum of $\hat{p}(l, r]$ is at most $O(w(\hat{p})\abs{\tilde{S}})$ and the total total sum of type III error is at most $w(\hat{p})\abs{\tilde{S}}\exp\left(-\omega\left(\log ^ {0.02}(n)\right)\right) = o(w(\hat{p}))$. 

    By summing over every interval in each of the $\abs{\tilde{S}}$ active scales (which is $O(S ^ {0.02}) = o\left(\log ^ {0.02} (n)\right)$ with probability $1 - o(1)$), the total sum of type II error is $o(w(\hat{p}))$.
    
    Finally, since $O(w(\hat{p}(l, r])) = o(M) \cdot \omega\left(\frac{w(\hat{p}(l, r])}{M}\right)$, the total sum of type I error is $o(M)$ times the total sum of $\md_M(W(\hat{p}(l, r]]))$ over interval $(l, r]$ on an active scale, namely $o(M\treemd_{M, \tilde{S}}{\left(W(p)\right)})$, which from \cref{claim:technicalityed} is $o(w(\hat{p}))$ with probability at least $4 / 5$.

    Combining all three types, with probability $4 / 5 - o(1) > 19 / 25$, the total overestimation is bounded by $o(w(p))$.
\end{proof}

We now prove the auxiliary claim from before:
\claimtechnicalityed*
\begin{proof}
    We show a stronger result: with probability at least $4 / 5$, we have $\treemd_{M, \tilde{S}}(W(\hat{p})) = o\left(\frac{w(p)}{M}\right)$ (which is stronger since $w(p) \le w(\hat{p})$). 

    We first show that $\treemd_{M, \tilde{S}}(W(\textcolor{red}{p})) = o\left(\frac{w(p)}{M}\right)$. Then we argue that the difference between $\treemd_{M, \tilde{S}}(W(\textcolor{blue}{\hat{p}}))$ and $\treemd_{M, \tilde{S}}(W(\textcolor{red}{p}))$ is also $o\left(\frac{w(p)}{M}\right)$.
   \paragraph{Part I: $\treemd_{M, \tilde{S}}(W(p))$ is small} Using the bound on tree total deviation in \cref{lemma:totalmeandeviation} we know that, 
    \begin{align*} 
    \treemd_{M, S}(W(p)) &= O\left(w(p) \cdot S ^ {0.51}  + \underbrace{{\mathcal{B}}_x}_{\text{max length of $p$ over scale $0$ interval}} \cdot \underbrace{\left(n / {\mathcal{B}}_x\right)}_{\text{number of scale $0$ intervals}} /  2 ^ {\log_M ^ {0.015}(n)}\right) \\
    &= O\left(w(p) \cdot S ^ {0.51}  + n /  2 ^ {\log_M ^ {0.015}(n)}\right).
    \end{align*}
    Since each scale is in $\tilde{S}$ u.a.r~with probability $1 / S ^ {0.98}$, and $p$ does not depend on $\tilde{S}$, with probability at least $9 / 10$,
    \begin{align*}
        \treemd_{M, \tilde{S}}(W(p)) &= O\left(\left(w(p) \cdot S ^ {0.51} + \textcolor{red}{n /  2 ^ {\log_M ^ {0.015}(n)}}\right) / S ^ {0.98}\right) \\
        &= \left(O\left(\left(w(p) \cdot S ^ {0.51}\right) + \textcolor{red}{o\left(\frac{w(p)}{M}\right)}\right)\right) / S ^ {0.98} && \text{since $w(p) = n / 2 ^ {o\left(\log ^ {0.01} (n)\right)}$} \\
        &= w(p) / S ^ {0.47} + o\left(\frac{w(p)}{M}\right) \\
        &= o\left(\frac{w(p)}{M}\right). && \text{since $M = \Theta\left(\log ^ {\mval} (n)\right)$}
    \end{align*}
    
    \paragraph{Part II: $\treemd_{M, \tilde{S}}(W(\hat{p}))$ is not too different from $\treemd_{M, \tilde{S}}(W(p))$} From \cref{lemma:regular}, with probability at least $19 / 20$. We can round $p$ to an $\tilde{S}$-regular path $\hat{p}$ such that there are at most $o\left(\frac{w(p)}{M\log ^ {0.02}(n)}\right)$ columns $x$ before which their weights differ. Each of these columns contributes at most $O(\abs{\tilde{S}})$ to the difference between $\treemd_{M, \tilde{S}}(W(p))$ and $\treemd_{M, \tilde{S}}(W(\hat{p}))$. Since the size of $\tilde{S}$ is $O\left(\log ^ {0.02} (n)\right)$ with probability $1 - o(1) > 19 / 20$, by the union bound, with probability at least $9 / 10$, the difference between $\treemd_{M, \tilde{S}}(W(p))$ and $\treemd_{M, \tilde{S}}(W(\hat{p}))$ is at most $o\left(\frac{w(p)}{M}\right)$.
    
    Combining the two parts, with probability at least $4 / 5$, we have $\treemd_{M, \tilde{S}}(W(\hat{p})) = o\left(\frac{w(\hat{p})}{M}\right)$. 
\end{proof}

\subsubsection{Our Estimate Is Sound}
We now show that our algorithm is sound. In order to rule out potential overfitting from our algorithm, we will need the following ingredient showing that certain random 0-1 weighted graphs with a few $1$ in weights likely does not admit a path from a fixed source that visits too many of these $1$'s:
\begin{restatable}{claim}{claimdirected} \label{claim:directed}
    Let $G$ be an $N$-vertex $h$-layer directed graph with 0-1 edge weights with a source $u$. Specifically:
    \begin{itemize}
        \item Every vertex belongs to a layer in $[h]$;
        \item For any edge from vertex $v$ to vertex $w$, $v$'s layer is equal to $w$'s layer minus one;
        \item Each edge has a weight $\in \{0, 1\}$.
    \end{itemize}
    
    We say an edge is a layer-$L$ edge if it connects a layer-$L$ vertex to a layer-$(L + 1)$ vertex. 
    
    Suppose that for any vertex $v$ on layer $L$, for any layer $L' < L + B$, the number of layer-$L'$ edges reachable by $v$ is at most $d$.
    
    Suppose that the edge weights satisfy the following:
    \begin{itemize}
        \item Each edge has weight $1$ with probability at most $p = O(1 / d ^ 2)$;
        \item Edges from different layers are independent in weights. Specifically, let $\e$ be a set of edges, and let $e$ be a layer-$L$ edge. Suppose that no edge in $\e$ is also a layer-$L$ edge. Then $\Pr\left(\text{$e$'s weight is $1$}\right) =  \Pr\left(\text{$e$'s weight is $1$} \mid \text{weights of $\e$ are known}\right)$.
    \end{itemize}
    With probability $1 - N ^ {-\omega(1)}$, no path from $u$ has a length (sum of edge weights) of $\omega(h\sqrt{\log (N)} / \min(d, B))$.
\end{restatable}

Before proving the ingredient above, we first give our main analysis that our estimate is sound:
\lemmasounded*
\begin{proof}   
For each active scale $s$, we consider the following question: how much would we underestimate if there were no subsampling on active scales larger than $s$? To capture this, consider the Scale-$s$ Sparsified Grid Graph $\mathcal{G}_s$ defined on $\whest$ (\cref{def:scalesparse}), namely the following graph:
    \begin{description}
        \item[Vertices] We only consider vertices $\vertex{x}{y}$ in the Sparsified Edit Distance Grid Graph with the additional constraint that $I_{s} \mid x$;
        \item[Edges] There is an edge from $\vertex{x}{y}$ to $\vertex{x + I_{s}}{y}$ (if both exist) such that $\abs{y' - y} \le I_{s}$ with a weight of $\whest[x, y, x + I_{s}, y']$.
    \end{description}
    Let $\mathcal{A}_s$ be the length of the shortest $\tilde{S}$-regular $\vertex{0}{{n_X}}$-$\vertex{n_X + {n_Y}}{n_Y}$ path in $\mathcal{G}_s$. We can see that $\mathcal{A}_0 = w(\hat{p})$, and that $\mathcal{A}_S = \whest[0, n_X, n_X + n_Y, n_Y]$, namely our final estimate for $w(\hat{p})$. To see how much smaller $\mathcal{A}_S$ can be compared to $\mathcal{A}_0$, we go through each scale $s$ and examine the difference between $\mathcal{A}_{s - 1}$ and $\mathcal{A}_s$. 

    If scale-$s$ is passive, then since we compute $\whest[l, y_l, r, y_r]$ to be the length of the shortest $\tilde{S}$-regular $\vertex{l}{y_l}$-$\vertex{r}{y_r}$ path in the Scale-$(s - 1)$ Sparsified Grid Graph, $\mathcal{A}_s = \mathcal{A}_{s - 1}$.

    If scale-$s$ is active, then from \cref{claim:onesideded}, for $\whest[l, y_l, r, y_r]$ such that $(l, r]$ is a scale-$s$ interval:
    \begin{itemize}
        \item The shortest $\tilde{S}$-regular $\vertex{l}{y_l}$-$\vertex{r}{y_r}$ path in $\mathcal{G}_{s - 1}$ is $X := \sum\limits_{i = 1} ^ M{\whest[{\anc}_{i - 1}, \overline{y}_i, {\anc}_i, \overline{y}_i]}$;
        \item The shortest $\tilde{S}$-regular $\vertex{l}{y_l}$-$\vertex{r}{y_r}$ path in $\mathcal{G}_{s}$, namely $\whest[l, y_l, r, y_r]$, is at least $X - o\left(\frac{w(\hat{p}(l, r])}{\log ^ {0.02} (n)}\right)$ with probability $1 - \exp\left(-\omega\left(\log ^ {0.02} (n)\right)\right)$. Even if the error guarantee fails, the underestimation is at most $X \le I_s$.
    \end{itemize}
    Thus, going from $\mathcal{A}_{s - 1}$ to $\mathcal{A}_s$, we have:
    \begin{itemize}
        \item Type I: per scale-$s$ interval $(l, r]$, underestimation of at most $o\left(\frac{w(\hat{p}(l, r])}{\log ^ {0.02} (n)}\right)$;
        \item Type II: for $\whest[l, y_l, r, y_r]$ such that $(l, r]$ is a scale-$s$ interval, there is a probability of $\exp\left(-\omega\left(\log ^ {0.02} (n)\right)\right)$ that $\whest[l, y_l, r, y_r]$ has an underestimation of at most $I_s$.
    \end{itemize}
    We now prove that with probability $1 - o(1)$, the total error across both types and all scales and intervals is $o(w(\hat{p}))$. 

    Recall that $|\tilde{S}| = O(S ^ {0.02}) = o\left(\log ^ {0.02} n\right)$ w.p.~$1 - o(1)$.
    Thus the total type I underestimation across all $\abs{\tilde{S}}$  active scales is $o((w(\hat{p})))$ w.p.~$1 - o(1)$.
    
    For type II underestimation, we say an estimate $\whest[l, y_l, r, y_r]$ is ``bad'' if it exhibits an underestimation of $\Omega(w(p(l, r]) / \log ^ {0.02} (n))$. The total underestimation from estimates being bad is upper bounded by the following model:
    \begin{definition} [Upper Bound for Underestimation due to ``Bad'' Estimates]
        Consider the following graph:
        \begin{description}
            \item[Vertices] We only consider vertices $\vertex{x}{y}$ in the Sparsified Edit Distance Grid Graph with the additional constraint that $I_s \mid x$;
            \item[Edges] There is an edge from $\vertex{x}{y}$ to $\vertex{x + I_s}{y'}$ such that $\abs{y' - y} \le I_s$ with a weight of either $I_s$, if $\whest[x, y, x + I_s, y']$ is ``bad,'' or $0$ otherwise.
            \item[Objective] We obtain the upper bound by taking the longest $\vertex{0}{{n_X}}$-$\vertex{n_X + {n_Y}}{n_Y}$ path in this graph.
        \end{description}
    \end{definition}
    We consider applying \cref{claim:directed} to this graph. Note that:
    \begin{itemize}
        \item The graph has $o(n ^ {0.03})$ vertices and has $(n / I_s + 1)$ layers. Edges only connect vertices from adjacent layers (i.e., layer-$L$ to layer-$(L + 1)$ for some $L$) with weight either $0$ or $I_s$;
        \item Each edge weight is $I_s$ with probability $\exp\left(-\omega\left(\log ^ {0.02} (n)\right)\right)$;
        \item Since we sample $\eta_{l, r, 1 \ldots M}$ independently for each $(l, r]$, edges from different layers are independent in weights;
        \item Fix a vertex $\vertex{x}{y}$. If an edge from $\vertex{x'}{y'}$ to $\vertex{x' + I_s}{y''}$ is reachable from $\vertex{x}{y}$ within $\exp\left(\Theta\left(\log ^ {0.02} (n)\right)\right)$ hops (i.e., number of edges visited), then $\max(\abs{x' - x}, \abs{y' - y}) = I_s \cdot \exp\left(\Theta\left(\log ^ {0.02} (n)\right)\right)$ and $\abs{y'' - y'} \le I_s$. Since $I_s$ divides each of $x, x'$ and $(I_s / \phi_y)$ divides each of $y, y', y''$, the number of reachable edges is less than $$\underbrace{\exp\left(\Theta\left(\log ^ {0.02} (n)\right)\right)}_{x'} \cdot \underbrace{\exp\left(\Theta\left(\log ^ {0.02} (n)\right)\right) \cdot \phi_y}_{y'} \cdot \underbrace{\phi_y}_{y''} = \exp\left(\Theta\left(\log ^ {0.02} (n)\right)\right).$$
    \end{itemize} 
    From \cref{claim:directed} (with weights scaled up by $I_s$) with 
    \begin{itemize}
        \item $h := O(n / I_s)$;
        \item $N := O(n ^ {0.03})$;
        \item $B := \exp\left(\Theta\left(\log ^ {0.02} (n)\right)\right)$;
        \item $d := \exp\left(\Theta\left(\log ^ {0.02} (n)\right)\right)$;
        \item $p := \exp\left(-\omega\left(\log ^ {0.02} (n)\right)\right)$;
        \item $u := \vertex{0}{n_X}$.
    \end{itemize}
    we obtain that with probability $1 - n ^ {-\omega(1)}$, an upper bound for the underestimation on each scale due to ``bad'' estimates (type II error) is 
    $$n \cdot \sqrt{\log (n ^ {0.03})} / \exp\left(\Theta\left(\log ^ {0.02} (n)\right)\right) = n \cdot \exp\left(-\Omega\left(\log ^ {0.02} (n)\right)\right).$$
    By union bound across all scales, we obtain that with probability $1 - o(1)$, the total type II error is
    $$n \cdot S \cdot \exp\left(-\Omega\left(\log ^ {0.02} (n)\right)\right) = n \cdot \exp\left(-\Omega\left(\log ^ {0.02} (n)\right)\right).$$
    Since $w(\hat{p}) = n / 2 ^ {o\left(\log ^ {0.01} (n)\right)}$, this is $o(w(\hat{p}))$.
\end{proof}

We now prove the auxiliary claim from before:
\claimdirected*
\begin{proof}
    Fix a vertex $v$. Fix an integer $0 \le B' < B$. Suppose that $v$ is on layer $L$. We know that there are at most $d$ layer-$(L + B')$ edges reachable by $v$. Note that $p = O(1 / d ^ 2)$. From the union bound, with probability at least $1 - O(1 / d)$, no layer-$(L + B')$ edges reachable by $v$ has weight $1$. Let $X_{k}$ be the indicator variable of whether the $(L + k - 1)$-th layer contains a reachable edge with weight $1$. From the independence condition in the claim statement, $\sum\limits_{1 \le k \le \min(d, B)}{X_k}$ is the sum of $\min(d, B)$ independent Bernoulli random variables each with mean $O(1 / d)$, and this sum has mean $O(1)$. By the standard concentration bound on this sum, the probability that, before layer-$(L + \min(d, B))$, there are $\omega(\log (N))$ layers on which there are at least one reachable edge with weight $1$ (i.e., the sum is in the $\omega(\sqrt{\log (N)})$-tail), is $N ^ {-\omega(1)}$. Thus with probability $1 - N ^ {-\omega(1)}$, there are no paths from $v$ visiting at most $\min(d, B)$ edges with length (sum of edge weights) $\omega(\sqrt{\log (N)})$. By union bound over every vertex $v$, there are no paths visiting at most $\min(d, B)$ edges with length $\omega(\sqrt{\log (N)})$.

    Consider the longest path $p_u$ from $u$. This path visits at most $h$ edges. Consider partitioning $p_u$ into $O(h / \min(d, B))$ consecutive blocks containing at most $B$ edges. Each block is a path visiting no more than $\min(d, B)$ edges. Thus we know that with probability $1 - N ^ {-\omega(1)}$, the length of each block is $O\left(\sqrt{\log (N)}\right)$ and the length of $p_u$ is $O(h\sqrt{\log (N)} / \min(d, B))$.
\end{proof}

\subsubsection{Our Can be computed Efficiently}
Finally we show that our algorithm is efficient:
\lemmaefficiencyed*
\begin{proof}
    For the running time, the total number of $(l, y_l, r, y_r)$ is 
    $$\underbrace{(n / {\mathcal{B}}_x)}_{l} \cdot \underbrace{(n / {\mathcal{B}}_x)}_{r} \cdot \underbrace{(n\phi_y / {\mathcal{B}}_x)}_{y_l} \cdot \underbrace{(n\phi_y / {\mathcal{B}}_x)}_{y_r} = n ^ 4 \cdot {\phi_y} ^ 2 / {{\mathcal{B}}_x} ^ 4 = o(n ^ {0.045}).$$
    For each $(l, y_l, r, y_r)$ with $(l, r]$ on scale $s$, the bottleneck for computing $\whest[l, y_l, r, y_r]$ is to compute the shortest path in the Scale-$(s - 1)$ Sparsified Grid Graph, which takes linear time in the number of edges. This number of edges is no more than the total number of $(l, y_l, r, y_r)$'s which is $o(n ^ {0.045})$. The overall running time is therefore at most $o(n ^ {0.09})$. We can do a similar analysis on the computation of the $\appr$ values and also obtain a running time of $o(n ^ {0.09})$.
    
    For query complexity, for each edge from $\vertex{x}{y}$ to $\vertex{x + {\mathcal{B}}_x}{y'}$, the edge is queried if every interval $(l, r] \supset (x, x + {\mathcal{B}}_x]$ on an active scale is sub-sampled, which happens with probability $1 / 2$ per active scale independently. Thus the probability that an edge ``survives'' sub-sampling is $2 ^ {-\abs{\tilde{S}}}$.
    With probability $1 - o(1)$, $\abs{\tilde{S}} = \Omega(S ^ {0.02})$, and the number of queries is at most $2 ^ {-\Omega(S ^ {0.02})} < 2 ^ {-\Omega\left(\log ^ {0.01} (n)\right)}$ times the total number of edges. The total number of edges is (notice that $\abs{y - y'} \le \mathcal{B}_x$) 
    $$\underbrace{(n / {\mathcal{B}}_x)}_{x} \cdot \underbrace{(n\phi_y / {\mathcal{B}}_x)}_{y} \cdot \underbrace{\phi_y}_{y'} = \frac{n ^ 2{\phi_y} ^ 2}{{{\mathcal{B}}_x} ^ 2}$$
    And the total number of queries is $\frac{n ^ 2{\phi_y} ^ 2}{{{\mathcal{B}}_x} ^ 2} / 2 ^ {\Omega\left(\log ^ {0.01} (n)\right)}$.
\end{proof}

\subsubsection{Proof of \cref{claim:onesideded}}
Finally, we give the missing proof of \cref{claim:onesideded}.
 \claimonesideded*
\begin{proof}
For $({\anc}_{i - 1}, {\anc}_i]$ we write $\overline{p}_{{\anc}_{i - 1}, {\anc}_i}$ as a shorthand for the part of $\overline{p}_{l, r}$ from column ${\anc}_{i - 1}$ to column ${\anc}_i$ (i.e., $\overline{p}_{l, r}({\anc}_{i - 1}, {\anc}_i]$).

\paragraph{We only give up when weight distribution is imbalanced} Since $\appr[{\anc}_{i - 1}, \overline{y}_{i - 1}, {\anc}_i, \overline{y}_i] = \Theta(w(\overline{p}_{{\anc}_{i - 1}, {\anc}_i}))$ and $\appr[l, y_l, r, y_r] = \Theta(w(\overline{p}_{l, r}))$, if we have $\max\limits_{1 \le i \le M}{\appr[{\anc}_{i - 1}, \overline{y}_{i - 1}, {\anc}_i, \overline{y}_i]} = \omega\left(\frac{\appr[l, y_l, r, y_r]}{M}\right)$, then we must also have $\max\limits_{1 \le i \le M}{w(\overline{p}_{{\anc}_{i - 1}, {\anc}_i})} = \omega\left(\frac{w(\overline{p}_{l, r})}{M}\right)$. Let $i'$ be the maximizier of $w(\overline{p}_{{\anc}_{i' - 1}, {\anc}_{i'}})$. We have
\begin{align*}
    \md_M(W(\overline{p}_{l, r}))   &=  \sum\limits_{1 \le i \le M}{\abs{w(\overline{p}_{{\anc}_{i - 1}, {\anc}_i}) - \frac{w(\overline{p}_{l, r})}{M}}} \\
                                                                &\ge \abs{w(\overline{p}_{{\anc}_{i' - 1}, {\anc}_{i'}}) - \frac{w(\overline{p}_{l, r})}{M}} \\
                                                                &= \abs{\omega\left(\frac{w(\overline{p}_{l, r})}{M}\right) - \frac{w(\overline{p}_{l, r})}{M}} \\
                                                                &= \omega\left(\frac{w(\overline{p}_{l, r})}{M}\right)
\end{align*}

\paragraph{If we do not give up, we obtain a good estimate} Notice that 
\begin{itemize}
    \item For every sub-interval, $\whest[{\anc}_{i - 1}, \overline{y}_{i - 1}, {\anc}_i, \overline{y}_i] \le \appr[{\anc}_{i - 1}, \overline{y}_{i - 1}, {\anc}_i, \overline{y}_i]$;
    \item For the parent interval, $\appr[l, y_l, r, y_r] = O(w(\overline{p}_{l, r}))$.
\end{itemize} 
If $\max\limits_{1 \le i \le M}{\appr[{\anc}_{i - 1}, \overline{y}_{i - 1}, {\anc}_i, \overline{y}_i]} \le \frac{\appr[l, y_l, r, y_r]}{M}\log \log (n)$, then for every sub-interval, $\whest[{\anc}_{i - 1}, \overline{y}_{i - 1}, {\anc}_i, \overline{y}_i] = O\left(\frac{w(\overline{p}_{l, r})}{M}\log \log (n)\right)$. By applying Hoeffding's Inequality (\cref{lemma:hoeffding}, with $c := O\left(\frac{w(\overline{p}_{l, r})}{M}\log \log (n)\right)$ and $\delta := \frac{w(\overline{p}_{l, r})\log \log (n)}{\log ^ {0.03} (n)} = o\left(\frac{w(\overline{p}_{l, r})}{\log ^ {0.02} (n)}\right)$), we obtain that
    \begin{equation*}
    \abs{2\sum\limits_{i = 1} ^ M{\eta_{l, r, i}\whest[{\anc}_{i - 1}, \overline{y}_i, {\anc}_i, \overline{y}_i]} - \sum\limits_{i = 1} ^ M{\whest[{\anc}_{i - 1}, \overline{y}_i, {\anc}_i, \overline{y}_i]}} = o\left(\frac{w(\overline{p}_{l, r})}{\log ^ {0.02} (n)}\right)
    \end{equation*}
    holds except with probability $\exp\left(-\Omega\left(M / \log ^ {0.06}(n)\right)\right) = \exp\left(-\omega\left(\log ^ {0.02}(n)\right)\right)$. 
\end{proof}

\section{Algorithm for LCS}
In this section, we present our algorithm for LCS:
\begin{theorem} \label{theo:mainlcs}
For any constant $\eps > 0$, there exists a randomized algorithm that runs in time 
\[
\frac{n^2}{2^{\log^{\Omega(1)}(n)}}
\]
and, with high probability, computes a $(1-\eps)$-approximation for the longest common subsequence of two strings with total length $n$.
\end{theorem}

Our algorithm solves the query-based algorithm on the sparsified grid introduced in \cref{sec:sparsified}.
\defgridgraphmax*
\lemmagridgraphmax*

\subsection{Regular Paths}
Let $\tilde{S} \subset [S]$ be a subset of active scales (other scales are \emph{passive}). We define the notion of $\tilde{S}$-regularity: 
\begin{definition} 
    Given $\tilde{S} \subset [S]$, a path $\hat{p}$ is $\tilde{S}$-regular (or is an $\tilde{S}$-regular path) if for every $s \in \tilde{S}$ and every scale-$s$ interval $(l, r]$, for every $0 \le i \le M$, we have
    $$\hat{p}({\anc}_i) = (I_{s - 1} / \phi_y) \cdot \round{\frac{\hat{p}(l) + i \cdot \frac{\hat{p}(r) - \hat{p}(l)}{M}}{I_{s - 1} / \phi_y}} =: \overline{y}_i.$$
\end{definition}
In other words, inside every active scale interval, an $\tilde{S}$-regular path is the straight line rounded to nearest multiples of $I_{s-1} / \phi_y$ (due to the way the sparsified grid is defined). 

We show that if $\tilde{S} \subset [S] \cap 3\mathbb{Z}$ is such that each scale that is a multiple of $3$ is chosen u.a.r.~with probability $1 / S ^ {0.98}$ (thus of size $\Theta(S ^ {0.02})$ with probability $1 - o(1)$), then enforcing $\tilde{S}$-regularity does not decrease the answer too much:
\begin{restatable}[Best regular path approximates best path]{lemma}{lemmaregularlcs} \label{lemma:regularlcs}
    Let $\tilde{S} \subset [S] \cap 3\mathbb{Z}$ be such that every scale that is a multiple of $3$ is chosen u.a.r.~with probability $1 / S ^ {0.98}$. Let $p$ be the longest $\vertex{0}{{n_X}}$-$\vertex{n_X + {n_Y}}{n_Y}$ \Aviad{is it important that $p$ is the longest?}\xiao{for the proof to work, yes} path with $w(p) = n / 2 ^ {o\left(\log ^ {0.01} (n)\right)}$. With probability at least $19 / 20$, we can round $p$ to an $\tilde{S}$-regular $\vertex{0}{{n_X}}$-$\vertex{n_X + {n_Y}}{n_Y}$ path $\hat{p}$ such that $$w(\hat{p}) > \left(1 - o\left(\frac{1}{M\log ^ {0.02} (n)}\right)\right)w(p).$$
    Moreover, there are at most $o\left(\frac{w(p)}{M\log ^ {0.02}{(n)}}\right)$ columns before which the weights of $p$ and $\hat{p}$ differ.
\end{restatable}
Before formally proving \cref{lemma:regularlcs}, we first give an overview. Suppose the weight of $p$ before column $x$ is $1$. We want to know if we can round $p$ to a $\tilde{S}$-regular path such that the weight before column $x$ still remains $1$. We will give an argument via \cref{lemma:treerd} that for recursion intervals $(l, r]$ that contain $x$, the curvatures of $p$ inside $(l, r]$ are small on average, namely that $p(l, r]$ is already ``relatively straight,'' meaning that each $p({\anc}_i)$ should already be close to $\overline{y}_i$. Because of this, rounding $p({\anc}_i)$ to $\overline{y}_i$ should only affect a small fraction of the path in the vicinity of column ${\anc}_i$. Equivalently, column $x$ is only affected if it is very close to an anchor ${\anc}_i$. We call $x$ a \emph{hitter} of $(l, r]$ if there exist an anchor ${\anc}_i = l + i \cdot (r - l) / M$ such that $x \in ({\anc}_i - I_{s - 3}, {\anc}_i + I_{s - 3}]$ (which is the union of the two scale-$(s - 3)$ intervals around ${\anc}_i$). We call a column $x'$ a \emph{frequent hitter} if there are more than $S ^ {0.83}$ scales $s'$ that is a multiple of $3$ such that $x'$ is a hitter of the scale-$s'$ interval $(l', r'] \ni x'$. Since each scale that is a multiple of $3$ is in $\tilde{S}$ u.a.r~with probability $1 / S ^ {0.98}$, we can see that if $x'$ is not a frequent hitter, then there is a probability of at least $1 - 1 / S ^ {0.15}$ that no interval where $x'$ is a hitter lies on an active scale. Thus it suffices to prove two claims:
\begin{itemize}
\item There are not too many frequent hitters (\cref{claim:frequenthitter});
\item If $x'$ is not a hitter of an interval that contains it, then on average, the rounding at that interval should not affect $x'$ (\cref{claim:nohit}). 
\end{itemize}
Now we introduce the two claims in detail. We claim that there are not too many frequent hitters:
\begin{restatable}{claim}{claimfrequenthitter} \label{claim:frequenthitter}
    The number of frequent hitters does not exceed $n / 2 ^ {\log ^ {0.02} (n)}$.
\end{restatable}
Recall that for any column $x'$, $\abs{p(x' + 1) - p(x')} \le 1$, namely that the slope of $p$ cannot be more than $45 ^ {\circ}$ anywhere (either in the $(+1, +1)$ or the $(+1, -1)$ direction). The closer $p(l, r]$ is to a 45-degree diagonal (i.e., the smaller $(r - l) - \abs{p_s(r) - p_s(l)}$ is), the less freedom the path $p(l, r]$ has. Therefore, it is beneficial to ``normalize'' curvature by dividing it by ${(r - l) - \abs{p_s(r) - p_s(l)}}$. Recall that for a path $p$ from $\vertex{l}{y_l}$ to $\vertex{r}{y_r}$, we let the forward difference operator $\Delta p$ be the sequence $$\Delta p := \left(p(l + {\mathcal{B}}_x) - p(l), p(l + 2{\mathcal{B}}_x) - p(l + {\mathcal{B}}_x), \ldots, p(r) - p(r - {\mathcal{B}}_x)\right).$$ We have the following claim which formally states that if such ``normalized'' total deviation is small, then rounding does not affect non-hitters: 
\begin{restatable}{claim}{claimnohit} \label{claim:nohit}
Given a scale-$s$ interval $(l, r]$, let $p_s$ be a path from column $l$ to column $r$ with $\abs{p_s(r) - p_s(l)} < r - l$. Suppose that
$$\frac{\md_{M ^ 3}(\Delta p_s(l, r])}{(r - l) - \abs{p_s(r) - p_s(l)}} = o\left(\frac{1}{M ^ 3}\right),$$
where $\md_{M ^ 3}(\Delta p_s(l, r])$ is the $M ^ 3$-way total deviation of $\Delta p_s(l, r]$:
$$\md_{M ^ 3}(\Delta p_s(l, r]) = \sum_{1 \le i' \le M ^ 3}{\abs{p_s({\anc}^{'}_{i'}) - p_s({\anc}^{'}_{i' - 1}) - \frac{p_s(r) - p_s(l)}{M ^ 3}}},$$
where for every $0 \le i' \le M ^ 3$, the $i'$-th $M ^ 3$-way anchor is ${\anc}^{'}_{i'} := l + i' \cdot \frac{r - l}{M ^ 3}$.

Then, there exists an $\{s\}$-regular $\vertex{l}{p_s(l)}$-$\vertex{r}{p_s(r)}$ path $p_{s - 1}$ such that $$p_s(x - 1, x] = p_{s - 1}(x - 1, x]$$ holds for every $x$ that is not a hitter of $(l, r]$.
\end{restatable}
To facilitate the use of \cref{claim:nohit}, we later derive the following claim  using \cref{lemma:treerd}:
\begin{restatable}{claim}{claimlcstreerd} \label{claim:lcstreerd}
Let $p$ be the path in \cref{lemma:regularlcs}.
For a random column $x$ such that the weight of $p$ before $x$ is non-zero, with probability $1 - o\left(\frac{1}{M \log ^ {0.02} {(n)}}\right)$, every interval $(l, r] \ni x$ on an active scale satisfies the following bound: 
\begin{equation} \label{eq:nohitraw}
\frac{\md_{M ^ 3}(\Delta (p(l, r]))}{(r - l) - \abs{p(r) - p(l)}} = o\left(\frac{1}{M ^ 3}\right),
\end{equation}
which is the requirement in \cref{claim:nohit} with $p_s$ replaced by $p$\footnote{Note that $\abs{p(r) - p(l)} < r - l$ is automatically satisfied from the lemma condition.}. 
\end{restatable}

We first present the full proof using these three claims:
\begin{proof} [Proof of \cref{lemma:regularlcs}]
To prove the lemma, we perform rounding on a scale-by-scale basis. Let $p_S := p$. For each scale $s \in [S]$ in decreasing order, if $s \in \tilde{S}$, we round $p_s$ to an $\{s\}$-regular path $p_{s - 1}$. Specifically, for every scale-$s$ interval $(l, r]$, we adjust $p_s({\anc}_i)$ to $p_{s - 1}({\anc}_i) := (I_{s - 1} / \phi_y) \cdot \round{\frac{p_s(l) + i \cdot \frac{p_s(r) - p_s(l)}{M}}{I_{s - 1} / \phi_y}}$. If $p_s(l, r]$ satisfies the condition in \cref{claim:nohit}, then we use the claim to obtain $p_{s - 1}(l, r]$ in a way that does not change any column that is not a hitter. If $p_s(l, r]$ does not satisfy the condition in \cref{claim:nohit}, we round $p_s(l, r]$ to $p_{s - 1}(l, r]$ arbitarily (e.g., using \cref{lemma:pathdiff}). If $s \notin \tilde{S}$ we simply let $p_{s - 1} := p_{s}$. Let the final path be $\hat{p} := p_0$. 

For the rest of the proof, we bound the number of columns before which the weight of $p$ and the weight of $\hat{p}$ differ. This is the sum of:
\begin{itemize}
\item $\mathcal{X} := $ the number of columns before which the weight of $p$ is $1$ and the weight of $\hat{p}$ is $0$;
\item $\mathcal{Y} := $ the number of columns before which the weight of $p$ is $0$ and the weight of $\hat{p}$ is $1$.
\end{itemize}
Since $p$ is the longest path, we have $w(p) \ge w(\hat{p}) = w(p) - \mathcal{X} + \mathcal{Y}$, and hence $\mathcal{X} \ge \mathcal{Y}$. Thus it suffices to prove that $\mathcal{X}= o\left(\frac{w(p)}{M \log ^ {0.02} {(n)}}\right)$. 
Now consider uniformly randomly sampling a column $x$ such that the weight of $p$ before column $x$ is $1$.
We will argue that the weight of $\hat{p}$ before column $x$ is $0$ with probability $o\left(\frac{1}{M \log ^ {0.02} {(n)}}\right)$, which would mean that the expected number of columns where the weights of $p$ and $\hat{p}$ differ is $o\left(\frac{w(p)}{M \log ^ {0.02} {(n)}}\right)$, and would finish the proof due to Markov's inequality. 

To present this argument, we will show that the only few exceptions to $\hat{p}(x - 1, x] = p(x - 1, x]$ are the columns $x$ such that either \eqref{eq:nohitraw} is not satisfied by some interval $(l, r] \ni x$ on an active scale, or $x$ hits an interval on an active scale. We first demonstrate that these two cases of exceptions indeed only happen with $o\left(\frac{1}{M \log ^ {0.02} {(n)}}\right)$ probability when $x$ is a uniformly random column with before which the weight of $p$ is non-zero.

\paragraph{Case I: \eqref{eq:nohitraw} is not satisfied by some interval $(l, r] \ni x$ on an active scale} \cref{claim:lcstreerd} shows that \eqref{eq:nohitraw} is not satisfied by some interval $(l, r] \ni x$ on an active scale with probability $o\left(\frac{1}{M \log ^ {0.02} {(n)}}\right)$.

\paragraph{Case II: $x$ hits an interval on an active scale} This can be further categorized into the following:
\begin{description}
\item[Frequent hitters] From \cref{claim:frequenthitter}, the number of $x$'s that are frequent hitters is at most $n / 2 ^ {\log ^ {0.02} (n)} = o\left(\frac{w(p)}{M \log ^ {0.02} {(n)}}\right)$, so $x$ is a frequent hitter with probability $o\left(\frac{1}{M \log ^ {0.02} {(n)}}\right)$;
\item[Non-frequent-hitter columns that hit intervals on active scales by chance] Given an $x$ that is not a frequent hitter, since the number of scales $s$, such that $x$ is a hitter of the interval $(l, r] \ni x$ on scale $s$, is at most $S ^ {0.83}$, and since each scale that is a multiple of $3$ is added u.a.r.~with probability $1 / S ^ {0.98}$ to $\tilde{S}$, only with probability at most $1 / S ^ {0.15} = o\left(\frac{1}{M \log ^ {0.02} {(n)}}\right)$ can there be an interval $(l, r] \ni x$ on an active scale such that $x$ is a hitter of $(l, r]$.
\end{description}

To show that these are the only exceptions and thereby finish the proof, for every scale $s$, we show the (scale-$s$) proposition that $p_s(x - 1, x] = p(x - 1, x]$ for every column $x$ (even those before which the weight of $p$ is $0$), such that both of the following conditions 
\begin{enumerate}
    \item $x$ does not hit any interval on an active scale $s' > s$, and
    \item \eqref{eq:nohitraw} is satisfied by every interval $(l, r] \ni x$ on an active scale $s' > s$.
\end{enumerate}   
This would imply that if every interval on an active scale containing $x$ satisfies \eqref{eq:nohitraw}, and $x$ does not hit any interval on an active scale, then we would have $\hat{p}(x - 1, x] = p(x - 1, x]$ (i.e., the only exceptions to the argument are these two cases above). We can prove the proposition by induction, using a step-wise decrement of 3 scales down from scale $S$ (for which the proposition is trivially true since there are no active scales above). As our induction assumption, suppose that the proposition is true for a scale $s$ that is a multiple of $3$ and we want to argue that it is also true for scale $(s - 1)$, $(s - 2)$ and $(s - 3)$. If $s$ is not an active scale then no rounding takes place and the proposition is therefore true for scale $(s - 1)$ (and thus for scales $(s - 2)$ and $(s - 3)$ since active scales are multiples of $3$). The remaining case is if $s$ is an active scale and the scale-$(s - 3)$ \textcolor{red}{proposition conditions} also hold for $x$ (i.e., $x$ also does not hit the scale-$s$ interval $(l_{s}, r_{s}] \ni x$, and \eqref{eq:nohitraw} holds for $(l_{s}, r_{s}]$). Let $(l_s, r_s]$ be the scale-$s$ interval containing $x$. For every $x' \in (l_s, r_s]$, for every scale $s' \ge s$, the scale-$s'$ interval containing $x'$ is identical to the scale-$s'$ interval containing $x$. Therefore the scale-$s$ \textcolor{red}{proposition conditions} hold for every $x' \in (l_s, r_s]$ and we have $p_s(l_s, r_s] = p(l_s, r_s]$ by the induction assumption that the scale-$s$ proposition holds. Since \eqref{eq:nohitraw} holds for $p(l_s, r_s]$, it also holds for $p_s(l_s, r_s]$. Since $x$ does not hit $(l_s, r_s)$, from \cref{claim:nohit}, we must have $p_{s - 1}(x - 1, x] = p(x - 1, x]$ (and thus $p_{s - 2}(x - 1, x] = p_{s - 3}(x - 1, x] = p(x - 1, x]$ since $(s - 1)$ and $(s - 2)$ are not active scales and no rounding takes place).
\end{proof}

Now we prove the three ingredients used in the proof of \cref{lemma:regularlcs} above:
\claimfrequenthitter*
\begin{proof}
    For a column $x$ and a scale $s$ that is a multiple of $3$, let $g_{x, s} \in [0, 1]$ be $1$ if $x$ is a hitter of the interval $(l, r] \ni x$ on scale $s$. Notice that the values of $g_{x',s'}$ for $s' \ge s + 3$ are identical for all $x' \in (l,r]$ since those values depend on anchors of larger intervals $(l',r'] \supseteq (l,r]$. Therefore $g_{x,s}$ is independent of $g_{x,s'}$ for $s' \ge s + 3$. Furthermore, $g_{x,s}$ is $+1$ for $O\left(\frac{1}{M ^ 2}\right)$ of the $x \in (l,r]$. Let $g_x := \sum_{s \in [S] \cap 3\mathbb{Z}}{g_{x, s}}$. For a uniformly random $x \in [n]$, We can see that $g_x$ is the sum of $S / 2$  i.i.d.~samples of $\text{Ber}(O(1 / M ^ 2))$, whose mean is $O(S / M ^ 2) = o(S ^ {0.83})$. 

    By applying Hoeffding's inequality (\cref{lemma:hoeffding}), the probability that $g_x > S ^ {0.83}$ is less than $1 / 2 ^ {\Omega(S ^ {0.66})} \ll 1 / 2 ^ {\log ^ {0.02} (n)}$. This means that the number of frequent hitters is no more than $n / 2 ^ {\log ^ {0.02} (n)}$.
\end{proof}

\claimnohit*
\begin{proof}
    Fix $1 \le i < M$. In this proof, we show that $\abs{p_s({\anc}_i - I_{s - 3}) - \overline{y}_i} \le I_{s - 3}$ and $\abs{p_s({\anc}_i + I_{s - 3}) - \overline{y}_i} \le I_{s - 3}$, which implies that there is a $\vertex{{\anc}_i - I_{s - 3}}{p_s({\anc}_i - I_{s - 3})}$-$\vertex{{\anc}_i}{\overline{y_i}}$-$\vertex{{\anc}_i + I_{s - 3}}{p_s({\anc}_i + I_{s - 3})}$ path. Thus to regularize $p_s({\anc}_i)$ to $\overline{y}_i$, we only need to change the sub-path of $p_s$ between column $\left({\anc}_i - I_{s - 3}\right)$ and column $\left({\anc}_i + I_{s - 3}\right)$, which only affects edges before hitters.

    We show that $\abs{p_s({\anc}_i - I_{s - 3}) - \overline{y}_i} \le I_{s - 3}$. The other direction is similar due to symmetry. 
    Without loss of generality, suppose that $p_s(r) \ge p_s(l)$. We will show that:
    \begin{description}
        \item[$p_s({\anc}_i)$ is not far from $\overline{y}_i$:] $\abs{p_s({\anc}_i) - \overline{y}_i} = o\left(\frac{(r - l) - \left(p_s(r) - p_s(l)\right)}{M ^ 3}\right)$;
        \item[$p_s({\anc}_i - I_{s - 3})$ is not far from $p_s({\anc}_i)$:] $\abs{p_s({\anc}_i - I_{s - 3}) - p_s({\anc}_i)} \le \frac{\left(p_s(r) - p_s(l)\right)}{M ^ 3} + o\left(\frac{(r - l) - \left(p_s(r) - p_s(l)\right)}{M ^ 3}\right)$.
    \end{description}
    We can see that these will imply
    $$\abs{p_s({\anc}_i - I_{s - 3}) - \overline{y}_i} \le \frac{\textcolor{blue}{\left(p_s(r) - p_s(l)\right)}}{M ^ 3} + o\left(\frac{(r - l) - \textcolor{blue}{\left(p_s(r) - p_s(l)\right)}}{M ^ 3}\right) \le \frac{r - l}{M ^ 3} = I_{s - 3}.$$

    \paragraph{Preliminary: $p_s({\anc}_i)$ is not far from the $\vertex{l}{p_s(l)}-\vertex{r}{p_s(r)}$ line} Recall that for every $0 \le i' \le M ^ 3$, the $i'$-th $M ^ 3$-way anchor is ${\anc}^{'}_{i'} = l + i' \cdot \frac{r - l}{M ^ 3}$. Recall that $\bar{y}_i = (I_{s - 1} / \phi_y) \cdot \round{\frac{p_s(l) + i \cdot \frac{p_s(r) - p_s(l)}{M}}{I_{s - 1} / \phi_y}}$ is the straight-line rounded to the closest multiple of $I_{s - 1} / \phi_y$.
    As a preliminary, we first show that for every $M ^ 3$-way anchor ${\anc}^{'}_{i'}$, $\vertex{{\anc}^{'}_{i'}}{p_s({\anc}^{'}_{i'})}$ is not far from the un-rounded straight line connecting $\vertex{l}{p_s(l)}$ and $\vertex{r}{p_s(r)}$:
    \begin{equation} \label{eq:claimnohitprelim}
        \abs{p_s({\anc}^{'}_{i'}) - \left(p_s(l) + i' \cdot \frac{p_s(r) - p_s(l)}{M ^ 3}\right)} = o\left(\frac{(r - l) - \left(p_s(r) - p_s(l)\right)}{M ^ 3}\right).
    \end{equation}
    By telescoping, we have
    \begin{equation*}
        \abs{\textcolor{red}{p_s({\anc}^{'}_{i'})} - \left(p_s(l) + i' \cdot \frac{p_s(r) - p_s(l)}{M ^ 3}\right)} = \abs{\textcolor{red}{\left(p_s(l) + \sum_{j' \le i'}{\left(p_s({\anc}^{'}_{j'}) - p_s({\anc}^{'}_{j' - 1})\right)}\right)} - \left(p_s(l) + i' \cdot \frac{p_s(r) - p_s(l)}{M ^ 3}\right)}.
    \end{equation*}
    The two $p_s(l)$ terms cancel out. Thus
    \begin{align*}
        \abs{p_s({\anc}^{'}_{i'}) - \left(p_s(l) + i' \cdot \frac{p_s(r) - p_s(l)}{M ^ 3}\right)} &= \abs{\left(\sum_{1 \le j' \le i'}{\left(p_s({\anc}^{'}_{j'}) - p_s({\anc}^{'}_{j' - 1})\right)}\right) - i' \cdot \frac{p_s(r) - p_s(l)}{M ^ 3}} \nonumber \\
        &= \abs{\sum_{1 \le j' \le i'}\left({p_s({\anc}^{'}_{j'}) - p_s({\anc}^{'}_{j' - 1}) - \frac{p_s(r) - p_s(l)}{M ^ 3}}\right)} \\
        &\le \sum_{1 \le j' \le i'}{\abs{p_s({\anc}^{'}_{j'}) - p_s({\anc}^{'}_{j' - 1}) - \frac{p_s(r) - p_s(l)}{M ^ 3}}} \\
        &\le \sum_{1 \le j' \le M ^ 3}{\abs{p_s({\anc}^{'}_{j'}) - p_s({\anc}^{'}_{j' - 1})  - \frac{p_s(r) - p_s(l)}{M ^ 3}}} \\
        &= \md_{M ^ 3}(\Delta p_s(l, r]). && \text{by definition}
    \end{align*}
    Thus \eqref{eq:claimnohitprelim} holds since $\frac{\md_{M ^ 3}(\Delta (p(l, r]))}{(r - l) - \abs{p(r) - p(l)}} = o\left(\frac{1}{M ^ 3}\right)$.

    \paragraph{Step I: $p_s({\anc}_i)$ is not far from $\overline{y}_i$} 
    Note that for $a, b \in \mathbb{R}$ and $c \in \mathbb{R} ^ +$, if $c \mid a$ and $\hat{b}$ is $b$ rounded to the closest multiple of $c$, then $\abs{a - \hat{b}} \le 2\abs{a - b}$. Since $I_{s - 1} / \phi_y \mid p_s({\anc}_i)$, we must have
    \begin{equation}  \label{eq:claimnohit0}
    \abs{p_s({\anc}_i) - (I_{s - 1} / \phi_y) \cdot \round{\frac{p_s(l) + i \cdot \frac{p_s(r) - p_s(l)}{M}}{I_{s - 1} / \phi_y}}} \le 2\abs{p_s({\anc}_i) - \left(p_s(l) + i \cdot \frac{p_s(r) - p_s(l)}{M}\right)}.
    \end{equation}
    Plugging \eqref{eq:claimnohitprelim} (with $i' = i \cdot M ^ 2$) to the right-hand side of \eqref{eq:claimnohit0}, we have that indeed,
    $$\abs{p_s({\anc}_i) - \overline{y}_i} = o\left(\frac{(r - l) - \left(p_s(r) - p_s(l)\right)}{M ^ 3}\right)$$

    \paragraph{Step II: $p_s({\anc}_i - I_{s - 3})$ is not far from $p_s({\anc}_i)$}
    From \eqref{eq:claimnohitprelim} we have that
    \begin{equation} \label{eq:claimnohita}
    \abs{p_s({\anc}_i - I_{s - 3}) - \left(p_s(l) + \left(i \cdot M ^ 2 - 1\right) \cdot \frac{p_s(r) - p_s(l)}{M ^ 3}\right)} = o\left(\frac{(r - l) - \left(p_s(r) - p_s(l)\right)}{M ^ 3}\right)
    \end{equation}
    ($\left({\anc}_i - I_{s - 3}\right)$ is the $(i \cdot M ^ 2 - 1)$-th $M ^ 3$-way anchor) and that
    \begin{equation} \label{eq:claimnohitb}
    \abs{p_s({\anc}_i) - \left(p_s(l) + i \cdot M ^ 2 \cdot \frac{p_s(r) - p_s(l)}{M ^ 3}\right)} = o\left(\frac{(r - l) - \left(p_s(r) - p_s(l)\right)}{M ^ 3}\right).
    \end{equation}
    ($\left({\anc}_i\right)$ is the $(i \cdot M ^ 2)$-th $M ^ 3$-way anchor).

    Note that for $a, b, c, d \in \mathbb{R}$, we have $\abs{a - b} \le \textcolor{red}{\abs{a - c}} + \textcolor{blue}{\abs{c - d}} + \textcolor{purple}{\abs{b - d}}$ due to triangle inequality. Particularly, this is true for
    \begin{itemize}
        \item $a = p_s({\anc}_i - I_{s - 3})$;
        \item $b = p_s({\anc}_i)$;
        \item $c = \left(p_s(l) + \left(i \cdot M ^ 2 - 1\right) \cdot \frac{p_s(r) - p_s(l)}{M ^ 3}\right)$;
        \item $d = \left(p_s(l) + i \cdot M ^ 2 \cdot \frac{p_s(r) - p_s(l)}{M ^ 3}\right)$.
    \end{itemize}
    From \eqref{eq:claimnohita} (for $\textcolor{red}{\abs{a - c}}$) and \eqref{eq:claimnohitb} (for $\textcolor{purple}{\abs{b - d}}$), we have that indeed
    \begin{align*}
        \abs{p_s({\anc}_i - I_{s - 3}) - p_s({\anc}_i)} &\le \textcolor{red}{o\left(\frac{(r - l) - \left(p_s(r) - p_s(l)\right)}{M ^ 3}\right)} + \textcolor{blue}{\frac{\left(p_s(r) - p_s(l)\right)}{M ^ 3}} + \textcolor{purple}{o\left(\frac{(r - l) - \left(p_s(r) - p_s(l)\right)}{M ^ 3}\right)} \\
        & \le \frac{\left(p_s(r) - p_s(l)\right)}{M ^ 3} + o\left(\frac{(r - l) - \left(p_s(r) - p_s(l)\right)}{M ^ 3}\right)
    \end{align*}
\end{proof}

{\renewcommand\footnote[1]{}\claimlcstreerd*}
\begin{proof}
We can see that the left-hand side of \eqref{eq:nohitraw} depends on the sign of $p(r) - p(l)$. Since the cases for the two signs are identical, we only need to show this for intervals where $p(r) \le p(l)$ --- we want to show that for a column $x$ such that the weight of $p$ before column $x$ is $1$, with probability $1 - o\left(\frac{1}{M \log ^ {0.02} {(n)}}\right)$, every interval $(l, r] \ni x$ on an active scale satisfies
\begin{equation} \label{eq:nohit}
\frac{\md_{M ^ 3}(\Delta (p(l, r]))}{p(r) - p(l) + (r - l)} = o\left(\frac{1}{M ^ 3}\right).
\end{equation}
Let $A = \{a_i\}$ be the non-negative sequence whose $(i \cdot \mathcal{B}_x)$-th element is
$${a}_{i \cdot \mathcal{B}_x} := p(i \cdot {\mathcal{B}}_x) - p((i - 1) \cdot {\mathcal{B}}_x) + {\mathcal{B}}_x = \left(\Delta p\right)_i + {\mathcal{B}}_x$$
(we will never explicitly use indices that are not multiples of $\mathcal{B}_x$ and these entries are treated as zeros)
which measures how close each block of $p$ is to a $(+1, -1)$-diagonal path, and let $B = (b_i)$ be the non-negative sequence whose $(i \cdot \mathcal{B}_x)$-th element is
$${b}_{i \cdot \mathcal{B}_x} := w(p\left((i - 1) \cdot {\mathcal{B}}_x, i \cdot {\mathcal{B}}_x\right]).$$
Since total deviation is invariant under a universal translation (i.e., by $+\mathcal{B}_x$ everywhere), for every interval $(l, r]$ on a scale that is a multiple of $3$, the $M ^ 3$-way total deviation of $\Delta (p(l, r])$ and $A(l, r]$ are equal:
\begin{equation} \label{eq:lcstreerdeq0}
    \md_{M ^ 3}(\Delta (p(l, r])) = \md_{M ^ 3}(A(l, r]).
\end{equation} 
It is straightforward to see that if we choose a random column $x$ before which the weight of $p$ is $1$, the expected value of the total sum over every interval $(l, r] \ni x$ on a scale that is a multiple of $3$ of the left-hand side of \eqref{eq:nohit} is
\begin{align*}
&\sum_{s \in [S]\cap 3\mathbb{Z}}\left(\sum_{(l, r] \in s \text{ s.t.} \abs{p(r) - p(l)} < r - l}{\underbrace{\frac{\textcolor{red}{w(l, r)}}{w(p)}}_{\text{probability of $x \in (l, r]$}} \cdot \frac{\md_{M ^ 3}(\Delta (p(l, r]))}{\textcolor{blue}{p(r) - p(l) + (r - l)}}}\right) \\ 
&= \sum_{s \in [S]\cap 3\mathbb{Z}}\left(\sum_{(l, r] \in s \text{ s.t.} \abs{p(r) - p(l)} < r - l}{\frac{\textcolor{red}{B_{\Sigma(l, r]}}}{w(p)}} \cdot \frac{\md_{M ^ 3}(\Delta (p(l, r]))}{\textcolor{blue}{A_{\Sigma(l, r]}}}\right) && \text{by definition of $A$ and $B$}
\\ 
&= \sum_{s \in [S]\cap 3\mathbb{Z}}\left(\sum_{(l, r] \in s \text{ s.t.} \abs{p(r) - p(l)} < r - l}{\frac{B_{\Sigma(l, r]}}{w(p)}} \cdot \frac{\md_{M ^ 3}(A(l, r])}{A_{\Sigma(l, r]}}\right) && \text{due to \eqref{eq:lcstreerdeq0}} \\
&= \frac{1}{w(p)}{\sum_{s \in [S]\cap 3\mathbb{Z}}\sum_{(l, r] \in s}{\rd_{M ^ 3}(A(l, r], B(l, r])}} && \text{by definition of $\rd$} \\
&= \frac{1}{w(p)}\treerd_{M ^ 3, S / 3}(A, B) && \text{by definition of $\treerd$}. \\
\end{align*}
To bound this total sum, we can apply \cref{lemma:treerd} on $A, B$, but we need to make sure that the $B$-to-$A$ ratio satisfies the bound in the lemma. We call a column $x$ ``dispensable'' if there exists an recursion interval $(l, r] \ni x$ such that $w(p(l, r]) < (r - l) / 2^{\log^{0.01}(n)}$ (i.e., at most a $1 / 2^{\log^{0.01}(n)}$ fraction of $(l, r]$ contributes to the total length of $p$). Since the total sum of sizes of recursion intervals is at most $O(Sn) = O(n \log n)$, there are at most $n / 2^{\Omega\left(\log^{0.01}(n)\right)}$ dispensable columns before which the weight of $p$ is non-zero, which can be ignored since this number is $o\left(\frac{w(p)}{M \log ^ {0.02} {(n)}}\right)$. Notice that for every interval $\left((i - 1) \cdot {\mathcal{B}}_x, i \cdot {\mathcal{B}}_x\right]$, the columns inside it are either all dispensable or all indispensable. We define a new non-negative sequence $B' = \{b^{'}_i\}$ whose $(i \cdot \mathcal{B}_x)$-th element is
\[
b^{'}_{i \cdot \mathcal{B}_x} = \begin{cases}
    b_{i \cdot \mathcal{B}_x} & \text{if every $x \in \left((i - 1) \cdot {\mathcal{B}}_x, i \cdot {\mathcal{B}}_x\right]$ is indispensable} \\
    a_{i \cdot \mathcal{B}_x} & \text{if every $x \in \left((i - 1) \cdot {\mathcal{B}}_x, i \cdot {\mathcal{B}}_x\right]$ is dispensable}
\end{cases}.
\]
Since the diagonal edges of the Rotated LCS Grid Graph have a weight of $0$, we have $b_{i \cdot \mathcal{B}_x} \le a_{i \cdot \mathcal{B}_x}$ and thus $b_{i \cdot \mathcal{B}_x} \le b^{'}_{i \cdot \mathcal{B}_x}$. Therefore 
\[
\treerd_{M ^ 3, S / 3}(A, B) \le 
\treerd_{M ^ 3, S / 3}(A, B').
\]
Moreover, it is not hard to see that for every recursion interval $(l, r]$ we have
\[
\frac{1}{2^{\log^{0.01}(n)}} A_{\Sigma(l, r]} \le B'_{\Sigma(l, r]} \le A_{\Sigma(l, r]}.
\]
Thus, we apply \cref{lemma:treerd} on (the defined entries of) $A, B'$ with $\alpha = \frac{1}{2^{\log^{0.01}(n)}}$ to obtain
$$\treerd_{M ^ 3, S / 3}(A, B) \le \treerd_{M ^ 3, S / 3}(A, B') = O\left(\sum_{x}{b^{'}_x}\left(S ^ {0.51}\log ^ {0.01}(n)\right) + n / 2 ^ {\log_M ^ {0.015}(n)}\right),$$
Where the $n / 2 ^ {\log_M ^ {0.015}(n)}$ term is from
$$\underbrace{{2\mathcal{B}}_x}_{\text{max absolute value of $A$}} \cdot \underbrace{\left(n / {\mathcal{B}}_x\right)}_{\text{number of indices that are multiples of $\mathcal{B}_x$}} = O\left(n / 2 ^ {\log_M ^ {0.015}(n)}\right).$$
Since the number of dispensable columns is at most $n / 2^{\Omega\left(\log^{0.01}(n)\right)}$ and $w(p) = n / 2 ^ {o\left(\log ^ {0.01} (n)\right)}$, hence $\sum_{x}{b^{'}_x} = (1 + o(1))w(p)$ and $n / 2 ^ {\log_M ^ {0.015}(n)} \ll w(p)$. Thus we have
$$\frac{1}{w(p)}{\sum_{s \in [S]\cap 3\mathbb{Z}}\sum_{(l, r] \in s}{\rd_{M ^ 3}(A(l, r], B(l, r])}} = O(S ^ {0.51}\log ^ {0.01}(n)).$$
Therefore, by Markov's inequality, the expected number of $(l, r] \ni x$ failing \eqref{eq:nohit}
is at most $O(M ^ 3S ^ {0.51}\log^{0.01}(n)) = o\left(\frac{S^{0.98}}{M\log ^ {0.02}{(n)}}\right)$. Since we sample active scales with probability $\Theta(1 / S^{0.98})$, with probability $o\left(\frac{1}{M\log ^ {0.02}{(n)}}\right)$, every interval $(l, r] \ni x$ on an active scale $s$ satisfies \eqref{eq:nohit} (and thereby \eqref{eq:nohitraw} due to symmetry of the signs).
\end{proof}

Now we have finished proving \cref{lemma:regularlcs} and its ingredients.
In \cref{sec:subsamplinglcs}, we will present how we can approximate the length of the longest $\tilde{S}$-regular path $\hat{p}$ with a small multiplicative error with probability at least $3 / 4$:
\begin{restatable}{lemma}{lemmasubsamplinglcs} \label{lemma:subsamplinglcs}
Let $\tilde{S} \subset [S]$ be such that each scale that is a multiple of $3$ is chosen u.a.r.~with probability $1 / S ^ {0.98}$. Let $\hat{p}$ be the longest $\tilde{S}$-regular $\vertex{0}{n_X}$-$\vertex{n_X + n_Y}{n_Y}$ path with $w(\hat{p}) = n / 2 ^ {o\left(\log ^ {0.01} (n)\right)}$. With probability at least $3 / 4$, we can $(1 \pm o(1))$-approximate $w(\hat{p})$ in $o(n ^ {0.09})$ time by querying only $\frac{n ^ 2{\phi_y} ^ 2}{{{\mathcal{B}}_x} ^ 2} / 2 ^ {\Omega\left(\log ^ {0.01} (n)\right)}$ edge weights.
\end{restatable}

Assuming~\cref{lemma:subsamplinglcs}, we now show how to finish the rest of the proof for \cref{lemma:gridgraphmax}. 
\lemmagridgraphmax*
\begin{proof}
    Let $\tilde{S} \subset [S]$ be such that each scale that is a multiple of $3$ is chosen u.a.r.~with probability $1 / S ^ {0.98}$. Let $p$ be the longest $\vertex{0}{{n_X}}$-$\vertex{n_X + {n_Y}}{n_Y}$ path. Let $\hat{p}$ be the longest $\tilde{S}$-regular $\vertex{0}{{n_X}}$-$\vertex{n_X + {n_Y}}{n_Y}$ path with $w(\hat{p}) = n / 2 ^ {o\left(\log ^ {0.01} (n)\right)}$. 
 
    From \cref{lemma:regularlcs}, with probability at least $19 / 20$, we have 
    \begin{equation*}
        w(\hat{p}) = \left(1 - o(1)\right)w(p).
    \end{equation*}
    
    From \cref{lemma:subsamplinglcs}, with probability at least $3 / 4$, we can obtain $(1 \pm o(1))$-approximate $w(\hat{p})$ by querying only $\frac{n ^ 2\phi_y}{{{\mathcal{B}}_x}^2} / 2 ^ {\Omega\left(\log ^ {0.01} (n)\right)}$ edge weights. 

    By the union bound, with probability at least $7 / 10$, we can $\left(1 \pm o(1)\right)$-approximate $w(p)$ in $o(n ^ {0.09})$ time by querying only $\frac{n ^ 2{\phi_y} ^ 2}{{{\mathcal{B}}_x} ^ 2} / 2 ^ {\Omega\left(\log ^ {0.01} (n)\right)}$ edge weights. By scaling down the answer by $(1 - o(1))$ this becomes a $(1 - o(1))$-approximation. By repeating the procedure $\Theta(\log (n))$ times and taking the median, we can approximate $w(p)$ with a multiplicative error of $\left(1 - o(1)\right)$ with high probability in $o\left(n ^ {0.09} \log (n)\right) = o(n ^ {0.1})$ time by querying only $\frac{n ^ 2{\phi_y} ^ 2\log (n)}{{{\mathcal{B}}_x} ^ 2} / 2 ^ {\Omega\left(\log ^ {0.01} (n)\right)} = \frac{n ^ 2{\phi_y} ^ 2}{{{\mathcal{B}}_x} ^ 2} / 2 ^ {\Omega\left(\log ^ {0.01} (n)\right)}$ edge weights. 
\end{proof}
\subsection{Algorithm Description} \label{sec:subsamplinglcs}
Let $\tilde{S} \subset [S]$ be such that each scale that is a multiple of $3$ is chosen u.a.r.~with probability $1 / S ^ {0.98}$. Our goal now is to approximate the length of the longest $\tilde{S}$-regular $\vertex{0}{n_X}$-$\vertex{n_X + n_Y}{n_Y}$ path $\hat{p}$:
\lemmasubsamplinglcs*
We can design an $M$-way $S$-scale divide-and-conquer scheme to compute $\est[l, y_l, r, y_r]$ which is the length of the longest $\tilde{S}$-regular $\vertex{l}{y_l}$-$\vertex{r}{y_r}$ path where $(l, r]$ is a scale-$s$ interval. 
\begin{itemize}
    \item If $s = 0$, we obtain $\est[l, y_l, r, y_r]$ using an edge weight query;
    \item Suppose that $s > 0$ and $s \notin \tilde{S}$. We find $\est[l, y_l, r, y_r]$ by maximizing over every anchor path defined by the sequence $\hat{Y} = (y_l  = \hat{y}_0, \hat{y}_1, \ldots, \hat{y}_M = y_r)$ where for $1 \le i \le M$, $\abs{\hat{y}_i - \hat{y}_{i - 1}} \le \frac{r - l}{M}$, visiting valid vertices $\left(\vertex{{\anc}_i}{\hat{y}_i}\right)_i$ in the sparsified grid graph. We use $\anchor$ to denote the set of possible anchor paths $\hat{Y}$.
    \item Suppose $s \in \tilde{S}$. Recall that we regularize the row at column ${\anc}_i$ to $\overline{y}_i := (I_{s - 1} / \phi_y) \cdot \round{\frac{y_l + i \cdot \frac{y_r - y_l}{M}}{I_{s - 1} / \phi_y}}$.
    $\est[l, y_l, r, y_r]$ is simply the sum of $\est[{\anc}_{i - 1}, \overline{y}_i, {\anc}_i, \overline{y}_i]$ over $1 \le i \le M$.  
\end{itemize} 
The full recursion is as follows:
\begin{equation*}
\est[l, y_l, r, y_r] := 
\begin{cases}
	\ew[l, y_l, r, y_r], & \text{if $s = 0$} \\
	\max\limits_{\hat{Y} \in \anchor}\sum\limits_{i = 1} ^ M{\est[{\anc}_{i - 1}, \hat{y}_{i - 1}, {\anc}_i, \hat{y}_i]}, & \text{if $s \notin \tilde{S}$} \\
	\sum\limits_{i = 1} ^ M{\est[{\anc}_{i - 1}, \overline{y}_i, {\anc}_i, \overline{y}_i]}, & \text{if $s \in \tilde{S}$}
\end{cases}.
\end{equation*}

Unfortunately, this scheme still queries too many edge weights. The key to efficiency is to use a \emph{sub-sampling} technique for active scales. We first show a naitve attempt that can guarantee a small additive error (we later improve it to also obtain a bound on the multiplicative error). Notice that if $s \in \tilde{S}$, since no overfitting can happen when the path is fixed, we can just sample half of these values and double our sum from the sample. Thus, for every scale-$s$ interval $(l, r]$, we first uniformly randomly sample $M$ 0-1 variables $\eta_{l, r, 1 \ldots M}$ containing exactly $M / 2$ 1's. We use the following scheme:
\begin{equation} \label{eq:estasterisklcs}
\whest ^ *[l, y_l, r, y_r] := 
\begin{cases}
	\ew[l, y_l, r, y_r], & \text{if $s = 0$} \\
	\max\limits_{\hat{Y} \in \anchor}\sum\limits_{i = 1} ^ M{\whest ^ *[{\anc}_{i - 1}, \hat{y}_{i - 1}, {\anc}_i, \hat{y}_i]}, & \text{if $s \notin \tilde{S}$} \\
	2\sum\limits_{i = 1} ^ M{\eta_{l, r, i}\whest ^ *[{\anc}_{i - 1}, \overline{y}_i, {\anc}_i, \overline{y}_i]}, & \text{if $s \in \tilde{S}$}
\end{cases},
\end{equation}
Note that for every active scale, half of the sub-intervals $({\anc}_{i - 1}, {\anc}_i]$ for which $\eta_{l, r, i} = 0$ get discarded. Over $\abs{\tilde S}$ active scales, only $2 ^ {-\abs{\tilde S}}$ fraction of the intervals remain, which is a significant speed-up.

Unfortunately, the scheme with \eqref{eq:estasterisklcs} is not enough for a small multiplicative error, as $\whest ^ *[{\anc}_{i - 1}, \overline{y}_{i - 1}, {\anc}_i, \overline{y}_i]$ may be imbalanced over $1 \le i \le M$. In the most extreme example, if $\whest ^ *[{\anc}_{i - 1}, \overline{y}_{i - 1}, {\anc}_i, \overline{y}_i] \ne 0$ for only a single \emph{outlier} $i$, then \eqref{eq:estasterisklcs} either gets us an estimate that is $0$ or an estimate that is twice the actual value, depending on the value of $\eta_{l, r, i}$. 

Due to the bound on tree total deviation, the weight distribution on the longest path should intuitively not be imbalanced everywhere and there are not too many outliers in total. Thus we can afford to have some large error when there are outliers. 

\xiao{Motivation} However, if our estimator consistently overestimates the answer (by a large factor) when there is an outlier, a path can overfit our estimator by specifically aiming for the regions that contain a large number of outliers. This means that our estimator has to remain sound (i.e., does not overestimate) even at the presence of outliers. To do this, we clip down the values of the outliers before obtainin our estimate. Given a non-negative sequence $A = (a_1, a_2, \ldots, a_M)$, we design a way to approximate $\sum_i{a_i}$ when no $a_i$ is too large compared to the mean of $A$ using the clipping idea above:
\begin{restatable}{claim}{claimonesidedlcs} \label{claim:onesidedlcs}
    Let $A = (a_1, a_2, \ldots, a_M)$ be a sequence of non-negative numbers. Suppose that we only know $A_I$, that is, the sub-sequence $(a_i)_{i \in I}$ for a uniformly randomly chosen subset $I \subset [M]$ of $M/2$ indices.
    
    We obtain a non-negative estimate $\ESTSUM(A_I)$ for the sum $\sum\limits_{i = 1} ^ M{a_i}$ in the following way:
    \begin{enumerate}
        \item Estimate the mean of $A$ using $\tilde{A} := \sum_{i \in I}{a_i} / \left(\frac{M}{2}\right)$;
        \item Estimate the sum of $A$ by clipping outliers down: $\ESTSUM(A_I) := 2\sum_{i \in I}{\min\left(a_i, 2\log ^ {0.02}{(n)}\tilde{A}\right)}$.
    \end{enumerate}
    The pseudocode for computing $\ESTSUM(A_I)$ is shown in \cref{alg:eta}.

    Then with probability $1 - \exp\left(-\omega\left(\log ^ {0.02} (n)\right)\right)$, the estimate we obtain satisfies
    \begin{description}
        \item[Upper bound] $\ESTSUM(A_I) \le \left(1 + O\left(1 / \log ^ {0.02} (n)\right)\right)\sum\limits_{i = 1} ^ M{a_i}$;
        \item[Lower bound if there are no outliers] Let the true mean of $A$ be $\overline{A} := \sum_i{a_i} / M$. If we have $\max_i{a_i} \le \log ^ {0.02}{(n)}\overline{A}$, then 
    $$\ESTSUM(A_I) \ge \left(1 - O\left(1 / \log ^ {0.02} (n)\right)\right) \cdot \left(\sum\limits_{i = 1} ^ M{a_i}\right).$$
    \end{description}
\end{restatable}
\begin{algorithm} 
    \caption{Computation of $\ESTSUM(A_I)$ via Clipping} \label{alg:eta}
    \begin{algorithmic}[1]
        \Function{$\ESTSUM$}{$A_I$}
            \Statex \textbullet~\textbf{requirement 1:} $I \subset [M]$ is a uniformly randomly chosen subset of $M/2$ indices
            \Statex \textbullet~\textbf{requirement 2:} $A_I$ is the sub-sequence $(a_i)_{i \in I}$ of length-$M$ sequence $A$ of non-negative numbers
            \Statex \textbullet~\textbf{returns:} an estimate of $\sum\limits_{i = 1} ^ M{a_i}$ satisfying the conditions in \cref{claim:onesidedlcs}
            \State $\tilde{A} \gets \sum_{i \in I}{a_i} / \left(\frac{M}{2}\right)$
            \State \Return $2\sum_{i \in I}{\min\left(a_i, 2\log ^ {0.02}{(n)}\tilde{A}\right)}$
        \EndFunction
    \end{algorithmic}
\end{algorithm}
Using the sub-routine in \cref{claim:onesidedlcs}, we compute the estimate in the following way:
\begin{equation*}
\whest[l, y_l, r, y_r] := 
\begin{cases}
	\ew[l, y_l, r, y_r], & \text{if $s = 0$} \\
	\max\limits_{\hat{Y} \in \anchor}\sum\limits_{i = 1} ^ M{\whest[{\anc}_{i - 1}, \hat{y}_{i - 1}, {\anc}_i, \hat{y}_i]}, & \text{if $s \notin \tilde{S}$} \\
	\ESTSUM\left(\left(\whest[{\anc}_{i - 1}, \overline{y}_i, {\anc}_i, \overline{y}_i]\right)_{i \mid \eta_{l, r, i} = 1}\right), & \text{if $s \in \tilde{S}$}
\end{cases},
\end{equation*}

To efficiently maximize over anchor path $\hat{Y}$ on passive scale $s$, consider the following ``Scale-$(s - 1)$ Sparsified Grid Graph'':
\begin{definition} [Scale-$(s - 1)$ Sparsified Grid Graph] \label{def:scalesparselcs}
    The Scale-$(s - 1)$-Sparsified Grid Graph defined on $\whest$ is as follows:
    \begin{description}
        \item[Vertices] We only consider vertices $\vertex{x}{y}$ in the Sparsified LCS Grid Graph with the additional constraint that $I_{s - 1} \mid x$;
        \item[Edges] There is an edge from $\vertex{x}{y}$ to $\vertex{x + I_{s - 1}}{y'}$ (if both exist) such that $\abs{y' - y} \le I_{s - 1}$ with a weight of $\whest[x, y, x + I_{s - 1}, y']$.
    \end{description}
\end{definition} 
To obtain $\max\limits_{\hat{Y} \in \anchor}\sum\limits_{i = 1} ^ M{\whest[{\anc}_{i - 1}, \hat{y}_{i - 1}, {\anc}_i, \hat{y}_i]}$, we can simply find the longest $\vertex{l}{y_l}$-$\vertex{r}{y_r}$ path. 

We give the full pseudocode for computing $\whest$ in \cref{alg:whestlcs}. We use the memoization technique to compute $\whest$ values on demand (i.e., Line \ref{line:memoizationstartlcs}-\ref{line:memoizationendlcs} and Line \ref{line:memoizationstart2lcs}-\ref{line:memoizationend2lcs}).
\begin{algorithm} 
    \caption{Computation of $\whest$ for LCS} \label{alg:whestlcs}
    \begin{algorithmic}[1]
        \Function{$\COMESTLCS$}{$l, y_l, r, y_r, s$}
            \Statex \textbullet~\textbf{requirement 1:} $\vertex{l}{y_l}$ and $\vertex{r}{y_r}$ are both vertices in the sparsified grid graph
            \Statex \textbullet~\textbf{requirement 2:} $\vertex{r}{y_r}$ is reachable from $\vertex{l}{y_l}$
            \Statex \textbullet~\textbf{requirement 3:} $s \in [S] \cup \{0\}$
            \Statex \textbullet~\textbf{requirement 4:} $(l, r]$ is a scale-$s$ interval
            \Statex \textbullet~\textbf{returns:} an estimate $\whest[l, y_l, r, y_r]$ for the longest $\vertex{l}{y_l}$-$\vertex{r}{y_r}$ path
            \If {$s = 0$}
                \Return $\ew[l, y_l, r, y_r]$
            \ElsIf {$s \notin \tilde{S}$}
                \State $\ver \gets \emptyset $\Comment{Initialize vertices for Scale-$(s - 1)$ Sparsified Grid Graph}
                \State $\e \gets \emptyset $\Comment{Initialize edges for Scale-$(s - 1)$ Sparsified Grid Graph}
                \For {$x \in $ multiples of $I_{s - 1}$ in $(l, r]$}
                    \For {$y \in $ multiples of $I_{s - 1} / \phi_y$ in $[0, n]$}
                        \State add $\vertex{x}{y}$ to $\ver$
                        \If {$x + I_{s - 1} \le r$}
                            \For {$y' \in $ multiples of $I_{s - 1} / \phi_y$ in $[0, n]$}
                                \If {$\abs{y' - y} \le I_{s - 1}$}
                                    \If {$\whest[x, y, x + I_{s - 1}, y'] = \perp$} \Comment{Estimate not yet computed} \label{line:memoizationstartlcs}
                                        \State $\whest[x, y, x + I_{s - 1}, y'] \gets \COMESTLCS\left(x, y, x + I_{s - 1}, y', s - 1\right)$ \Comment{Compute estimate for sub-interval}
                                    \EndIf \label{line:memoizationendlcs}
                                    \State add an edge from $\vertex{x}{y}$ to $\vertex{x + I_{s - 1}}{y'}$ with weight $\whest[x, y, x + I_{s - 1}, y']$ to $\e$
                                \EndIf
                            \EndFor
                        \EndIf
                    \EndFor
                \EndFor
                \State $G \gets \langle \ver, \e \rangle$ \Comment{Create the Scale-$(s - 1)$ Sparsified Grid Graph}
                \State \Return length of longest $\vertex{l}{y_l}$-$\vertex{r}{y_r}$ path in $G$ via standard techniques for acyclic graphs
            \Else
                \State $\overline{y}_i \gets (I_{s - 1} / \phi_y) \cdot \round{\frac{y_l + i \cdot \frac{y_r - y_l}{M}}{I_{s - 1} / \phi_y}}$ for $1 \le i \le M$
                \State ${\anc}_i \gets l + i \cdot \frac{r - l}{M}$ for $1 \le i \le M$
                \For {$i \in [1, M]$}  
                    \If {$\eta_{l, r, i} = 1$} \Comment{$i$-th sub-interval is in sample}
                        \If {$\whest[{\anc}_{i - 1}, \overline{y}_{i - 1}, {\anc}_i, \overline{y}_{i}] = \perp$} \Comment{Estimate not yet computed} \label{line:memoizationstart2lcs}
                            \State $\whest[{\anc}_{i - 1}, \overline{y}_{i - 1}, {\anc}_i, \overline{y}_{i}] \gets \COMESTLCS\left({\anc}_{i - 1}, \overline{y}_{i - 1}, {\anc}_i, \overline{y}_i, s - 1\right)$ \Comment{Compute estimate for sub-interval}
                        \EndIf \label{line:memoizationend2lcs}
                    \EndIf
                \EndFor
                \State \Return $\ESTSUM\left(\left(\whest[{\anc}_{i - 1}, \overline{y}_i, {\anc}_i, \overline{y}_i]\right)_{\eta_{l, r, i} = i \mid i \in [M]}\right)$
            \EndIf
        \EndFunction
    \end{algorithmic}
\end{algorithm}

\subsection{Algorithm Analysis} \label{sec:analysislcs}
\subsubsection{Overview}
Our goal is to use our algorithm to prove \cref{lemma:subsamplinglcs}:
\lemmasubsamplinglcs*
The analysis of our algorithm proof consists of three parts:
\begin{restatable}[Our estimate is complete]{lemma}{lemmacompletelcs}
    Let $\hat{p}$ be the longest $\tilde{S}$-regular $\vertex{0}{n_X}$-$\vertex{n_X + n_Y}{n_Y}$ path.
     With probability at least $19 / 25$, $\whest[0, n_X, n_X + n_Y, n_Y] > \left(1 - o(1)\right)w(\hat{p})$.
\end{restatable}
\begin{restatable}[Our estimate is sound]{lemma}{lemmasoundlcs}
    Let $\hat{p}$ be the longest $\tilde{S}$-regular $\vertex{0}{n_X}$-$\vertex{n_X + n_Y}{n_Y}$ path. Suppose $w(\hat{p}) = n / 2 ^ {o\left(\log ^ {0.01} (n)\right)}$.
    With probability $1 - o(1)$, $\whest[0, n_X, n_X + n_Y, n_Y] < \left(1 + o(1)\right)w(\hat{p})$;
\end{restatable}
\begin{restatable}[Our estimates can be computed efficiently]{lemma}{lemmaefficiencylcs}
    With probability $1 - o(1)$, to compute $\whest[0, n_X, n_X + n_Y, n_Y]$, the running time is at most $o(n ^ {0.09})$, and the total number of edge weight queries is at most $\frac{n ^ 2{\phi_y} ^ 2}{{{\mathcal{B}}_x} ^ 2} / 2 ^ {\Omega\left(\log ^ {0.01} (n)\right)}$.
\end{restatable}
By union bound, we can see that our algorithm is both correct and efficient with probability at least $3 / 4$, thereby proving \cref{lemma:subsamplinglcs}.
\subsubsection{Our Estimate Is Complete}
We first show that our algorithm is complete. Recall that our sub-routine for obtaining the estimate has an error that is contingent on there being no outliers (\cref{claim:onesidedlcs}). To show that outliers are rare on average, we will use the following ingredient, which shows that the deviation of the weight distribution of $\hat{p}$ is small on active layers:
\begin{restatable}[The deviation of the weight distribution of $\hat{p}$ is small on active layers]{claim}{claimtechnicalitylcs} \label{claim:technicalitylcs}    
    Assume that the length of the longest $\vertex{0}{n_X}$-$\vertex{n_X + n_Y}{n_Y}$ path is at least $n / 2 ^ {o\left(\log ^ {0.02} (n)\right)}$.
    Let $\tilde{S} \subset [S]$ be such that each scale that is a multiple of $3$ is chosen u.a.r.~with probability $1 / S ^ {0.98}$. Let the longest $\tilde{S}$-regular $\vertex{0}{n_X}$-$\vertex{n_X + n_Y}{n_Y}$ path be $\hat{p}$ with $w(\hat{p}) = \Omega(w(p))$. Recall that $W(\hat{p})$ is the sequence of length $n / {\mathcal{B}}_x$ whose $i$-th entry is:
    \begin{equation} \label{eq:defimprovement}
    W(\hat{p})_i := w(\hat{p}(i{\mathcal{B}}_x, (i + 1){\mathcal{B}}_x]),
    \end{equation}
    which captures the ``weight distribution'' of $\hat{p}$.

    With probability at least $4 / 5$, we have $\treemd_{M, \tilde{S}}(W(\hat{p})) = o\left(\frac{w(\hat{p})}{M}\right)$. 
\end{restatable}

Before proving the ingredient above, we first give our main analysis that our estimate is complete:
\lemmacompletelcs*
\begin{proof}
For each active scale $s$, we consider the following question: how much would we underestimate if there were no subsampling on active scales larger than $s$? We can see that the final estimate is at least:
$$\mathcal{P}_s := \sum\limits_{i = 1} ^ {n / I_s}{\whest[(i - 1) \cdot I_s, \hat{p}((i - 1) \cdot I_s), i \cdot I_s, \hat{p}(i \cdot I_s)]}.$$
We can see that $\mathcal{P}_0 = w(\hat{p})$, and that $\mathcal{P}_S = \whest[0, n_X, n_X + n_Y, n_Y]$, namely our final estimate for $w(\hat{p})$. We want to know how much smaller $\mathcal{P}_S$ can be compared to $\mathcal{P}_0$.

The important part of the upcoming analysis will be to analyze the error due to outliers. Namely, we want to argue on average, for an interval $(l, r]$ on an active scale, the estimates from sub-intervals, or $\whest[{\anc}_{i - 1}, \hat{p}({\anc}_{i - 1}), {\anc}_i, \hat{p}({\anc}_i)]$'s, are ``balanced,'' as is required by \cref{claim:onesidedlcs}. The main technicality is that there is already error inside each $\whest[{\anc}_{i - 1}, \hat{p}({\anc}_{i - 1}), {\anc}_i, \hat{p}({\anc}_i)]$, and the errors can have a cascading effect: if estimates inside a subset of sub-intervals have a large error, it can make the estimates across sub-intervals imbalanced, thereby also giving the parent interval a large error.  

To better analyze this cascading effect, we first obtain a lower bound $\whest'$ for $\whest$ in the following way:
\begin{equation*}
\whest'[l, y_l, r, y_r] := 
\begin{cases}
	\ew[l, y_l, r, y_r], & \text{if $s = 0$} \\
	\max\limits_{\hat{Y} \in \anchor}\sum\limits_{i = 1} ^ M{\whest'[{\anc}_{i - 1}, \hat{y}_{i - 1}, {\anc}_i, \hat{y}_i]}, & \text{if $s \notin \tilde{S}$} \\
	\min\left(w(\hat{p}(l, r]), \ESTSUM\left(\left(\whest'[{\anc}_{i - 1}, \overline{y}_i, {\anc}_i, \overline{y}_i]\right)_{\eta_{l, r, i} = i \mid i \in [M]}\right)\right) & \text{if $s \in \tilde{S}$}
\end{cases},
\end{equation*}
Namely, we do the same recursion scheme as in $\whest$, but on active scales, we clamp our estimates on interval $(l, r]$ down to at most $w(\hat{p}(l, r])$. Note that this clamping is not algorithmic and we use it for analysis purposes only. The purpose of this lower bound is that it makes errors one-directional: for every recursion interval $(l, r]$ we have $\whest'[l, \hat{p}(l), r, \hat{p}(r)] \le w(\hat{p}(l, r])$. Moreover, $\whest'$ is indeed a lower bound for $\whest$: observe that $\ESTSUM$ is monotone. Namely if we have two sets of inputs $A_I, A^{'}_I$ such that $a_i \le a^{'}_i$ for every $i \in I$, then $\ESTSUM(A_I) \le \ESTSUM(A^{'}_I)$. This means that we always have $\whest'[l, y_l, r, y_r] \le \whest[l, y_l, r, y_r]$. 

For each active scale $s$, if there were no subsampling on active scales larger than $s$, we can see that the final estimate would be at least:
$${\mathcal{P}}^{'}_s := \sum\limits_{i = 1} ^ {n / I_s}{\whest'[(i - 1) \cdot I_s, \hat{p}((i - 1) \cdot I_s), i \cdot I_s, \hat{p}(i \cdot I_s)]}.$$
We can see that ${\mathcal{P}}^{'}_0 = w(\hat{p})$, and that ${\mathcal{P}}^{'}_S = \whest'[0, n_X, n_X + n_Y, n_Y] \le \whest[0, n_X, n_X + n_Y, n_Y]$, namely ${\mathcal{P}}^{'}_s$ is a lower bound for our final estimate for $w(\hat{p})$. To see how much smaller ${\mathcal{P}}^{'}_S$ can be compared to ${\mathcal{P}}^{'}_0$, we go through each scale $s$ and examine the difference between ${\mathcal{P}}^{'}_{s - 1}$ and ${\mathcal{P}}^{'}_s$. 

If scale-$s$ is passive, then since we compute every $\whest'[l, \hat{p}(l), r, \hat{p}(r)]$ to be the length of the longest $\tilde{S}$-regular $\vertex{l}{\hat{p}(l)}$-$\vertex{r}{\hat{p}(r)}$ path in the Scale-$(s - 1)$ Sparsified Grid Graph (defined on $\whest'$), ${\mathcal{P}}^{'}_s \ge {\mathcal{P}}^{'}_{s - 1}$.  \Aviad{should they be $=$?}\xiao{no because with the roundings and stuff the optimal path might no longer be $\hat{p}$}

If scale-$s$ is active, fix a scale-$s$ interval $(l, r]$. From \cref{claim:onesidedlcs}, we can have at most the following types of errors:
\begin{itemize}
    \item Type I: A multiplicative error of $1 - O\left(\frac{1}{\log ^ {0.02} (n)}\right)$;
    \item Type II: If $\max_i{\whest'[{\anc}_{i - 1}, \hat{p}({\anc}_{i - 1}), {\anc}_i, \hat{p}({\anc}_i)]} = \Omega\left(\frac{\log ^ {0.02} (n)}{M}\sum_i{\whest'[{\anc}_{i - 1}, \hat{p}({\anc}_{i - 1}), {\anc}_i, \hat{p}({\anc}_i)]}\right)$, an additive underestimation of at most $w(\hat{p}(l, r])$ (since our estimate is non-negative);
    \item Type III: With probability $\exp\left(-\omega\left(\log ^ {0.02}(n)\right)\right)$, an additive underestimation of at most $w(\hat{p}(l, r])$.
\end{itemize}

We can first see that with probability $1 - o(1)$ the multiplicative error accumulated over $\tilde{S}$ (whose size is $O(S ^ {0.02}) = o\left(\log ^ {0.02} (n)\right)$ with probability $1 - o(1)$) active scales is $1 - o(1)$. 

The total sum of type III error is at most $\exp\left(-\omega\left(\log ^ {0.02}(n)\right)\right)$ times the total sum of $w(\hat{p}(l, r])$ on every interval $(l, r]$ on an active scale. The total sum of $w(\hat{p}(l, r])$ is at most $O(w(\hat{p})\abs{\tilde{S}})$ and the total total sum of type III error is at most $w(\hat{p}) \cdot \abs{\tilde{S}} \cdot \exp\left(-\omega\left(\log ^ {0.02}(n)\right)\right) = o(w(\hat{p}))$. 

The most technically involved part of this proof is the analysis of the type II error. Since our estimate $\whest'[{\anc}_{i - 1}, \hat{p}({\anc}_{i - 1}), {\anc}_i, \hat{p}({\anc}_i)]$ can be different from $w(\hat{p}({\anc}_{i - 1}, {\anc}_i])$, the errors can have a cascading effect. 
We say an interval is ``imbalanced'' if $\max_i{w(\hat{p}({\anc}_{i - 1}, {\anc}_i])} > \frac{2w(\hat{p}(l, r])}{M}$ and ``balanced'' otherwise. Suppose $(l, r]$ is ``balanced.'' Namely the weight distribution of $\hat{p}$ inside $(l, r]$ is balanced. We now argue that the cascading effect is insignificant for ``balanced'' interval.
 We now show that in order for there to be a type II error on ``balanced'' $(l, r]$, namely for the distribution of $\whest'$ to be imbalanced despite the actual answer distribution being balanced, the $\whest'$ values must have decayed so much that they are small enough to be ignored. Formally, notice that every $\whest'[{\anc}_{i - 1}, \hat{p}({\anc}_{i - 1}), {\anc}_i, \hat{p}({\anc}_i)] \le w(\hat{p}({\anc}_{i - 1}, {\anc}_i])$. Thus type II error only exists if 
$$\max_i{w(\hat{p}({\anc}_{i - 1}, {\anc}_i])} =  \Omega\left(\frac{\log ^ {0.02} (n)}{M}\sum_i{\whest'[{\anc}_{i - 1}, \hat{p}({\anc}_{i - 1}), {\anc}_i, \hat{p}({\anc}_i)]}\right).$$ 
Since we assume $(l, r]$ is ``balanced,'' we have $\max_i{w(\hat{p}({\anc}_{i - 1}, {\anc}_i])} \le \frac{2w(\hat{p}(l, r])}{M}$. This means that 
$$\sum_i{\whest'[{\anc}_{i - 1}, \hat{p}({\anc}_{i - 1}), {\anc}_i, \hat{p}({\anc}_i)]} = O\left(\frac{w(\hat{p}(l, r])}{\log ^ {0.02} {(n)}}\right).$$
The worst total extra additive error incurred this way inside a ``balanced'' $(l, r]$ on scale $s$ is not worse than if the estimation simply decays to $0$ completely, but even then the error is only $O\left(\frac{w(\hat{p}(l, r])}{\log ^ {0.02} {(n)}}\right)$.

Therefore, the total magnitude of type II error on scale $s$ is at most the sum of:
\begin{description}
\item [Type II(a) -- ``imbalanced'' intervals] the sum of $w(\hat{p}(l, r])$ over ``imbalanced'' interval $(l, r]$ on scale $s$;
\item [Type II(b) -- ``balanced'' intervals] the sum of $w(\hat{p}(l, r])$ over ``balanced'' interval on scale $s$, which we now know is at most $O\left(\frac{w(\hat{p})}{\log ^ {0.02} {(n)}}\right)$.
\end{description}
Summing over active scales, we can see that with probability $1 - o(1)$, the total type II(b) error is at most $O\left(\frac{w(\hat{p})}{\log ^ {0.02} {(n)}} \cdot \abs{\tilde{S}}\right) = o(w(\hat{p}))$.

We now bound the sum of type II(a) error on active scales, namely the sum of $w(\hat{p}(l, r])$ over ``imbalanced'' interval $(l, r]$ on active scales. If an interval $(l, r]$ is ``imbalanced,'' namely $\max_i{w(\hat{p}({\anc}_{i - 1}, {\anc}_i])} > \frac{2w(\hat{p}(l, r])}{M}$, then $$\md_M(W(\hat{p}(l, r])) = \sum_i{\abs{w(\hat{p}({\anc}_{i - 1}, {\anc}_i]) - \frac{w(\hat{p}(l, r])}{M}}} = \Omega\left(\frac{w(\hat{p}(l, r])}{M}\right).$$ Thus $w(\hat{p}(l, r]) = O(M\md_M(W(\hat{p}(l, r])))$, which means that the sum of $w(\hat{p}(l, r])$ over ``imbalanced'' interval $(l, r]$ on active scales is $O(M\treemd_{M, \tilde{S}}(W(\hat{p})))$, which is $o(w(\hat{p}))$ from \cref{claim:technicalitylcs} with probability $4 / 5$.

Summing up all errors and taking the union bound over the probability that $\abs{\tilde{S}} = O(S ^ {0.02})$, we can see that the lemma indeed holds with probability $4 / 5 - o(1) > 19 / 25$.
\end{proof}

We now prove the auxiliary claim from before:
\claimtechnicalitylcs*
\begin{proof}
    We first show that $\treemd_{M, \tilde{S}}(W(\textcolor{red}{p})) = o(w(p) / M)$. Then we argue that the difference between $\treemd_{M, \tilde{S}}(W(\textcolor{red}{\hat{p}}))$ and $\treemd_{M, \tilde{S}}(W(\textcolor{red}{p}))$ is also $o(w(p) / M)$.
    \paragraph{Part I: $\treemd_{M, \tilde{S}}(W(p))$ is small} Using the bound on tree total deviation in \cref{lemma:totalmeandeviation} we know that, 
    \begin{align*} 
    \treemd_{M, S}(W(p)) &= O\left(w(p) \cdot S ^ {0.51}  + \underbrace{{\mathcal{B}}_x}_{\text{max length of $p$ over scale $0$ interval}} \cdot \underbrace{\left(n / {\mathcal{B}}_x\right)}_{\text{number of scale $0$ intervals}} /  2 ^ {\log_M ^ {0.015}(n)}\right) \\
    &= O\left(w(p) \cdot S ^ {0.51}  + n /  2 ^ {\log_M ^ {0.015}(n)}\right)
    \end{align*}
    Since each scale that is a multiple of $3$ is in $\tilde{S}$ u.a.r~with probability $1 / S ^ {0.98}$, and $p$ does not depend on $\tilde{S}$, we have that with probability at least $9 / 10$,
    \begin{align*}
        \treemd_{M, \tilde{S}}(W(p)) &= O\left(\left(w(p) \cdot S ^ {0.51} + \textcolor{red}{n /  2 ^ {\log_M ^ {0.015}(n)}}\right) / S ^ {0.98}\right) \\
        &= \left(O\left(\left(w(p) \cdot S ^ {0.51}\right) + \textcolor{red}{o\left(\frac{w(p)}{M}\right)}\right)\right) / S ^ {0.98} && \text{since $w(p) = n / 2 ^ {o\left(\log ^ {0.01} (n)\right)}$} \\
        &= w(p) / S ^ {0.47} + o\left(\frac{w(p)}{M}\right) \\
        &= o\left(\frac{w(p)}{M}\right). && \text{since $M = \Theta\left(\log ^ {\mval} (n)\right)$}
    \end{align*}

    \paragraph{Part II: $\treemd_{M, \tilde{S}}(W(\hat{p}))$ is not too different from $\treemd_{M, \tilde{S}}(W(p))$} From \cref{lemma:regularlcs}, with probability at least $19 / 20$. We can round $p$ to an $\tilde{S}$-regular path $\hat{p}$ such that there are at most $o\left(\frac{w(p)}{M\log ^ {0.02}(n)}\right)$ columns $x$ before which their weights differ. Each of these columns contributes at most $O(\abs{\tilde{S}})$ to the difference between $\treemd_{M, \tilde{S}}(W(p))$ and $\treemd_{M, \tilde{S}}(W(\hat{p}))$. Since the size of $\tilde{S}$ is $O(S ^ {0.02}) = o\left(\log ^ {0.02} (n)\right)$ with probability $1 - o(1) > 19 / 20$, by the union bound, with probability at least $9 / 10$, the difference between $\treemd_{M, \tilde{S}}(W(p))$ and $\treemd_{M, \tilde{S}}(W(\hat{p}))$ is at most $o\left(\frac{w(p)}{M}\right)$.
    
    Combining the two parts, with probability at least $4 / 5$, we have $\treemd_{M, \tilde{S}}(W(\hat{p})) = o\left(\frac{w(\hat{p})}{M}\right)$. 
\end{proof}

\subsubsection{Our Estimate Is Sound}
Now we show that our estimate is sound:
\lemmasoundlcs*
\begin{proof}   
For each active scale $s$, we consider the following question: how much would we overestimate if there were no subsampling on active scales larger than $s$? To capture this, consider the Scale-$s$ Sparsified Grid Graph $\mathcal{G}_s$ defined on $\whest$ (\cref{def:scalesparselcs}), namely the following graph:
\begin{description}
    \item[Vertices] We only consider vertices $\vertex{x}{y}$ in the Sparsified LCS Grid Graph with the additional constraint that $I_{s} \mid x$;
    \item[Edges] There is an edge from $\vertex{x}{y}$ to $\vertex{x + I_{s}}{y'}$ (if both exist) such that $\abs{y' - y} \le I_{s}$ with a weight of $\whest[x, y, x + I_{s}, y']$.
\end{description}
Let $\mathcal{A}_s$ be the length of the longest $\tilde{S}$-regular $\vertex{0}{{n_X}}$-$\vertex{n_X + {n_Y}}{n_Y}$ path in $\mathcal{G}_s$. We can see that $\mathcal{A}_0 = w(\hat{p})$, and that $\mathcal{A}_S = \whest[0, n_X, n_X + n_Y, n_Y]$, namely our final estimate for $w(\hat{p})$. To see how much larger $\mathcal{A}_S$ \Aviad{$\mathcal{A}_{|S|}$? (missing $|\cdot|$)}\xiao{$S$ is a number $\hat{S}$ is a set} can be compared to $\mathcal{A}_0$, we go through each scale $s$ and examine the difference between $\mathcal{A}_{s - 1}$ and $\mathcal{A}_s$.

If scale-$s$ is passive, then since we compute $\whest[l, y_l, r, y_r]$ to be the length of the longest $\tilde{S}$-regular $\vertex{l}{y_l}$-$\vertex{r}{y_r}$ path in the Scale-$(s - 1)$ Sparsified Grid Graph, $\mathcal{A}_s = \mathcal{A}_{s - 1}$.

If scale-$s$ is active, then from \cref{claim:onesidedlcs}, for $\whest[l, y_l, r, y_r]$ such that $(l, r]$ is a scale-$s$ interval:
\begin{itemize}
    \item The length of the longest $\tilde{S}$-regular $\vertex{l}{y_l}$-$\vertex{r}{y_r}$ path in $\mathcal{G}_{s - 1}$ is $X := \sum\limits_{i = 1} ^ M{\whest[{\anc}_{i - 1}, \overline{y}_i, {\anc}_i, \overline{y}_i]}$;
    \item The length of the longest $\tilde{S}$-regular $\vertex{l}{y_l}$-$\vertex{r}{y_r}$ path in $\mathcal{G}_{s}$, namely $\whest[l, y_l, r, y_r]$, is, with probability $1 - \exp\left(-\omega\left(\log ^ {0.02} (n)\right)\right)$, at most $\left(1 + O(1 / \log ^ {0.02} {(n)})\right)X$. Notice that the our estimate for the length of a path from column $l$ to column $r$ cannot be more than the ``width'' $r - l = I_s$. Even if the error guarantee fails, the overestimation is at most $I_s$.
\end{itemize}
Thus, going from $\mathcal{A}_{s - 1}$ to $\mathcal{A}_s$, we have:
\begin{itemize}
    \item Type I: per scale-$s$ interval $(l, r]$, multiplicative overestimation of at most $\left(1 + O(1 / \log ^ {0.02} {(n)})\right)$;
    \item Type II: for $\whest[l, y_l, r, y_r]$ such that $(l, r]$ is a scale-$s$ interval, there is a probability of $\exp\left(-\omega\left(\log ^ {0.02} (n)\right)\right)$ that $\whest[l, y_l, r, y_r]$ has an overestimation of at most $I_s$.
\end{itemize}
We now prove that with probability $1 - o(1)$, either error across all scales and intervals is $o(w(\hat{p}))$. 

Recall that $|\tilde{S}| = O(S ^ {0.02}) = o\left(\log ^ {0.02} n\right)$ w.p.~$1 - o(1)$.
Thus the total type I multiplicative overestimation across all $\abs{\tilde{S}}$  active scales is $(1 + o(1))$ w.p.~$1 - o(1)$.

For type II overestimation, we say an estimate $\whest[l, y_l, r, y_r]$ is ``bad'' if it exhibits an overestimation of $(1 + \Omega(1 / \log ^ {0.02} (n)))$. the total overestimation from estimates being bad is upper bounded by the following model:
\begin{definition} [Upper Bound for Overestimation due to ``Bad'' Estimates]
    Consider the following graph:
    \begin{description}
        \item[Vertices] We only consider vertices $\vertex{x}{y}$ in the Sparsified LCS Grid Graph with the additional constraint that $I_s \mid x$;
        \item[Edges] There is an edge from $\vertex{x}{y}$ to $\vertex{x + I_s}{y'}$ such that $\abs{y' - y} \le I_s$ with a weight of either $I_s$, if $\whest[x, y, x + I_s, y']$ is ``bad,'' or $0$ otherwise.
        \item[Objective] We obtain the upper bound by taking the longest $\vertex{0}{{n_X}}$-$\vertex{n_X + {n_Y}}{n_Y}$ path in this graph.
    \end{description}
\end{definition}
We consider applying \cref{claim:directed} to this graph. Note that:
\begin{itemize}
    \item The graph has $o(n ^ {0.03})$ vertices and has $\left(n / I_s + 1\right)$ layers. Edges only connect vertices from adjacent layers with weight either $0$ or $I_s$;
    \item Each edge weight is $I_s$ with probability $\exp\left(-\omega\left(\log ^ {0.02} (n)\right)\right)$;
    \item Since we sample $\eta_{l, r, 1 \ldots M}$ independently for each $(l, r]$, edges from different layers are independent in weights;
    \item Fix a vertex $\vertex{x}{y}$. If an edge from $\vertex{x'}{y'}$ to $\vertex{x' + I_s}{y''}$ is reachable from $\vertex{x}{y}$ within $\exp\left(\Theta\left(\log ^ {0.02} (n)\right)\right)$ hops (i.e., number of edges visited), then $\max(\abs{x' - x}, \abs{y' - y}) = I_s \cdot \exp\left(\Theta\left(\log ^ {0.02} (n)\right)\right)$ and $\abs{y'' - y'} \le I_s$. Since $I_s$ divides each of $x, x'$ and $(I_s / \phi_y)$ divides each of $y, y', y''$, the number of reachable edges is less than $$\underbrace{\exp\left(\Theta\left(\log ^ {0.02} (n)\right)\right)}_{x'} \cdot \underbrace{\exp\left(\Theta\left(\log ^ {0.02} (n)\right)\right) \cdot \phi_y}_{y'} \cdot \underbrace{\phi_y}_{y''} = \exp\left(\Theta\left(\log ^ {0.02} (n)\right)\right).$$
\end{itemize} 
From \cref{claim:directed} (with weights scaled up by $I_s$) with 
\begin{itemize}
    \item $h := O(n / I_s)$;
    \item $N := O(n ^ {0.03})$;
    \item $B := \exp\left(\Theta\left(\log ^ {0.02} (n)\right)\right)$;
    \item $d := \exp\left(\Theta\left(\log ^ {0.02} (n)\right)\right)$;
    \item $p := \exp\left(-\omega\left(\log ^ {0.02} (n)\right)\right)$;
    \item $u := \vertex{0}{n_X}$.
\end{itemize}
we obtain that with probability $1 - n ^ {-\omega(1)}$, an upper bound for the overestimation on each scale due to ``bad'' estimates (type II error) is 
$$n \cdot \sqrt{\log (n ^ {0.03})} / \exp\left(\Theta\left(\log ^ {0.02} (n)\right)\right) = n \cdot \exp\left(-\Omega\left(\log ^ {0.02} (n)\right)\right).$$
By union bound across all scales, we obtain that with probability $1 - o(1)$, the total type II error is
$$n \cdot S \cdot \exp\left(-\Omega\left(\log ^ {0.02} (n)\right)\right) = n \cdot \exp\left(-\Omega\left(\log ^ {0.02} (n)\right)\right).$$
Since $w(\hat{p}) = n / 2 ^ {o\left(\log ^ {0.01} (n)\right)}$, this is $o(w(\hat{p}))$.
\end{proof}

\subsubsection{Our Can be computed Efficiently}
Finally we show that our algorithm is efficient:
\lemmaefficiencylcs*
\begin{proof}
    For the running time, the total number of $(l, y_l, r, y_r)$ is 
    $$\underbrace{(n / {\mathcal{B}}_x)}_{l} \cdot \underbrace{(n / {\mathcal{B}}_x)}_{r} \cdot \underbrace{(n\phi_y / {\mathcal{B}}_x)}_{y_l} \cdot \underbrace{(n\phi_y / {\mathcal{B}}_x)}_{y_r} = n ^ 4 \cdot {\phi_y} ^ 2 / {{\mathcal{B}}_x} ^ 4 = o(n ^ {0.045}).$$
    For each $(l, y_l, r, y_r)$ with $(l, r]$ on scale $s$, the bottleneck for computing $\whest[l, y_l, r, y_r]$ is to compute the longest path in the Scale-$(s - 1)$ Sparsified Grid Graph, which takes linear time in the number of edges. This number of edges  is no more than the total number of $(l, y_l, r, y_r)$'s which is $o(n ^ {0.045})$. The overall running time is therefore at most $o(n ^ {0.09})$. 
    
    For query complexity, for each edge from $\vertex{x}{y}$ to $\vertex{x + {\mathcal{B}}_x}{y'}$, the edge is queried if every interval $(l, r] \supset (x, x + {\mathcal{B}}_x]$ on an active scale is sub-sampled, which happens with probability $1 / 2$ per active scale independently. Thus the probability that an edge ``survives'' sub-sampling is $2 ^ {-\abs{\tilde{S}}}$.
    With probability $1 - o(1)$, $\abs{\tilde{S}} = \Omega(S ^ {0.02})$, and the number of queries is at most $2 ^ {-\Omega(S ^ {0.02})} < 2 ^ {-\Omega\left(\log ^ {0.01} (n)\right)}$ times the total number of edges. The total number of edges is (notice that $\abs{y - y'} \le \mathcal{B}_x$) 
    $$\underbrace{(n / {\mathcal{B}}_x)}_{x} \cdot \underbrace{(n\phi_y / {\mathcal{B}}_x)}_{y} \cdot \underbrace{\phi_y}_{y'} = \frac{n ^ 2{\phi_y} ^ 2}{{{\mathcal{B}}_x} ^ 2}$$
    And the total number of queries is $\frac{n ^ 2{\phi_y} ^ 2}{{{\mathcal{B}}_x} ^ 2} / 2 ^ {\Omega\left(\log ^ {0.01} (n)\right)}$.
\end{proof}

\subsubsection{Proof of \cref{claim:onesidedlcs}}
Finally, we give the missing proof of \cref{claim:onesidedlcs}.
\claimonesidedlcs*
\begin{proof}
    As a warm-up, suppose that $\tilde{A}$ is exactly equal to the actual mean of $A$ which is $\overline{A} := \sum\limits_{i = 1} ^ M{a_i} / M$. Then $\ESTSUM(A_I) = 2\sum_{i \in I}{\min\left(a_i, 2\log ^ {0.02}{(n)}\overline{A}\right)}$ is an estimate for the sum $\sum\limits_{i = 1} ^ M{\min\left(a_i, 2\log ^ {0.02}{(n)}\overline{A}\right)}$. Since every term $\min\left(a_i, 2\log ^ {0.02}{(n)}\overline{A}\right)$ is clipped to be at most $2\log ^ {0.02}{(n)}\overline{A}$, we can apply Hoeffding's Inequality to obtain that $\ESTSUM(A_I)$ $\left(1 \pm O\left(1 / \log ^ {0.02} (n)\right)\right)$-approximates $\sum\limits_{i = 1} ^ M{\min\left(a_i, 2\log ^ {0.02}{(n)}\overline{A}\right)}$ with the desired probability. We can see that this satisfies both conditions in the lemma.

    Now we begin the analysis using the actual $\tilde{A}$.
    For the upper bound, notice that:
    \begin{align*}
        \tilde{A}   &= \sum_{i \in I}{a_i} / \left(\frac{M}{2}\right) \\
                    &\le \sum\limits_{i = 1} ^ M{a_i} / \left(\frac{M}{2}\right) \\
                    &= 2\sum\limits_{i = 1} ^ M{a_i} / M \\
                    &= 2\overline{A}.
    \end{align*}
    Thus $\ESTSUM(A_I) \le 2\sum_{i \in I}{\min\left(a_i, 4\log ^ {0.02}{(n)}\overline{A}\right)}$, which is an estimate for the sum $\sum\limits_{i = 1} ^ M{\min\left(a_i, 4\log ^ {0.02}{(n)}\overline{A}\right)}$. Since $\overline{A}$ is the mean of $A$, we have $\sum\limits_{i = 1} ^ M{\min\left(a_i, \omega(\overline{A})\right)} = \Omega(M\overline{A})$. Thus we can apply Hoeffding's Inequality (\cref{lemma:hoeffding}, with $c := 4\log ^ {0.02}{(n)}\overline{A}, \delta := \Theta\left(\frac{1}{\log ^ {0.02} (n)}\sum\limits_{i = 1} ^ M{\left(a_i, 4\log ^ {0.02}{(n)}\overline{A}\right)}\right)$) to obtain that
    \begin{equation} \label{eq:onesidedlcsupper}
    \ESTSUM(A_I) \le \left(1 + O\left(1 / \log ^ {0.02} (n)\right)\right)\sum\limits_{i = 1} ^ M{\min\left(a_i, 4\log ^ {0.02}{(n)}\overline{A}\right)}
    \end{equation}
    holds except with probability at most $\exp\left(-\Omega\left(M / \log ^ {0.08} (n)\right)\right) = \exp\left(-\omega\left(\log ^ {0.02}{(n)}\right)\right)$. Since we have $\min\left(a_i, 4\log ^ {0.02}{(n)}\overline{A}\right) \le a_i$, when \eqref{eq:onesidedlcsupper} holds, we indeed have
    $$\ESTSUM(A_I) \le \left(1 + O\left(1 / \log ^ {0.02} (n)\right)\right)\sum\limits_{i = 1} ^ M{a_i}.$$

    For the lower bound, suppose that every $a_i \le \log ^ {0.02}{(n)}\overline{A}$. 
    We can apply Hoeffding's Inequality (\cref{lemma:hoeffding}, with $c := \log ^ {0.02}{(n)}\overline{A}, \delta := \frac{1}{2}\sum\limits_{i = 1} ^ M{a_i}$) to obtain that
    \begin{equation} \label{eq:onesidedlcslower}
    \tilde{A} \ge \frac{1}{2}\overline{A}
    \end{equation}
    holds except with probability at most $\exp\left(-\Omega\left(M / \log ^ {0.04} (n)\right)\right) = \exp\left(-\omega\left(\log ^ {0.02}{(n)}\right)\right)$. When \eqref{eq:onesidedlcslower} holds, we have $\min\left(a_i, 2\log ^ {0.02}{(n)} \tilde{A}\right) = a_i$.
    Thus we can apply Hoeffding's Inequality (\cref{lemma:hoeffding}, with $c := \log ^ {0.02}{(n)}\overline{A}, \delta := \Theta\left(\frac{1}{\log ^ {0.02} (n)}\sum\limits_{i = 1} ^ M{a_i}\right)$) on $\sum_{i \in I}{a_i}$ to obtain that
    \begin{equation*}
    \ESTSUM(A_I) \ge \left(1 - O\left(1 / \log ^ {0.02} (n)\right)\right)\sum\limits_{i = 1} ^ M{a_i}
    \end{equation*}
    holds except with probability at most $\exp\left(-\Omega\left(M / \log ^ {0.08} (n)\right)\right) = \exp\left(-\omega\left(\log ^ {0.02}{(n)}\right)\right)$.
\end{proof}

\bibliographystyle{alphaurl} 
\bibliography{CAREER}

@TechReport{Knuth72,
  author      = {Václav Chvátal and David A. Klarner and Donald E. Knuth},
  title       = {Selected combinatorial research problems.},
  institution = {Computer Science Department, Stanford University},
  year        = {1972},
  url         = {https://pdfs.semanticscholar.org/e832/afd53b720bbb2c3dfff89f5005d5c395719f.pdf},
}

@inproceedings{AB17-deterministic,
  author    = {Amir Abboud and
               Arturs Backurs},
  title     = {Towards Hardness of Approximation for Polynomial Time Problems},
  booktitle = {8th Innovations in Theoretical Computer Science Conference, {ITCS}
               2017, January 9-11, 2017, Berkeley, CA, {USA}},
  pages     = {11:1--11:26},
  year      = {2017},
  crossrefignore  = {DBLP:conf/innovations/2017},
  url       = {https://doi.org/10.4230/LIPIcs.ITCS.2017.11},
  doi       = {10.4230/LIPIcs.ITCS.2017.11},
  bibsource = {dblp computer science bibliography, https://dblp.org}
}

@inproceedings{AHWW16-polylog_shaved,
  author    = {Amir Abboud and
               Thomas Dueholm Hansen and
               Virginia Vassilevska Williams and
               Ryan Williams},
  title     = {Simulating branching programs with edit distance and friends: or:
               a polylog shaved is a lower bound made},
  booktitle = {Proceedings of the 48th Annual {ACM} {SIGACT} Symposium on Theory
               of Computing, {STOC} 2016, Cambridge, MA, USA, June 18-21, 2016},
  pages     = {375--388},
  year      = {2016},
  crossrefignore  = {DBLP:conf/stoc/2016},
  url       = {http://doi.acm.org/10.1145/2897518.2897653},
  doi       = {10.1145/2897518.2897653},
  bibsource = {dblp computer science bibliography, https://dblp.org}
}

@inproceedings{AbboudB18,
  author    = {Amir Abboud and
               Karl Bringmann},
  title     = {Tighter Connections Between Formula-SAT and Shaving Logs},
  booktitle = {45th International Colloquium on Automata, Languages, and Programming,
               {ICALP} 2018, July 9-13, 2018, Prague, Czech Republic},
  pages     = {8:1--8:18},
  year      = {2018},
  url       = {https://doi.org/10.4230/LIPIcs.ICALP.2018.8},
  doi       = {10.4230/LIPIcs.ICALP.2018.8},
  bibsource = {dblp computer science bibliography, https://dblp.org}
}

@article{LMS98,
 author = {Landau, Gad M. and Myers, Eugene W. and Schmidt, Jeanette P.},
 title = {Incremental String Comparison},
 journal = {SIAM J. Comput.},
 issue_date = {April 1998},
 volume = {27},
 number = {2},
 year = {1998},
 issn = {0097-5397},
 pages = {557--582},
 numpages = {26},
 doi = {10.1137/S0097539794264810},
 publisher = {SIAM},
}

@inproceedings{ABW15-LCS,
  author    = {Amir Abboud and
               Arturs Backurs and
               Virginia Vassilevska Williams},
  title     = {Tight Hardness Results for {LCS} and Other Sequence Similarity Measures},
  booktitle = {{IEEE} 56th Annual Symposium on Foundations of Computer Science, {FOCS}
               2015, Berkeley, CA, USA, 17-20 October, 2015},
  pages     = {59--78},
  year      = {2015},
  crossrefignore  = {DBLP:conf/focs/2015},
  url       = {https://doi.org/10.1109/FOCS.2015.14},
  doi       = {10.1109/FOCS.2015.14},
  bibsource = {dblp computer science bibliography, https://dblp.org}
}

@inproceedings{AWW14-Local_Alignment,
  author    = {Amir Abboud and
               Virginia Vassilevska Williams and
               Oren Weimann},
  title     = {Consequences of Faster Alignment of Sequences},
  booktitle = {Automata, Languages, and Programming - 41st International Colloquium,
               {ICALP} 2014, Copenhagen, Denmark, July 8-11, 2014, Proceedings, Part
               {I}},
  pages     = {39--51},
  year      = {2014},
  crossrefignore  = {DBLP:conf/icalp/2014-1},
  url       = {https://doi.org/10.1007/978-3-662-43948-7_4},
  doi       = {10.1007/978-3-662-43948-7_4},
  bibsource = {dblp computer science bibliography, https://dblp.org}
}

@inproceedings{BK15-LCS,
  author    = {Karl Bringmann and
               Marvin K{\"{u}}nnemann},
  title     = {Quadratic Conditional Lower Bounds for String Problems and Dynamic
               Time Warping},
  booktitle = {{IEEE} 56th Annual Symposium on Foundations of Computer Science, {FOCS}
               2015, Berkeley, CA, USA, 17-20 October, 2015},
  pages     = {79--97},
  year      = {2015},
  crossrefignore  = {DBLP:conf/focs/2015},
  url       = {https://doi.org/10.1109/FOCS.2015.15},
  doi       = {10.1109/FOCS.2015.15},
  bibsource = {dblp computer science bibliography, https://dblp.org}
}

@article{BI18,
  author       = {Arturs Backurs and
                  Piotr Indyk},
  title        = {Edit Distance Cannot Be Computed in Strongly Subquadratic Time (Unless
                  {SETH} is False)},
  journal      = {{SIAM} J. Comput.},
  volume       = {47},
  number       = {3},
  pages        = {1087--1097},
  year         = {2018},
  url          = {https://doi.org/10.1137/15M1053128},
  doi          = {10.1137/15M1053128},
  bibsource    = {dblp computer science bibliography, https://dblp.org}
}

@article{AO12-edit,
  author    = {Alexandr Andoni and
               Krzysztof Onak},
  title     = {Approximating Edit Distance in Near-Linear Time},
  journal   = {{SIAM} J. Comput.},
  volume    = {41},
  number    = {6},
  pages     = {1635--1648},
  year      = {2012},
  url       = {https://doi.org/10.1137/090767182},
  doi       = {10.1137/090767182},
  bibsource = {dblp computer science bibliography, https://dblp.org}
}

@article{Ukkonen85,
  author    = {Esko Ukkonen},
  title     = {Algorithms for Approximate String Matching},
  journal   = {Information and Control},
  volume    = {64},
  number    = {1-3},
  pages     = {100--118},
  year      = {1985},
  url       = {https://doi.org/10.1016/S0019-9958(85)80046-2},
  doi       = {10.1016/S0019-9958(85)80046-2},
  bibsource = {dblp computer science bibliography, https://dblp.org}
}

@article{AK12,
  author    = {Alexandr Andoni and
               Robert Krauthgamer},
  title     = {The smoothed complexity of edit distance},
  journal   = {{ACM} Trans. Algorithms},
  volume    = {8},
  number    = {4},
  pages     = {44:1--44:25},
  year      = {2012},
  url       = {http://doi.acm.org/10.1145/2344422.2344434},
  doi       = {10.1145/2344422.2344434},
  bibsource = {dblp computer science bibliography, https://dblp.org}
}

@inproceedings{BSS20,
  author       = {Mahdi Boroujeni and
                  Masoud Seddighin and
                  Saeed Seddighin},
  editor       = {Shuchi Chawla},
  title        = {Improved Algorithms for Edit Distance and {LCS:} Beyond Worst Case},
  booktitle    = {Proceedings of the 2020 {ACM-SIAM} Symposium on Discrete Algorithms,
                  {SODA} 2020, Salt Lake City, UT, USA, January 5-8, 2020},
  pages        = {1601--1620},
  publisher    = {{SIAM}},
  year         = {2020},
  url          = {https://doi.org/10.1137/1.9781611975994.99},
  doi          = {10.1137/1.9781611975994.99},
  bibsource    = {dblp computer science bibliography, https://dblp.org}
}

@inproceedings{Kuszmaul19,
  author       = {William Kuszmaul},
  editor       = {Timothy M. Chan},
  title        = {Efficiently Approximating Edit Distance Between Pseudorandom Strings},
  booktitle    = {Proceedings of the Thirtieth Annual {ACM-SIAM} Symposium on Discrete
                  Algorithms, {SODA} 2019, San Diego, California, USA, January 6-9,
                  2019},
  pages        = {1165--1180},
  publisher    = {{SIAM}},
  year         = {2019},
  url          = {https://doi.org/10.1137/1.9781611975482.71},
  doi          = {10.1137/1.9781611975482.71},
  bibsource    = {dblp computer science bibliography, https://dblp.org}
}

@article{BGS21,
  author       = {Mahdi Boroujeni and
                  Mohammad Ghodsi and
                  Saeed Seddighin},
  title        = {Improved {MPC} Algorithms for Edit Distance and Ulam Distance},
  journal      = {{IEEE} Trans. Parallel Distributed Syst.},
  volume       = {32},
  number       = {11},
  pages        = {2764--2776},
  year         = {2021},
  url          = {https://doi.org/10.1109/TPDS.2021.3076534},
  doi          = {10.1109/TPDS.2021.3076534},
  bibsource    = {dblp computer science bibliography, https://dblp.org}
}

@inproceedings{JNW21,
  author       = {Ce Jin and
                  Jelani Nelson and
                  Kewen Wu},
  editor       = {Markus Bl{\"{a}}ser and
                  Benjamin Monmege},
  title        = {An Improved Sketching Algorithm for Edit Distance},
  booktitle    = {38th International Symposium on Theoretical Aspects of Computer Science,
                  {STACS} 2021, March 16-19, 2021, Saarbr{\"{u}}cken, Germany (Virtual
                  Conference)},
  series       = {LIPIcs},
  volume       = {187},
  pages        = {45:1--45:16},
  publisher    = {Schloss Dagstuhl - Leibniz-Zentrum f{\"{u}}r Informatik},
  year         = {2021},
  url          = {https://doi.org/10.4230/LIPIcs.STACS.2021.45},
  doi          = {10.4230/LIPICS.STACS.2021.45},
  bibsource    = {dblp computer science bibliography, https://dblp.org}
}

@article{AK10-edit-CC,
  author    = {Alexandr Andoni and
               Robert Krauthgamer},
  title     = {The Computational Hardness of Estimating Edit Distance},
  journal   = {{SIAM} J. Comput.},
  volume    = {39},
  number    = {6},
  pages     = {2398--2429},
  year      = {2010},
  url       = {https://doi.org/10.1137/080716530},
  doi       = {10.1137/080716530},
  bibsource = {dblp computer science bibliography, https://dblp.org}
}

@inproceedings{AKO10-edit,
  author    = {Alexandr Andoni and
               Robert Krauthgamer and
               Krzysztof Onak},
  title     = {Polylogarithmic Approximation for Edit Distance and the Asymmetric
               Query Complexity},
  booktitle = {51th Annual {IEEE} Symposium on Foundations of Computer Science, {FOCS}
               2010, October 23-26, 2010, Las Vegas, Nevada, {USA}},
  pages     = {377--386},
  year      = {2010},
  crossrefignore  = {DBLP:conf/focs/2010},
  url       = {https://doi.org/10.1109/FOCS.2010.43},
  doi       = {10.1109/FOCS.2010.43},
  bibsource = {dblp computer science bibliography, https://dblp.org}
}

@inproceedings{ADGIR03-edit,
  author    = {Alexandr Andoni and
               Michel Deza and
               Anupam Gupta and
               Piotr Indyk and
               Sofya Raskhodnikova},
  title     = {Lower bounds for embedding edit distance into normed spaces},
  booktitle = {Proceedings of the Fourteenth Annual {ACM-SIAM} Symposium on Discrete
               Algorithms, January 12-14, 2003, Baltimore, Maryland, {USA.}},
  pages     = {523--526},
  year      = {2003},
  crossrefignore  = {DBLP:conf/soda/2003},
  url       = {http://dl.acm.org/citation.cfm?id=644108.644196},
  bibsource = {dblp computer science bibliography, https://dblp.org}
}

@article{OR07-edit,
  author    = {Rafail Ostrovsky and
               Yuval Rabani},
  title     = {Low distortion embeddings for edit distance},
  journal   = {J. {ACM}},
  volume    = {54},
  number    = {5},
  pages     = {23},
  year      = {2007},
  url       = {http://doi.acm.org/10.1145/1284320.1284322},
  doi       = {10.1145/1284320.1284322},
  bibsource = {dblp computer science bibliography, https://dblp.org}
}

@inproceedings{BJKK04-edit,
  author    = {Ziv Bar{-}Yossef and
               T. S. Jayram and
               Robert Krauthgamer and
               Ravi Kumar},
  title     = {Approximating Edit Distance Efficiently},
  booktitle = {45th Symposium on Foundations of Computer Science {(FOCS} 2004), 17-19
               October 2004, Rome, Italy, Proceedings},
  pages     = {550--559},
  year      = {2004},
  crossrefignore  = {DBLP:conf/focs/2004},
  url       = {https://doi.org/10.1109/FOCS.2004.14},
  doi       = {10.1109/FOCS.2004.14},
  bibsource = {dblp computer science bibliography, https://dblp.org}
}

@article{KR09-edit,
  author    = {Robert Krauthgamer and
               Yuval Rabani},
  title     = {Improved Lower Bounds for Embeddings into $L_1$},
  journal   = {{SIAM} J. Comput.},
  volume    = {38},
  number    = {6},
  pages     = {2487--2498},
  year      = {2009},
  url       = {https://doi.org/10.1137/060660126},
  doi       = {10.1137/060660126},
  bibsource = {dblp computer science bibliography, https://dblp.org}
}

@inproceedings{CGKK18,
  author       = {Moses Charikar and
                  Ofir Geri and
                  Michael P. Kim and
                  William Kuszmaul},
  editor       = {Ioannis Chatzigiannakis and
                  Christos Kaklamanis and
                  D{\'{a}}niel Marx and
                  Donald Sannella},
  title        = {On Estimating Edit Distance: Alignment, Dimension Reduction, and Embeddings},
  booktitle    = {45th International Colloquium on Automata, Languages, and Programming,
                  {ICALP} 2018, July 9-13, 2018, Prague, Czech Republic},
  series       = {LIPIcs},
  volume       = {107},
  pages        = {34:1--34:14},
  publisher    = {Schloss Dagstuhl - Leibniz-Zentrum f{\"{u}}r Informatik},
  year         = {2018},
  url          = {https://doi.org/10.4230/LIPIcs.ICALP.2018.34},
  doi          = {10.4230/LIPIcs.ICALP.2018.34},
  bibsource    = {dblp computer science bibliography, https://dblp.org}
}

@inproceedings{BES06-edit,
  author    = {Tugkan Batu and
               Funda Erg{\"{u}}n and
               S{\"{u}}leyman Cenk Sahinalp},
  title     = {Oblivious string embeddings and edit distance approximations},
  booktitle = {Proceedings of the Seventeenth Annual {ACM-SIAM} Symposium on Discrete
               Algorithms, {SODA} 2006, Miami, Florida, USA, January 22-26, 2006},
  pages     = {792--801},
  year      = {2006},
  crossrefignore  = {DBLP:conf/soda/2006},
  url       = {http://dl.acm.org/citation.cfm?id=1109557.1109644},
  bibsource = {dblp computer science bibliography, https://dblp.org}
}

@inproceedings{BEKMRRS03-edit-testing,
  author    = {Tugkan Batu and
               Funda Erg{\"{u}}n and
               Joe Kilian and
               Avner Magen and
               Sofya Raskhodnikova and
               Ronitt Rubinfeld and
               Rahul Sami},
  title     = {A sublinear algorithm for weakly approximating edit distance},
  booktitle = {Proceedings of the 35th Annual {ACM} Symposium on Theory of Computing,
               June 9-11, 2003, San Diego, CA, {USA}},
  pages     = {316--324},
  year      = {2003},
  crossrefignore  = {DBLP:conf/stoc/2003},
  url       = {http://doi.acm.org/10.1145/780542.780590},
  doi       = {10.1145/780542.780590},
  bibsource = {dblp computer science bibliography, https://dblp.org}
}

@article{BEGHS21,
  author       = {Mahdi Boroujeni and
                  Soheil Ehsani and
                  Mohammad Ghodsi and
                  MohammadTaghi Hajiaghayi and
                  Saeed Seddighin},
  title        = {Approximating Edit Distance in Truly Subquadratic Time: Quantum and
                  MapReduce},
  journal      = {J. {ACM}},
  volume       = {68},
  number       = {3},
  pages        = {19:1--19:41},
  year         = {2021},
  url          = {https://doi.org/10.1145/3456807},
  doi          = {10.1145/3456807},
  bibsource    = {dblp computer science bibliography, https://dblp.org}
}

@inproceedings{HRS19,
  author       = {Bernhard Haeupler and
                  Aviad Rubinstein and
                  Amirbehshad Shahrasbi},
  editor       = {Moses Charikar and
                  Edith Cohen},
  title        = {Near-linear time insertion-deletion codes and (1+\emph{{\(\epsilon\)}})-approximating
                  edit distance via indexing},
  booktitle    = {Proceedings of the 51st Annual {ACM} {SIGACT} Symposium on Theory
                  of Computing, {STOC} 2019, Phoenix, AZ, USA, June 23-26, 2019},
  pages        = {697--708},
  publisher    = {{ACM}},
  year         = {2019},
  url          = {https://doi.org/10.1145/3313276.3316371},
  doi          = {10.1145/3313276.3316371},
  bibsource    = {dblp computer science bibliography, https://dblp.org}
}

@article{CDGKS20,
  author       = {Diptarka Chakraborty and
                  Debarati Das and
                  Elazar Goldenberg and
                  Michal Kouck{\'{y}} and
                  Michael E. Saks},
  title        = {Approximating Edit Distance Within Constant Factor in Truly Sub-quadratic
                  Time},
  journal      = {J. {ACM}},
  volume       = {67},
  number       = {6},
  pages        = {36:1--36:22},
  year         = {2020},
  url          = {https://doi.org/10.1145/3422823},
  doi          = {10.1145/3422823},
  bibsource    = {dblp computer science bibliography, https://dblp.org}
}

@inproceedings{AN10-edit-query,
  author    = {Alexandr Andoni and
               Huy L. Nguyen},
  title     = {Near-Optimal Sublinear Time Algorithms for Ulam Distance},
  booktitle = {Proceedings of the Twenty-First Annual {ACM-SIAM} Symposium on Discrete
               Algorithms, {SODA} 2010, Austin, Texas, USA, January 17-19, 2010},
  pages     = {76--86},
  year      = {2010},
  crossrefignore  = {DBLP:conf/soda/2010},
  url       = {https://doi.org/10.1137/1.9781611973075.8},
  doi       = {10.1137/1.9781611973075.8},
  bibsource = {dblp computer science bibliography, https://dblp.org}
}

@inproceedings{BR20,
  author       = {Joshua Brakensiek and
                  Aviad Rubinstein},
  editor       = {Konstantin Makarychev and
                  Yury Makarychev and
                  Madhur Tulsiani and
                  Gautam Kamath and
                  Julia Chuzhoy},
  title        = {Constant-factor approximation of near-linear edit distance in near-linear
                  time},
  booktitle    = {Proccedings of the 52nd Annual {ACM} {SIGACT} Symposium on Theory
                  of Computing, {STOC} 2020, Chicago, IL, USA, June 22-26, 2020},
  pages        = {685--698},
  publisher    = {{ACM}},
  year         = {2020},
  url          = {https://doi.org/10.1145/3357713.3384282},
  doi          = {10.1145/3357713.3384282},
  bibsource    = {dblp computer science bibliography, https://dblp.org}
}

@article{BCR20,
  author       = {Joshua Brakensiek and
                  Moses Charikar and
                  Aviad Rubinstein},
  title        = {A Simple Sublinear Algorithm for Gap Edit Distance},
  journal      = {CoRR},
  volume       = {abs/2007.14368},
  year         = {2020},
  url          = {https://arxiv.org/abs/2007.14368},
  eprinttype    = {arXiv},
  eprint       = {2007.14368},
  bibsource    = {dblp computer science bibliography, https://dblp.org}
}

@inproceedings{RS20,
  author       = {Aviad Rubinstein and
                  Zhao Song},
  editor       = {Shuchi Chawla},
  title        = {Reducing approximate Longest Common Subsequence to approximate Edit
                  Distance},
  booktitle    = {Proceedings of the 2020 {ACM-SIAM} Symposium on Discrete Algorithms,
                  {SODA} 2020, Salt Lake City, UT, USA, January 5-8, 2020},
  pages        = {1591--1600},
  publisher    = {{SIAM}},
  year         = {2020},
  url          = {https://doi.org/10.1137/1.9781611975994.98},
  doi          = {10.1137/1.9781611975994.98},
  bibsource    = {dblp computer science bibliography, https://dblp.org}
}

@inproceedings{HeL23,
  author       = {Xiaoyu He and
                  Ray Li},
  editor       = {Barna Saha and
                  Rocco A. Servedio},
  title        = {Approximating Binary Longest Common Subsequence in Almost-Linear Time},
  booktitle    = {Proceedings of the 55th Annual {ACM} Symposium on Theory of Computing,
                  {STOC} 2023, Orlando, FL, USA, June 20-23, 2023},
  pages        = {1145--1158},
  publisher    = {{ACM}},
  year         = {2023},
  url          = {https://doi.org/10.1145/3564246.3585104},
  doi          = {10.1145/3564246.3585104},
  bibsource    = {dblp computer science bibliography, https://dblp.org}
}

@misc{Rub18-blog,
author = {Aviad Rubinstein},
title = {Approximating Edit Distance},
journal = {Theory Dish blog},
type = {Blog},
date = {July 20},
year = {2018},
howpublished = {\url{https://theorydish.blog/2018/07/20/approximating-edit-distance/}}
}

@inproceedings{AR18,
  author    = {Amir Abboud and
               Aviad Rubinstein},
  title     = {Fast and Deterministic Constant Factor Approximation Algorithms for
               {LCS} Imply New Circuit Lower Bounds},
  booktitle = {9th Innovations in Theoretical Computer Science Conference, {ITCS}
               2018, January 11-14, 2018, Cambridge, MA, {USA}},
  pages     = {35:1--35:14},
  year      = {2018},
  crossrefignore  = {DBLP:conf/innovations/2018},
  url       = {https://doi.org/10.4230/LIPIcs.ITCS.2018.35},
  doi       = {10.4230/LIPIcs.ITCS.2018.35},
  bibsource = {dblp computer science bibliography, https://dblp.org}
}

@InProceedings{CGLRR19,
  author         = {Lijie Chen and Shafi Goldwasser and Kaifeng Lyu and Guy N. Rothblum and Aviad Rubinstein},
  title          = {Fine-grained Complexity Meets {IP} = {PSPACE}},
  booktitle      = {Proceedings of the Thirtieth Annual {ACM-SIAM} Symposium on Discrete Algorithms, {SODA} 2019, San Diego, California, USA, January 6-9, 2019},
  year           = {2019},
  pages          = {1--20},
  bibsource      = {dblp computer science bibliography, https://dblp.org},
  crossrefignore = {DBLP:conf/soda/2019},
  doi            = {10.1137/1.9781611975482.1},
  url            = {https://doi.org/10.1137/1.9781611975482.1},
}

@article{Williams18,
  author    = {R. Ryan Williams},
  title     = {Faster All-Pairs Shortest Paths via Circuit Complexity},
  journal   = {{SIAM} J. Comput.},
  volume    = {47},
  number    = {5},
  pages     = {1965--1985},
  year      = {2018},
  url       = {https://doi.org/10.1137/15M1024524},
  doi       = {10.1137/15M1024524},
  bibsource = {dblp computer science bibliography, https://dblp.org}
}

@inproceedings{HSSS19,
  author    = {MohammadTaghi Hajiaghayi and
               Masoud Seddighin and
               Saeed Seddighin and
               Xiaorui Sun},
  title     = {Approximating {LCS} in Linear Time: Beating the {\(\surd\)}n Barrier},
  booktitle = {Proceedings of the Thirtieth Annual {ACM-SIAM} Symposium on Discrete
               Algorithms, {SODA} 2019, San Diego, California, USA, January 6-9,
               2019},
  pages     = {1181--1200},
  year      = {2019},
  crossrefignore  = {DBLP:conf/soda/2019},
  url       = {https://doi.org/10.1137/1.9781611975482.72},
  doi       = {10.1137/1.9781611975482.72},
  bibsource = {dblp computer science bibliography, https://dblp.org}
}

@inproceedings{GKS19,
  author       = {Elazar Goldenberg and
                  Robert Krauthgamer and
                  Barna Saha},
  editor       = {David Zuckerman},
  title        = {Sublinear Algorithms for Gap Edit Distance},
  booktitle    = {60th {IEEE} Annual Symposium on Foundations of Computer Science, {FOCS}
                  2019, Baltimore, Maryland, USA, November 9-12, 2019},
  pages        = {1101--1120},
  publisher    = {{IEEE} Computer Society},
  year         = {2019},
  url          = {https://doi.org/10.1109/FOCS.2019.00070},
  doi          = {10.1109/FOCS.2019.00070},
  bibsource    = {dblp computer science bibliography, https://dblp.org}
}

@inproceedings{AkmalW21,
  author       = {Shyan Akmal and
                  Virginia {Vassilevska Williams}},
  editor       = {Nikhil Bansal and
                  Emanuela Merelli and
                  James Worrell},
  title        = {Improved Approximation for Longest Common Subsequence over Small Alphabets},
  booktitle    = {48th International Colloquium on Automata, Languages, and Programming,
                  {ICALP} 2021, July 12-16, 2021, Glasgow, Scotland (Virtual Conference)},
  series       = {LIPIcs},
  volume       = {198},
  pages        = {13:1--13:18},
  publisher    = {Schloss Dagstuhl - Leibniz-Zentrum f{\"{u}}r Informatik},
  year         = {2021},
  url          = {https://doi.org/10.4230/LIPIcs.ICALP.2021.13},
  doi          = {10.4230/LIPICS.ICALP.2021.13},
  bibsource    = {dblp computer science bibliography, https://dblp.org}
}

@article{RSSS23,
author = {Rubinstein, Aviad and Seddighin, Saeed and Song, Zhao and Sun, Xiaorui},
title = {Approximation Algorithms for LCS and LIS with Truly Improved Running Times},
journal = {SIAM Journal on Computing},
volume = {54},
number = {4},
pages = {FOCS19-276-FOCS19-331},
year = {2025},
doi = {10.1137/20M1316500},

URL = { 
    
        https://doi.org/10.1137/20M1316500
    
    

},
eprint = { 
    
        https://doi.org/10.1137/20M1316500
    
    

}
}

@inproceedings{HSS19,
  author    = {MohammadTaghi Hajiaghayi and
               Saeed Seddighin and
               Xiaorui Sun},
  title     = {Massively Parallel Approximation Algorithms for Edit Distance and
               Longest Common Subsequence},
  booktitle = {Proceedings of the Thirtieth Annual {ACM-SIAM} Symposium on Discrete
               Algorithms, {SODA} 2019, San Diego, California, USA, January 6-9,
               2019},
  pages     = {1654--1672},
  year      = {2019},
  crossrefignore  = {DBLP:conf/soda/2019},
  url       = {https://doi.org/10.1137/1.9781611975482.100},
  doi       = {10.1137/1.9781611975482.100},
  bibsource = {dblp computer science bibliography, https://dblp.org}
}

@article{MP80,
  author       = {William J. Masek and
                  Mike Paterson},
  title        = {A Faster Algorithm Computing String Edit Distances},
  journal      = {J. Comput. Syst. Sci.},
  volume       = {20},
  number       = {1},
  pages        = {18--31},
  year         = {1980},
  url          = {https://doi.org/10.1016/0022-0000(80)90002-1},
  doi          = {10.1016/0022-0000(80)90002-1},
  bibsource    = {dblp computer science bibliography, https://dblp.org}
}

@inproceedings{SS22,
  author       = {Masoud Seddighin and
                  Saeed Seddighin},
  editor       = {Mark Braverman},
  title        = {3+{\(\epsilon\)} Approximation of Tree Edit Distance in Truly Subquadratic
                  Time},
  booktitle    = {13th Innovations in Theoretical Computer Science Conference, {ITCS}
                  2022, January 31 - February 3, 2022, Berkeley, CA, {USA}},
  series       = {LIPIcs},
  volume       = {215},
  pages        = {115:1--115:22},
  publisher    = {Schloss Dagstuhl - Leibniz-Zentrum f{\"{u}}r Informatik},
  year         = {2022},
  url          = {https://doi.org/10.4230/LIPIcs.ITCS.2022.115},
  doi          = {10.4230/LIPIcs.ITCS.2022.115},
  bibsource    = {dblp computer science bibliography, https://dblp.org}
}

@inproceedings{CFHJLRSZ21,
  author       = {Kuan Cheng and
                  Alireza Farhadi and
                  MohammadTaghi Hajiaghayi and
                  Zhengzhong Jin and
                  Xin Li and
                  Aviad Rubinstein and
                  Saeed Seddighin and
                  Yu Zheng},
  editor       = {Nikhil Bansal and
                  Emanuela Merelli and
                  James Worrell},
  title        = {Streaming and Small Space Approximation Algorithms for Edit Distance
                  and Longest Common Subsequence},
  booktitle    = {48th International Colloquium on Automata, Languages, and Programming,
                  {ICALP} 2021, July 12-16, 2021, Glasgow, Scotland (Virtual Conference)},
  series       = {LIPIcs},
  volume       = {198},
  pages        = {54:1--54:20},
  publisher    = {Schloss Dagstuhl - Leibniz-Zentrum f{\"{u}}r Informatik},
  year         = {2021},
  url          = {https://doi.org/10.4230/LIPIcs.ICALP.2021.54},
  doi          = {10.4230/LIPIcs.ICALP.2021.54},
  bibsource    = {dblp computer science bibliography, https://dblp.org}
}

@inproceedings{GRS20,
  author       = {Elazar Goldenberg and
                  Aviad Rubinstein and
                  Barna Saha},
  editor       = {Konstantin Makarychev and
                  Yury Makarychev and
                  Madhur Tulsiani and
                  Gautam Kamath and
                  Julia Chuzhoy},
  title        = {Does preprocessing help in fast sequence comparisons?},
  booktitle    = {Proccedings of the 52nd Annual {ACM} {SIGACT} Symposium on Theory
                  of Computing, {STOC} 2020, Chicago, IL, USA, June 22-26, 2020},
  pages        = {657--670},
  publisher    = {{ACM}},
  year         = {2020},
  url          = {https://doi.org/10.1145/3357713.3384300},
  doi          = {10.1145/3357713.3384300},
  bibsource    = {dblp computer science bibliography, https://dblp.org}
}

@article{BCD23,
  author       = {Karl Bringmann and
                  Vincent Cohen{-}Addad and
                  Debarati Das},
  title        = {A Linear-Time \emph{n}\({}^{\mbox{0.4}}\)-Approximation for Longest
                  Common Subsequence},
  journal      = {{ACM} Trans. Algorithms},
  volume       = {19},
  number       = {1},
  pages        = {9:1--9:24},
  year         = {2023},
  url          = {https://doi.org/10.1145/3568398},
  doi          = {10.1145/3568398},
  bibsource    = {dblp computer science bibliography, https://dblp.org}
}

@inproceedings{BCFN22a,
  author       = {Karl Bringmann and
                  Alejandro Cassis and
                  Nick Fischer and
                  Vasileios Nakos},
  editor       = {Mikolaj Bojanczyk and
                  Emanuela Merelli and
                  David P. Woodruff},
  title        = {Improved Sublinear-Time Edit Distance for Preprocessed Strings},
  booktitle    = {49th International Colloquium on Automata, Languages, and Programming,
                  {ICALP} 2022, July 4-8, 2022, Paris, France},
  series       = {LIPIcs},
  volume       = {229},
  pages        = {32:1--32:20},
  publisher    = {Schloss Dagstuhl - Leibniz-Zentrum f{\"{u}}r Informatik},
  year         = {2022},
  url          = {https://doi.org/10.4230/LIPIcs.ICALP.2022.32},
  doi          = {10.4230/LIPIcs.ICALP.2022.32},
  bibsource    = {dblp computer science bibliography, https://dblp.org}
}

@inproceedings{BCFN22b,
  author       = {Karl Bringmann and
                  Alejandro Cassis and
                  Nick Fischer and
                  Vasileios Nakos},
  editor       = {Stefano Leonardi and
                  Anupam Gupta},
  title        = {Almost-optimal sublinear-time edit distance in the low distance regime},
  booktitle    = {{STOC} '22: 54th Annual {ACM} {SIGACT} Symposium on Theory of Computing,
                  Rome, Italy, June 20 - 24, 2022},
  pages        = {1102--1115},
  publisher    = {{ACM}},
  year         = {2022},
  url          = {https://doi.org/10.1145/3519935.3519990},
  doi          = {10.1145/3519935.3519990},
  bibsource    = {dblp computer science bibliography, https://dblp.org}
}

@inproceedings{AN20,
  author       = {Alexandr Andoni and
                  Negev Shekel Nosatzki},
  editor       = {Sandy Irani},
  title        = {Edit Distance in Near-Linear Time: it's a Constant Factor},
  booktitle    = {61st {IEEE} Annual Symposium on Foundations of Computer Science, {FOCS}
                  2020, Durham, NC, USA, November 16-19, 2020},
  pages        = {990--1001},
  publisher    = {{IEEE}},
  year         = {2020},
  url          = {https://doi.org/10.1109/FOCS46700.2020.00096},
  doi          = {10.1109/FOCS46700.2020.00096},
  bibsource    = {dblp computer science bibliography, https://dblp.org}
}

@article{Nosatzki21,
  author       = {Negev Shekel Nosatzki},
  title        = {Approximating the Longest Common Subsequence problem within a sub-polynomial
                  factor in linear time},
  journal      = {CoRR},
  volume       = {abs/2112.08454},
  year         = {2021},
  url          = {https://arxiv.org/abs/2112.08454},
  eprinttype    = {arXiv},
  eprint       = {2112.08454},
  bibsource    = {dblp computer science bibliography, https://dblp.org}
}

@inproceedings{KS20,
  author       = {Michal Kouck{\'{y}} and
                  Michael E. Saks},
  editor       = {Konstantin Makarychev and
                  Yury Makarychev and
                  Madhur Tulsiani and
                  Gautam Kamath and
                  Julia Chuzhoy},
  title        = {Constant factor approximations to edit distance on far input pairs
                  in nearly linear time},
  booktitle    = {Proccedings of the 52nd Annual {ACM} {SIGACT} Symposium on Theory
                  of Computing, {STOC} 2020, Chicago, IL, USA, June 22-26, 2020},
  pages        = {699--712},
  publisher    = {{ACM}},
  year         = {2020},
  url          = {https://doi.org/10.1145/3357713.3384307},
  doi          = {10.1145/3357713.3384307},
  bibsource    = {dblp computer science bibliography, https://dblp.org}
}

@inproceedings{KS23,
  author       = {Michal Kouck{\'{y}} and
                  Michael E. Saks},
  editor       = {Nikhil Bansal and
                  Viswanath Nagarajan},
  title        = {Simple, deterministic, fast (but weak) approximations to edit distance
                  and Dyck edit distance},
  booktitle    = {Proceedings of the 2023 {ACM-SIAM} Symposium on Discrete Algorithms,
                  {SODA} 2023, Florence, Italy, January 22-25, 2023},
  pages        = {5203--5219},
  publisher    = {{SIAM}},
  year         = {2023},
  url          = {https://doi.org/10.1137/1.9781611977554.ch188},
  doi          = {10.1137/1.9781611977554.ch188},
  bibsource    = {dblp computer science bibliography, https://dblp.org}
}

@inproceedings{CGK16,
  author       = {Diptarka Chakraborty and
                  Elazar Goldenberg and
                  Michal Kouck{\'{y}}},
  editor       = {Daniel Wichs and
                  Yishay Mansour},
  title        = {Streaming algorithms for embedding and computing edit distance in
                  the low distance regime},
  booktitle    = {Proceedings of the 48th Annual {ACM} {SIGACT} Symposium on Theory
                  of Computing, {STOC} 2016, Cambridge, MA, USA, June 18-21, 2016},
  pages        = {712--725},
  publisher    = {{ACM}},
  year         = {2016},
  url          = {https://doi.org/10.1145/2897518.2897577},
  doi          = {10.1145/2897518.2897577},
  bibsource    = {dblp computer science bibliography, https://dblp.org}
}

@inproceedings{GKKS22,
  author       = {Elazar Goldenberg and
                  Tomasz Kociumaka and
                  Robert Krauthgamer and
                  Barna Saha},
  title        = {Gap Edit Distance via Non-Adaptive Queries: Simple and Optimal},
  booktitle    = {63rd {IEEE} Annual Symposium on Foundations of Computer Science, {FOCS}
                  2022, Denver, CO, USA, October 31 - November 3, 2022},
  pages        = {674--685},
  publisher    = {{IEEE}},
  year         = {2022},
  url          = {https://doi.org/10.1109/FOCS54457.2022.00070},
  doi          = {10.1109/FOCS54457.2022.00070},
  bibsource    = {dblp computer science bibliography, https://dblp.org}
}

@inproceedings{KociumakaS20,
  author       = {Tomasz Kociumaka and
                  Barna Saha},
  editor       = {Sandy Irani},
  title        = {Sublinear-Time Algorithms for Computing {\&} Embedding Gap Edit
                  Distance},
  booktitle    = {61st {IEEE} Annual Symposium on Foundations of Computer Science, {FOCS}
                  2020, Durham, NC, USA, November 16-19, 2020},
  pages        = {1168--1179},
  publisher    = {{IEEE}},
  year         = {2020},
  url          = {https://doi.org/10.1109/FOCS46700.2020.00112},
  doi          = {10.1109/FOCS46700.2020.00112},
  bibsource    = {dblp computer science bibliography, https://dblp.org}
}

@misc{Andoni20,
author = {Alexandr Andoni},
title = {Simple Constant-Factor Approximation to Edit Distance},
year = {2020},
howpublished = {\url{https://www.cs.columbia.edu/~andoni/papers/edit/simple.pdf}}
}

@article{hoeffding,
 ISSN = {01621459},
 URL = {http://www.jstor.org/stable/2282952},
 author = {Wassily Hoeffding},
 journal = {Journal of the American Statistical Association},
 number = {301},
 pages = {13--30},
 publisher = {[American Statistical Association, Taylor & Francis, Ltd.]},
 title = {Probability Inequalities for Sums of Bounded Random Variables},
 urldate = {2023-09-30},
 volume = {58},
 year = {1963}
}

@article{10.1145/322033.322044, author = {Hirschberg, Daniel S.}, title = {Algorithms for the Longest Common Subsequence Problem}, year = {1977}, issue_date = {Oct. 1977}, publisher = {Association for Computing Machinery}, address = {New York, NY, USA}, volume = {24}, number = {4}, issn = {0004-5411}, url = {https://doi.org/10.1145/322033.322044}, doi = {10.1145/322033.322044}, journal = {J. ACM}, month = {oct}, pages = {664–675}, numpages = {12} }

@INPROCEEDINGS{AKW23,
  author={Cassis, Alejandro and Kociumaka, Tomasz and Wellnitz, Philip},
  booktitle={2023 IEEE 64th Annual Symposium on Foundations of Computer Science (FOCS)}, 
  title={Optimal Algorithms for Bounded Weighted Edit Distance}, 
  year={2023},
  volume={},
  number={},
  pages={2177-2187},
  doi={10.1109/FOCS57990.2023.00135}}

@inproceedings{AFK+24,
author = {Abboud, Amir and Fischer, Nick and Kelley, Zander and Lovett, Shachar and Meka, Raghu},
title = {New Graph Decompositions and Combinatorial Boolean Matrix Multiplication Algorithms},
year = {2024},
isbn = {9798400703836},
publisher = {Association for Computing Machinery},
address = {New York, NY, USA},
url = {https://doi.org/10.1145/3618260.3649696},
doi = {10.1145/3618260.3649696},
booktitle = {Proceedings of the 56th Annual ACM Symposium on Theory of Computing},
pages = {935–943},
numpages = {9},
location = {Vancouver, BC, Canada},
series = {STOC 2024}
}

@inbook{LW17,
author = {Kasper Green Larsen and Ryan Williams},
title = {Faster Online Matrix-Vector Multiplication},
booktitle = {Proceedings of the 2017 Annual ACM-SIAM Symposium on Discrete Algorithms (SODA)},
year={2017},
chapter = {},
pages = {2182-2189},
doi = {10.1137/1.9781611974782.142},
URL = {https://epubs.siam.org/doi/abs/10.1137/1.9781611974782.142},
eprint = {https://epubs.siam.org/doi/pdf/10.1137/1.9781611974782.142}
}

@inproceedings{gorbachev2024bounded,
  author    = {Gorbachev, Egor and Kociumaka, Tomasz},
  title     = {Bounded Edit Distance: Optimal Static and Dynamic Algorithms for Small Integer Weights},
  booktitle = {Proceedings of the 57th Annual {ACM} {SIGACT} Symposium on Theory of Computing, {STOC} 2025, San Francisco, CA, USA, June 22-25, 2025},
  year      = {2025}
}

@inproceedings{cheng2025constant,
  author    = {Cheng, Siu-Wing and Huang, Haoqiang and Zhang, Shuo},
  title     = {Constant Approximation of Fr{\'{e}}chet Distance in Strongly Subquadratic Time},
  booktitle = {Proceedings of the 57th Annual {ACM} {SIGACT} Symposium on Theory of Computing, {STOC} 2025, Prague, Czech Republic, June 23-27, 2025},
  year      = {2025}
}

@inproceedings{Bringmann14,
author = {Bringmann, Karl},
title = {Why Walking the Dog Takes Time: Frechet Distance Has No Strongly Subquadratic Algorithms Unless SETH Fails},
year = {2014},
isbn = {9781479965175},
publisher = {IEEE Computer Society},
address = {USA},
url = {https://doi.org/10.1109/FOCS.2014.76},
doi = {10.1109/FOCS.2014.76},
booktitle = {Proceedings of the 2014 IEEE 55th Annual Symposium on Foundations of Computer Science},
pages = {661–670},
numpages = {10},
series = {FOCS '14}
}

@inproceedings{NewSSSP,
  author = {Duan, Ran and Mao, Jiayi and Mao, Xiao and Shu, Xinkai and Yin, Longhui},
title = {Breaking the Sorting Barrier for Directed Single-Source Shortest Paths},
year = {2025},
isbn = {9798400715105},
publisher = {Association for Computing Machinery},
address = {New York, NY, USA},
url = {https://doi.org/10.1145/3717823.3718179},
doi = {10.1145/3717823.3718179},
booktitle = {Proceedings of the 57th Annual ACM Symposium on Theory of Computing},
pages = {36–44},
numpages = {9},
location = {Prague, Czechia},
series = {STOC '25}
}

@inproceedings{JX24,
author = {Jin, Ce and Xu, Yinzhan},
title = {Shaving Logs via Large Sieve Inequality: Faster Algorithms for Sparse Convolution and More},
year = {2024},
isbn = {9798400703836},
publisher = {Association for Computing Machinery},
address = {New York, NY, USA},
url = {https://doi.org/10.1145/3618260.3649605},
doi = {10.1145/3618260.3649605},
booktitle = {Proceedings of the 56th Annual ACM Symposium on Theory of Computing},
pages = {1573–1584},
numpages = {12},
location = {Vancouver, BC, Canada},
series = {STOC 2024}
}

@InProceedings{CGMW21,
  author =	{Charalampopoulos, Panagiotis and Gawrychowski, Pawe{\l} and Mozes, Shay and Weimann, Oren},
  title =	{{An Almost Optimal Edit Distance Oracle}},
  booktitle =	{48th International Colloquium on Automata, Languages, and Programming (ICALP 2021)},
  pages =	{48:1--48:20},
  series =	{Leibniz International Proceedings in Informatics (LIPIcs)},
  ISBN =	{978-3-95977-195-5},
  ISSN =	{1868-8969},
  year =	{2021},
  volume =	{198},
  editor =	{Bansal, Nikhil and Merelli, Emanuela and Worrell, James},
  publisher =	{Schloss Dagstuhl -- Leibniz-Zentrum f{\"u}r Informatik},
  address =	{Dagstuhl, Germany},
  URL =		{https://drops.dagstuhl.de/entities/document/10.4230/LIPIcs.ICALP.2021.48},
  URN =		{urn:nbn:de:0030-drops-141175},
  doi =		{10.4230/LIPIcs.ICALP.2021.48}
}

@article{gawrychowski2019minimum,
  title={Minimum Cut in $ {O}(m \log ^ 2 n) $ Time},
  author={Gawrychowski, Pawe{\l} and Mozes, Shay and Weimann, Oren},
  journal={arXiv preprint arXiv:1911.01145},
  year={2019},
  note={Best ICALP Paper}
}

@inproceedings{KS24,
  author       = {Michal Kouck{\'{y}} and
                  Michael E. Saks},
  editor       = {Bojan Mohar and
                  Igor Shinkar and
                  Ryan O'Donnell},
  title        = {Almost Linear Size Edit Distance Sketch},
  booktitle    = {Proceedings of the 56th Annual {ACM} Symposium on Theory of Computing,
                  {STOC} 2024, Vancouver, BC, Canada, June 24-28, 2024},
  pages        = {956--967},
  publisher    = {{ACM}},
  year         = {2024},
  url          = {https://doi.org/10.1145/3618260.3649783},
  doi          = {10.1145/3618260.3649783},
  bibsource    = {dblp computer science bibliography, https://dblp.org}
}

@inproceedings{BCFK24,
  author       = {Karl Bringmann and
                  Alejandro Cassis and
                  Nick Fischer and
                  Tomasz Kociumaka},
  editor       = {David P. Woodruff},
  title        = {Faster Sublinear-Time Edit Distance},
  booktitle    = {Proceedings of the 2024 {ACM-SIAM} Symposium on Discrete Algorithms,
                  {SODA} 2024, Alexandria, VA, USA, January 7-10, 2024},
  pages        = {3274--3301},
  publisher    = {{SIAM}},
  year         = {2024},
  url          = {https://doi.org/10.1137/1.9781611977912.117},
  doi          = {10.1137/1.9781611977912.117},
  bibsource    = {dblp computer science bibliography, https://dblp.org}
}

@inproceedings{BCFN22,
  author       = {Karl Bringmann and
                  Alejandro Cassis and
                  Nick Fischer and
                  Vasileios Nakos},
  editor       = {Stefano Leonardi and
                  Anupam Gupta},
  title        = {Almost-optimal sublinear-time edit distance in the low distance regime},
  booktitle    = {{STOC} '22: 54th Annual {ACM} {SIGACT} Symposium on Theory of Computing,
                  Rome, Italy, June 20 - 24, 2022},
  pages        = {1102--1115},
  publisher    = {{ACM}},
  year         = {2022},
  url          = {https://doi.org/10.1145/3519935.3519990},
  doi          = {10.1145/3519935.3519990},
  bibsource    = {dblp computer science bibliography, https://dblp.org}
}

@inproceedings{KPS21,
  author       = {Tomasz Kociumaka and
                  Ely Porat and
                  Tatiana Starikovskaya},
  title        = {Small-space and streaming pattern matching with $k$
                  edits},
  booktitle    = {62nd {IEEE} Annual Symposium on Foundations of Computer Science, {FOCS}
                  2021, Denver, CO, USA, February 7-10, 2022},
  pages        = {885--896},
  publisher    = {{IEEE}},
  year         = {2021},
  url          = {https://doi.org/10.1109/FOCS52979.2021.00090},
  doi          = {10.1109/FOCS52979.2021.00090},
  bibsource    = {dblp computer science bibliography, https://dblp.org}
}

@article{Grabowski16,
title = {New tabulation and sparse dynamic programming based techniques for sequence similarity problems},
journal = {Discrete Applied Mathematics},
volume = {212},
pages = {96-103},
year = {2016},
note = {Stringology Algorithms},
issn = {0166-218X},
doi = {https://doi.org/10.1016/j.dam.2015.10.040},
url = {https://www.sciencedirect.com/science/article/pii/S0166218X15005284},
author = {Szymon Grabowski}
}

@inproceedings{BoroujeniGHS19,
  author       = {Mahdi Boroujeni and
                  Mohammad Ghodsi and
                  MohammadTaghi Hajiaghayi and
                  Saeed Seddighin},
  editor       = {Moses Charikar and
                  Edith Cohen},
  title        = {1+\emph{{\(\epsilon\)}} approximation of tree edit distance in quadratic
                  time},
  booktitle    = {Proceedings of the 51st Annual {ACM} {SIGACT} Symposium on Theory
                  of Computing, {STOC} 2019, Phoenix, AZ, USA, June 23-26, 2019},
  pages        = {709--720},
  publisher    = {{ACM}},
  year         = {2019},
  url          = {https://doi.org/10.1145/3313276.3316388},
  doi          = {10.1145/3313276.3316388},
  bibsource    = {dblp computer science bibliography, https://dblp.org}
}

@article{APred13,
author = {Drucker, Andrew},
title = {High-confidence predictions under adversarial uncertainty},
year = {2013},
issue_date = {August 2013},
publisher = {Association for Computing Machinery},
address = {New York, NY, USA},
volume = {5},
number = {3},
issn = {1942-3454},
url = {https://doi.org/10.1145/2493252.2493257},
doi = {10.1145/2493252.2493257},
journal = {ACM Trans. Comput. Theory},
month = aug,
articleno = {12},
numpages = {18}
}

@misc{peng2025,
      title={High dimensional online calibration in polynomial time}, 
      author={Binghui Peng},
      year={2025},
      eprint={2504.09096},
      archivePrefix={arXiv},
      primaryClass={cs.LG},
      url={https://arxiv.org/abs/2504.09096}, 
}

\appendix
\section{An Estimate-to-Search Reduction}
In this appendix we give a reduction from the problem of searching an approximately optimal path on an ED/LCS grid with diagonal shortcuts (equivalently, ED transformation of insertions, deletions, and substitutions, or the actual longest common subsequence) to the problem of estimating the length of the optimal path (respectively, estimating the ED or the length of the LCS). 

An $(\eps, \delta)$-approximation of ED/LCS is an estimate within $[A, (1 + \eps)A + \delta]$ where $A$ is the true ED/LCS.
\begin{lemma}[Estimate-to-Search]
Given an algorithm $\mA$ that runs in time $n^2 / 2^{\log^{\Omega(1)}(n)}$ and estimates the ED/LCS between two input strings $X,Y$ of lengths $n_X,n_Y = O(n)$ to within $(\varepsilon(n), 2^{\log^{\Omega(1)}(n)})$-approximation (i.e., the multiplicative approximation factor is a function on $n$), we obtain an algorithm $\mA'$ that, runs in time $n^2 / 2^{\log^{\Omega(1)}(n)}$ and finds a transformation whose answer is an $(\varepsilon(\Theta\left(\sqrt{n}\right)), 2^{\log^{\Omega(1)}(n)})$-approximation of the optimal answer.
\end{lemma}
\begin{proof}
Let $\gamma = 2^{\log^{c}(n)}$ for some constant $c > 0$ be such that the running time of $\mA$ is $T(\mA, n) = O\left(n^2 / 2^{\log^{2c}(n)}\right)$. Let $\mu = 2 ^ {\floor{\log(n) / 2}} = \Theta\left(\sqrt{n}\right)$. We say a path $p$ on the Rotated Edit Distance/LCS Grid Graph is $\gamma$-rounded if for every $\mu \mid x$, we have $\floor{\mu / \gamma} \mid p(x)$.

We can round the optimal path $p$ to a $\gamma$-rounded path using \cref{lemma:pathdiff}, incurring an additive error of only $O\left(\sqrt{n} \cdot (\sqrt{n} / \gamma)\right) = O(n / \gamma) = n / 2^{\log^{\Omega(1)}(n)}$. Thus it suffices to find the optimal $\gamma$-rounded path.

Using $\mA$, we can approximate the length of the optimal $\vertex{x}{y}$-$\vertex{x + \mu}{y'}$ path for every $\mu \mid x$, $\floor{\mu / \gamma} \mid y, y'$ with $\abs{y - y'} \le \mu$. Since for each path we only need a $\Theta(\mu) \cdot \Theta(\mu)$ sub-grid, the total running time is $$O\left(\underbrace{\sqrt{n} \cdot 2^{\log^{c}(n)}\sqrt{n}}_{\text{total number of sub-grids}} \cdot \underbrace{\left(\sqrt{n}\right) ^ 2 / 2^{\log^{2c}\left(\sqrt{n}\right)}}_{\text{time for each sub-grid}}\right) = n ^ 2 / 2^{\log^{\Omega(1)}(n)}.$$

We can use standard dynamic programming on a $\gamma$-rounded path $\hat{p}$ that is optimal with respect to the estimates computed by $\mA$. This gives us an $(\varepsilon(\Theta\left(\sqrt{n}\right)), 2^{\log^{\Omega(1)}(n)})$-approximation of the optimal answer. This gives us $\hat{p}(x)$ for every column $\mu \mid x$. 

Finally we can fully recover $\hat{p}$ by using an exact algorithm for each of the sub-grids that contain a sub-path of $\hat{p}$. The exact algorithm runs in time quadratic in the side-length of the underlying sub-grids. 
Hence the total time of this part is 
$$O\left(\underbrace{\sqrt{n}}_{\text{number of sub-paths}} \cdot \underbrace{\sqrt{n} \cdot \sqrt{n}}_{\text{time for each sub-path}}\right) = O(n ^ {1.5}).$$
\end{proof}

\section{An Information-Theory-Based Tree Total Deviation Upper Bound} \label{app:infotheory}

In this section, we present an alternative proof for a stronger version of \cref{lemma:totalmeandeviation} using information-theoretic tools. Specifically, we relate the notion of tree total deviation to the relative entropy (Kullback-Leibler divergence) between the normalized sequence $A$ and the uniform distribution.

\begin{lemma} \label{lemma:strong_totalmeandeviation}
    Let $A = (a_1, \dots, a_n)$ be a sequence of real numbers with $n = M^S$. Let $\|A\|_1 := \sum_x |a_x|$ and $\|A\|_\infty := \max_x |a_x|$. Then,
    $$ \treemd_{M, S}(A) \le \|A\|_1 \sqrt{2S \ln \left( \frac{n \|A\|_\infty}{\|A\|_1} \right)}. $$
\end{lemma}

\begin{proof}
    Without loss of generality, assume $a_x \ge 0$ for all $x$ (otherwise, replace $a_x$ with $|a_x|$). We define a probability distribution $P$ over the indices $x \in [n]$ where $p_x = a_x / {\|A\|_1}$. Let $U$ denote the uniform distribution over $[n]$, i.e., $u_x = 1/n$.
    
    Recall that the Kullback-Leibler (KL) divergence between $P$ and $U$ is defined as:
    $$ D(P \| U) = \sum_{x=1}^n p_x \ln\left(\frac{p_x}{1/n}\right). $$
    
    We view the recursive partitioning of $A$ as a tree $\mathcal{T}$ of depth $S$. For any recursion interval $I \in \mathcal{T}$ (a node in the tree), let ${\sigma}(I) := \sum_{x \in I} p_x$ be the total probability mass in that interval. For a scale-$s$ interval $I$ (where $s < S$), let its $M$ sub-intervals on scale $s+1$ be $I_1, \dots, I_M$. We define the \emph{local distribution} at $I$ as the vector $q_I \in \mathbb{R}^M$ where the $k$-th component is $q_I(k) = {\sigma}(I_k) / {\sigma}(I)$. Note that if the mass were distributed uniformly, this local distribution would be $u_{\text{local}} = (1/M, \dots, 1/M)$.
    
    By the chain rule for KL divergence, the total divergence $D(P \| U)$ decomposes into the sum of local divergences weighted by the mass of the parent intervals:
    \begin{equation} \label{eq:chain_rule}
        D(P \| U) = \sum_{s \in [S]} \sum_{I \in s} {\sigma}(I) \cdot D(q_I \| u_{\text{local}}).
    \end{equation}
    
    Now, consider the mean deviation term $\md_M(A_I)$ for a specific interval $I$. Substituting $A_{\Sigma I_k} := {\|A\|_1} \cdot {\sigma}(I_k)$, we have:
    \begin{align*}
        \md_M(A_I) &= \sum_{k=1}^M \left| A_{\Sigma I_k} - \frac{1}{M} A_{\Sigma I} \right| \\
        &= {\|A\|_1} \sum_{k=1}^M \left| {\sigma}(I_k) - \frac{{\sigma}(I)}{M} \right| \\
        &= {\|A\|_1} \cdot {\sigma}(I) \sum_{k=1}^M \left| \frac{{\sigma}(I_k)}{{\sigma}(I)} - \frac{1}{M} \right| \\
        &= {\|A\|_1} \cdot {\sigma}(I) \sum_{k=1}^M \left|q_I(k) - \frac{1}{M} \right|.
    \end{align*}
    The summation term $\sum |q_I(k) - 1/M|$ is exactly twice the Total Variation distance between the local distribution $q_I$ and the uniform distribution $u_{\text{local}}$. Using Pinsker's inequality ($2 \text{TV}(P, Q)^2 \le D(P \| Q)$), we bound this term:
    $$ \sum_{k=1}^M \left| q_I(k) - \frac{1}{M} \right| = 2 \text{TV}(q_I, u_{\text{local}}) \le \sqrt{2 D(q_I \| u_{\text{local}})}. $$
    Therefore, the contribution of interval $I$ to the total deviation is bounded by:
    $$ \md_M(A_I) \le {\|A\|_1} \cdot {\sigma}(I) \sqrt{2 D(q_I \| u_{\text{local}})}. $$
    
    Let $C_s$ be the sum of deviations on a scale $s$. Summing over all intervals $I$ on scale $s$:
    $$ C_s = \sum_{I \in s} \md_M(A_I) \le {\|A\|_1} \sum_{I} {\sigma}(I) \sqrt{2 D(q_I \| u_{\text{local}})}. $$
    Since $\sum_{I} {\sigma}(I) = 1$, we can apply Jensen's inequality (using the concavity of the square root function) to move the summation inside the square root:
    $$ C_s \le {\|A\|_1} \sqrt{\sum_{I} {\sigma}(I) \cdot 2 D(q_I \| u_{\text{local}})}. $$
    Let $G_s := \sum_{I \in s} {\sigma}(I) D(q_I \| u_{\text{local}})$ be the expected information gain on scale $s$. Then $C_s \le {\|A\|_1} \sqrt{2 G_s}$.
    
    Finally, we sum over all scales $s \in [S]$ to bound $\treemd_{M, S}(A)$. Using the Cauchy-Schwarz inequality:
    $$ \treemd_{M, S}(A) = \sum_{s=1}^S C_s \le {\|A\|_1} \sum_{s=1}^S \sqrt{2 G_s} \le {\|A\|_1} \sqrt{S \sum_{s=1}^S 2 G_s}. $$
    From \cref{eq:chain_rule}, we know $\sum G_s = D(P \| U)$. Thus:
    $$ \treemd_{M, S}(A) \le {\|A\|_1} \sqrt{2S \cdot D(P \| U)}. $$
    
    It remains to bound $D(P \| U)$. We maximize $\sum p_x \ln(n p_x)$ subject to $0 \le p_x \le \|A\|_\infty / {\|A\|_1}$ and $\sum p_x = 1$. The divergence is maximized when the distribution is most concentrated (minimum entropy). Since $p_x \le \|A\|_\infty / \|A\|_1$ for all $x$, we have:
    $$ D(P \| U) = \sum_{x} p_x \ln(n p_x) \le \sum_{x} p_x \ln\left(n \frac{\|A\|_\infty}{\|A\|_1}\right) = \ln\left(\frac{n \|A\|_\infty}{\|A\|_1}\right). $$
    Plugging this back into the tree deviation bound yields the lemma.
\end{proof}

\section{A Counter-Example to a Ratio-Free Version of \cref{lemma:treerd}} \label{appendix:treerdold}
We provide an example showing that the $\alpha$-dependence of \cref{lemma:treerd} is necessary, even for $M = 2$: The authors are grateful to an anonymous reviewer for this construction.
\begin{lemma}
    We can construct two sequences $A = (a_1, a_2, \ldots, a_n)$ and $B = (b_1, b_2, \ldots, b_n)$ of non-negative real numbers where $n = 2 ^ S$, such that
    \begin{equation}
    \treerd_{2, S}(A, B) = \Omega\left(S\sum_x{b_x}\right).
    \end{equation}
\end{lemma}
\begin{proof}
    We let $B = \{1, 1, \ldots, 1\}$ simply be the length-$n$ sequence of $1$'s. We let $A = \{1,2,2,4,2,4,4,8, \ldots\}$ be the sequence where
    \[
    a_i = 2 ^ {\text{number of 1's in the binary representation of $(i - 1)$}}.
    \]

    Notice that for every recursion interval $(l, r]$, we have
    $\frac{A_{\Sigma(l, (l+r)/2]}}{A_{\Sigma((l+r)/2, r]}} = \frac{1}{2}$. Therefore $\frac{\md_2(A(l, r])}{A_{\Sigma(l, r]}} = 1/3$. Since $B_{\Sigma(l, r]} = r - l$, we have $\rd_2(A(l, r], B(l, r]) = \frac{r - l}{3}$. Thus
    \[\treerd_{2, S}(A, B) = \sum_{s \in [S]}\sum_{(l,r] \in s}{\rd_M(A(l, r], B(l, r])} = \sum_{s \in [S]}{n / 3} = nS / 3 = S\sum_x{b_x} / 3  = \Omega\left(S\sum_x{b_x}\right).\]
\end{proof}

\end{document}